\documentclass[letterpaper,12pt]{article}

\usepackage{parskip}
\usepackage{graphicx}                
\usepackage{float}                   
\usepackage[hypertexnames=false]{hyperref}
\usepackage[margin=1in]{geometry}  
\usepackage[short]{datetime}         
\usepackage[labelfont=bf]{caption}   
\usepackage[bottom]{footmisc}        
\usepackage{lineno}
\usepackage{ragged2e}
\usepackage[protrusion=true,expansion=true]{microtype} 
\usepackage[T1]{fontenc}                               
\usepackage{lmodern}
\usepackage{tikz}
\usepackage{todonotes}

\usepackage[compact]{titlesec}
\titleformat*{\section}{\large\bfseries}
\titlespacing*{\section}{0pt}{2ex}{2ex}
\titleformat*{\subsection}{\normalsize\bfseries}
\titlespacing*{\subsection}{0pt}{1ex}{1ex}
\titleformat*{\subsubsection}{\normalsize\bfseries}
\titlespacing*{\subsubsection}{0pt}{1ex}{1ex}

\usepackage{amsfonts,amssymb,amsmath, mathtools}
\usepackage{xcolor, color, soul}
\usepackage{graphicx}
\usepackage{multicol}
\usepackage{url}
\usepackage{faktor}
\usepackage{ulem}
\usepackage{pdfpages}
\usepackage{soul}
\usepackage{bm}
\usepackage{natbib}
\usepackage{subcaption}
\usepackage{booktabs}
\usepackage{setspace}
\usepackage{amsthm}
\usepackage{algorithm} 
\usepackage{algpseudocode}
\usepackage{setspace}
\usepackage{bibunits}
\defaultbibliographystyle{apalike} 
\defaultbibliography{refs}

\hypersetup{
  colorlinks=true,
  linkcolor={blue},
  citecolor={blue},
  urlcolor={blue}}

\newtheorem{theorem}{Theorem}
\newtheorem{corollary}{Corollary}[theorem]
\newtheorem{lemma}[theorem]{Lemma}
\newtheorem{proposition}[theorem]{Proposition}

\newtheorem*{assumption*}{\assumptionnumber}
\providecommand{\assumptionnumber}{}
\makeatletter
\newenvironment{assumption}[2]
 {%
  \renewcommand{\assumptionnumber}{Assumption #1: {#2}}%
  \begin{assumption*}%
  \protected@edef\@currentlabel{#1}%
 }
 {%
  \end{assumption*}
 }
\makeatother

\def\indep{{\,\perp \!\!\! \perp\,}}

\newcommand{\Var}[0]{\text{Var}}
\newcommand{\Cov}[0]{\text{Cov}}

\begin{document}

\center{\textbf{\Large Partially Retargeted Balancing Weights for Causal Effect Estimation Under Positivity Violations}}
\center{ Martha Barnard$^{*}$, Jared D. Huling, Julian Wolfson}
\center{\footnotesize Division of Biostatistics and Health Data Science, School of Public Health, University of Minnesota}
\vspace{-.1in}
\justifying

\begin{abstract}
Positivity violations, which occur when some subgroups either always or never receive a treatment of interest, pose significant challenges for causal effect estimation with observational data. Recent balancing weight methods have proved to be highly effective in confounding control, however their utility is diminished in the presence of positivity violations, resulting in bias and excess variance. Approaches that deal with positivity violations, on the other hand, work by targeting a modified estimand that may be misaligned with the original research question. To address these challenges, we propose a novel balancing weights approach, which mitigates positivity violations while attempting to retain the original estimand by a targeted relaxation of the balancing constraints. Our proposed weighted estimator is consistent for the original estimand when either 1) the implied propensity score model is correct; or 2) all treatment effect modifiers are balanced to the target population. When these conditions do not hold, our estimator is consistent for a slightly modified treatment effect estimand. Furthermore, our proposed weighted estimator has reduced asymptotic variance when positivity does not hold. We evaluate our approach through applications to synthetic data, an observational study, and when transporting a treatment effect from a randomized trial.
\end{abstract}

\noindent%
{\it Keywords:} causal inference, direct balancing, overlap, transportability, observational studies
\vfill
\noindent\rule{2in}{0.4pt} \\
{\small * barna126@umn.edu}

\newpage
\setstretch{1.2}
\begin{bibunit}
\section{Introduction}
Confounding is a key challenge when estimating the causal effect of a binary treatment with observational data. While many methods exist for controlling confounding by balancing covariates between treatment groups, these methods do not fully control confounding when there are positivity violations. Positivity violations occur when there exists at least one combination of covariate values which perfectly predicts treatment assignment, a phenomenon which is especially likely to occur with high-dimensional covariates \citep{damour_deconfounding_2021}. When there are positivity violations,  inverse probability weighting (IPW) estimators \citep{rosenbaum1983,hahn_role_1998,robins_marginal_2000,hirano2001,hirano_efficient_2003,imbens_nonparametric_2004} have substantial statistical bias and inflated variance. A proposed alternative to IPW, direct balancing weights, \citep{hainmueller_entropy_2012,zubizarreta_stable_2015,wang_minimal_2020} tend to yield reduced estimator error and variance compared to IPW methods through achieving exact covariate balance for a prespecified set of covariates. However, when there are positivity violations there is often no solution to the direct balancing optimization problem prohibiting their use altogether.

To mitigate the impacts of positivity violations on causal effect estimation, a variety of methods, such as propensity score trimming, matching weights, and overlap weights, have been developed that modify the estimand and corresponding target population
to achieve estimators with reduced variance \citep{crump_dealing_2009,li_weighting_2013,yang_asymptotic_2018,li2018,li2019}. In fact, overlap weights \citep{li2018} do not require positivity by targeting the so-called ``overlap population''; this results in an estimator with minimum asymptotic variance among weighted treatment effects under homoscedasticity. While these methods achieve reduced estimator variance, they target a population that is modified from the original population of interest such that the corresponding treatment effect estimate may be misaligned with the original research question. Yet, a key benefit of causal effect estimation with observational data is the ability to evaluate causal effects in a meaningful population of interest such that modifying the estimand is often either undesirable or challenging to interpret. Furthermore, these estimators are inconsistent for the original estimand when there is treatment effect heterogeneity. Modified treatment policy estimands result in improved estimator performance by defining an entirely different estimand arrived at by imagining modifications to the natural value of treatment \citep{kennedy_nonparametric_2019}. While the corresponding estimators perform well, these modified treatment policies do not always answer the original research question.

Thus, under positivity violations, existing methods either 1) target an estimand of interest but result in poor estimator performance or technical issues (e.g., direct balancing weights yielding no solution); or 2) alleviate these technical issues by targeting a modified estimand that may be misaligned with the original estimand of interest (e.g., overlap weights). To address these challenges within estimator statistical performance and interpretation, we propose a novel balancing approach whose optimization problem 
yields a solution under positivity violations through a targeted relaxation of the balancing constraints.
Our corresponding weighted estimator is consistent for the original estimand when either the implied propensity score model is correct or set of treatment effect modifiers is properly specified and smaller than the set of confounders. When these conditions do not hold, our weighted estimator is consistent for a slightly modified estimand and target population. Furthermore, we demonstrate that overlap weights are a special case of our proposed weights such that our method offers an intermediate between direct balancing and overlap weights. We derive asymptotic results which show that 1) our weights can be derived under a relaxed positivity assumption; and 2) our proposed estimator has reduced asymptotic variance when there are positivity violations. We propose a design- and model-based procedure for selecting the balancing constraints to relax and we evaluate the corresponding estimators through applications to synthetic data, an electronic health record (EHR) study, and when transporting a treatment effect from a randomized controlled trial (RCT) to a new population.

The remainder of the paper is organized as follows. In Section \ref{sec:method_back}, we introduce notation and methodological background on balancing weights. We introduce and discuss the implications of our proposed partially retargeted balancing weights in terms of positivity, estimator interpretation, and extensions to additional estimation problems in Section \ref{sec:our_method}. In Section \ref{sec:asymptotics}, we discuss the asymptotic properties of our proposed weighted estimator. We propose design- and model-based procedures for identifying which balancing constraints to relax in Section \ref{sec:implementation}. In Sections \ref{sec:sim} and \ref{sec:illustrative_app}, we apply our methods to synthetic data, an EHR study, and when transporting a treatment effect from a RCT. We discuss results and future work in Section \ref{sec:discussion}.

\section{Methodological background}
\label{sec:method_back}
\subsection{Notation, assumptions, and covariate balance}
\label{sec:3w_bal}
Consider an independent sample $\{(Z_i, Y_i, \bm{X}_i)\}_{i=1}^n$ of size $n$ from a population where $Z_i = z$, indicates belonging to the treatment ($z=1$) or control group ($z=0$), $Y_i$ is the outcome of interest, and $\bm{X}_i = (X_{i1}, \ldots X_{ip})$ is a $p$ length vector of pre-treatment covariates. Let $Y(z)$ be the potential outcome that would be observed if assigned to treatment group $z$ (\cite{neyman_application_1990, rubin_estimating_1974, rubin_bayesian_1978,hernan_causal_2024}). We assume the stable unit treatment value assumption (SUTVA) which implies that $Y_i = Y_i(Z_i)$ such that we only observe one potential outcome for each individual. We focus on the average treatment effect (ATE), $\tau = E[Y(1) - Y(0)]$ as the estimand of interest, although this work can be readily extended to other estimands of interest which we discuss in Section \ref{sec:related_estimation}. Let $e(\bm{x}) = \Pr(Z_i =1 | \bm{X}_i = \bm{x})$ be the propensity score. For the identification of $\tau$ the following assumptions are typically made: 
\begin{assumption}{1}{Strong ignorability}\label{strong-ignorability}
$\{Y(0), Y(1)\} \indep Z | \bm{X}$ and
\end{assumption}
\begin{assumption}{2}{Positivity}\label{positivity}
$0 < e(\bm{x}) < 1$ for all $\bm{x}$.
\end{assumption}
The focus of our work surrounds the positivity assumption, specifically, scenarios where this assumption does not hold. Let $\mu_z(\bm{x}) = E[Y(z)|\bm{X} = \bm{x}]$ and $\tau(\bm{x}) = \mu_1(\bm{x}) - \mu_0(\bm{x})$. Then, $\tau  = E\left\{\frac{\mu_1(\bm{X})Z}{e(\bm{X})} - \frac{\mu_0(\bm{X})(1-Z)}{1-e(\bm{X})}\right\} = E\left\{\frac{YZ}{e(\bm{X})} - \frac{Y(1-Z)}{1-e(\bm{X})}\right\}$ holds with the above assumptions, where the sample version of the last equation corresponds to the IPW estimator of the ATE. For a treatment effect estimator be consistent, the expectations of $\mu_0(\bm{x})$ and $\mu_1(\bm{x})$ need to be balanced between the treated and control populations and to the target population. IPW (i.e., $w_1(\bm{x}) = \{e(\bm{x})\}^{-1}$ and $w_0(\bm{x}) = \{1-e(\bm{x})\}^{-1}$) balance all functions of $\bm{x}$ in this way such that the following equations hold,
\begin{align}
    E\left\{w_1(\bm{X})Zb(\bm{X})\right\} = E\{b(\bm{X})\}  \:\: &\text{and} \:\: E\left\{w_0(\bm{X})(1-Z)b(\bm{X})\right\} = E\{b(\bm{X})\} \label{eq:bal_pop}\\
    E\left\{w_1(\bm{X})Zb(\bm{X})\right\} &= E\left\{w_0(\bm{X})(1-Z)b(\bm{X})\right\}
    \label{eq:bal3}
\end{align}
for any function $b(\bm{x})$, which we reference as three-way balance \citep{chan_globally_2016}. Thus, IPW satisfy three-way balance of both $\mu_0(\bm{x})$ and $\mu_1(\bm{x})$, regardless of their true forms. Specifically, Equation \eqref{eq:bal_pop} balances the weighted treated and control group expectations of $b(\bm{x})$ to the population expectation of $b(\bm{x})$ and Equation \eqref{eq:bal3} balances the weighted treated and control group expectations of $b(\bm{x})$ to each other. Thus, Equation \eqref{eq:bal_pop} and \eqref{eq:bal3} serve different purposes; while Equation \eqref{eq:bal_pop} implies Equation \eqref{eq:bal3}, Equation \eqref{eq:bal_pop} primarily serves to balance the treated and control groups to the estimand target population.
When only Equation \eqref{eq:bal3} holds, $b(\bm{x})$ may be balanced to modified target population that corresponds not to the ATE, but to some other causal estimand.

 However, when there is a lack of overlap, even approximate balance of the empirical counterparts of Equations \eqref{eq:bal_pop} and \eqref{eq:bal3} may not be achieved after IPW, resulting in estimators with substantial error in addition to inflated variance. To achieve balance when there is a lack overlap, some methods target a different estimand, $E[h(\bm{X})\tau(\bm{X})]$, which corresponds to a modified target population with improved overlap. 
 \cite{li2018} propose one such method, overlap weights ($h(\bm{x}) = e(\bm{x})\{1-e(\bm{x})\}$), which are designed to ensure that the weights exist even when positivity does not hold. By modifying the target population, overlap weights can more easily satisfy Equation \eqref{eq:bal_pop} yielding an estimator with reduced bias and variance with respect to the modified estimand, the ATO. However, the overlap weights estimator tends to be biased for the ATE when there is treatment effect heterogeneity. 
While the weights corresponding to these modified estimands tend to achieve improved finite-sample balance in scenarios with a lack of overlap, they rarely achieve the empirical versions of Equations \eqref{eq:bal_pop} and \eqref{eq:bal3} resulting in estimator error. In fact, even when positivity does hold, exact finite-sample balance is rarely achieved by IPW-style estimators. While exact balance of covariate sample means is achieved by overlap weights when the propensity score is estimated with a logistic regression \citep{li2018}, most IPW-style weights do not achieve exact balance of the sample means, especially under positivity violations. 
This has motivated the development of sample weights that explicitly achieve three-way finite sample balance.

\subsection{Review of direct balancing weights}
Direct balancing weights for the ATE achieve exact three-way balance for some pre-specified set of covariate functions; for the direct balancing weights estimator to be consistent for the ATE, $\mu_z(\bm{x})$ must be a linear combination of this set of covariate functions. Consider the case where we want to balance the covariate functions $b(\bm{x}) = \{b_j(\bm{x})\}_{j=1}^J$. Then, direct balancing weights for the ATE take the following form,
\begin{subequations}\label{eq:db}
\begin{align}
\{w^{bw}_i\}_{i=1}^n  = \text{argmin}_{w_i} \sum_{i=1}^n D(w_i) \:\:\: &\text{subject to} \label{eq:db0}\\
\frac1n\sum_{i=1}^n w_iI[Z_i = z]b_j(\bm{X}_i) -\frac1n\sum_{i=1}^n b_j(\bm{X}_i) &= 0 \:\:\: j = 1, \ldots, J,  \:\:\: z = 0,1, \label{eq:db1}
\end{align}
\end{subequations}

where $D(\cdot)$ is a convex measure of dispersion. The constraints $w_i \geq 0$ and $\sum_{i=1}^n w_iZ_i = \sum_{i=1}^n w_i(1-Z_i)$ are often added to ensure that the weights do not extrapolate the observed outcomes. The resulting direct balancing weights estimator is,
\begin{equation}
    \hat{\tau}_{w^{bw}} = \frac1n\sum_{i=1}^n w^{bw}_iZ_iY_i - \frac1n\sum_{i=1}^n w^{bw}_i(1-Z_i)Y_i. \label{eq:estimator}
\end{equation}

In this constrained optimization problem, constraint \eqref{eq:db1} ensures that the treated and control covariate functions are balanced to the sample population. These constraints also imply that the treated and control covariate functions are balanced to each other (i.e., $\frac1n\sum_{i=1}^n w_iZ_ib_j(\bm{X}_i) = \frac1n\sum_{i=1}^n w_i(1-Z_i)b_j(\bm{X}_i)$) such that three-way balance is satisfied by $\{w^{bw}_i\}_{i=1}^n$ for $b(\bm{x})$. Since balance between treated and control groups is implicitly satisfied, Problem \eqref{eq:db} is separable by treatment group such that two separate constrained optimization problems can be solved to derive $\{w^{bw}_i\}_{i=1}^n$ \citep{chan_globally_2016}; we discuss implications of this separation in Section \ref{sec:positivity_considerations}. When there is a lack of overlap between $\{b(\bm{x})\}_{z=1}$ and $\{b(\bm{x})\}_{z=0}$, often there is no solution to Problem \eqref{eq:db} (i.e., Equations  \eqref{eq:db0}-\eqref{eq:db1}) with the addition of the constraints $w_i \geq 0$ and $\sum_{i=1}^n w_iZ_i = \sum_{i=1}^n w_i(1-Z_i)$ such that weights which achieve exact three-way balance without extrapolating cannot be derived. In these scenarios, the exact balance constraint is relaxed and approximate balance is achieved by constraining the absolute difference between sample means to be less than or equal to some parameter $\delta$; weights derived through this process are called minimal weights \citep{wang_minimal_2020}.
While the minimal weights estimator does have reduced variance compared to IPW and direct balancing estimators, minimal weights achieve none of the three balancing equations such that the corresponding estimator is often biased due to covariate imbalance. Thus, there is a need for methods that 1) achieve exact finite-sample covariate balance between treated and control groups to reduce estimator error; and 2) modify the target population minimally to yield an estimator with reduced variance when there is a lack of overlap.



\section{A partially retargeted balancing approach}
\label{sec:our_method}
To address these challenges, we propose a novel constrained optimization problem for deriving weights, primarily for scenarios where there is no solution to direct balancing weights due to positivity violations. 
Consider the case where Assumption \ref{positivity} does not hold, but $0 < \text{Pr}\{Z_i=1|c(\bm{X}) = c(\bm{x})\} < 1$ for some $c(\bm{x}) \subset b(\bm{x})$ such that positivity holds when conditioning on this restricted subset. Let $g(\bm{x})$ be the remaining covariate functions in $b(\bm{x})$ such that $b(\bm{x}) = [c(\bm{x}) = \{c_k(\bm{x})\}_{k=1}^K, g(\bm{x}) = \{g_l(\bm{x})\}_{l=1}^L]$. The key idea of our method is to relax the balancing constraints by balancing all covariate functions, $b(\bm{x}) = \{c(\bm{x}), g(\bm{x})\}$, between treated and control groups, but only  $c(\bm{x})$ to the sample population, yielding the following set of balancing equations:
\begin{subequations}\label{eq:our_db}
\begin{align}
\{w^*_i\}_{i=1}^n  = \text{argmin}_{w_i} \sum_{i=1}^n D(w_i) \:\:\: &\text{subject to}\label{eq:proposed1} \\
    \frac1n\sum_{i=1}^n w_iI[Z_i = z]c_k(\bm{X}_i) - \frac1n\sum_{i=1}^n c_k(\bm{X}_i)&= 0, \:\:\: k = 1, \ldots K, \:\:\: z = 0,1\label{eq:db_b1} \\
    \sum_{i=1}^n w_iZ_ig_l(\bm{X}_i) - \sum_{i=1}^n w_i(1-Z_i)g_l(\bm{X}_i)&= 0 , \:\:\: l = 1, \ldots L. \label{eq:db_b3}
\end{align}
\end{subequations}

The primary difference between this constrained optimization problem and the direct balancing weights constrained optimization problem is Equation \eqref{eq:db_b3}; in this equation, the functions, $g(\bm{x}) = \{g_l(\bm{x})\}_{l=1}^L$ are only balanced between treated and control groups and not to the sample population. Thus, the derived weights satisfy exact three-way balance for functions $c(\bm{x}) = \{c_k(\bm{x})\}_{k=1}^K$ but only balance between treated and control groups for $g(\bm{x})$. In doing so, the derived weights do not balance $g(\bm{x})$ to a specific modified population (such as with overlap weights). However, the weighted treated and control group sample means of $g(\bm{x})$ are implicitly modified when minimizing weight dispersion subject to only the treated and control balance constraint. We call these weights, $\{w_i^*\}_{i=1}^n$, partially retargeted balancing weights (PRTBW). In Section \ref{sec:positivity_considerations}, we formally show that this approach is justified under a relaxed positivity assumption. In addition, the corresponding estimator to Problem \eqref{eq:our_db} (i.e., Equations \eqref{eq:proposed1} - \eqref{eq:db_b3}) is consistent for the ATE when $c(\bm{x})$ contains all treatment effect modifiers, which we discuss further in Sections \ref{sec:estimator_interp} and \ref{sec:asymptotics}.

\begin{figure}[h!]
    \centering
    \includegraphics[width=\linewidth]{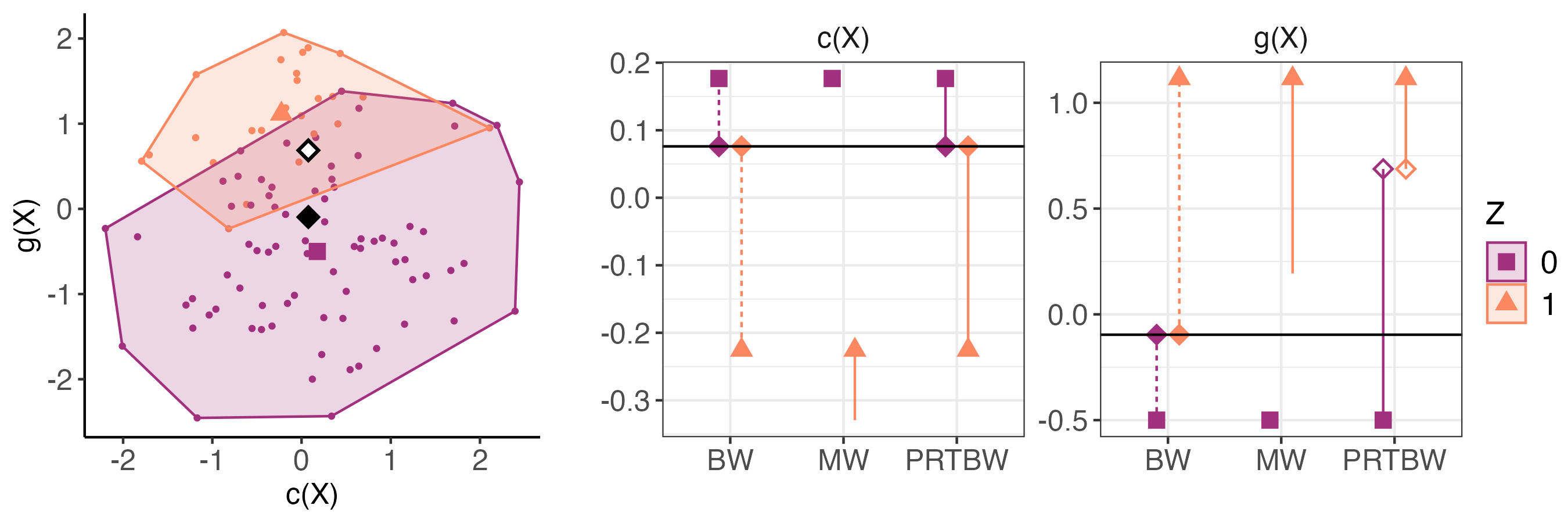}
    \caption{\footnotesize Simulated data illustrating the differences between direct balancing weights (BW), minimal weights (MW), and the proposed partially retargeted balancing weights (PRTBW). Squares (purple) and triangles (orange) indicate the unweighted control and treated sample means, respectively. The filled diamond indicates the sample population mean while the outlined diamond indicates the weighted sample mean given $\{w_i^*\}_{i=1}^n$ derived through Problem \eqref{eq:our_db}. In the right-hand figures, the horizontal line indicates the sample mean while the vertical lines end at the weighted treated and control means for each weighting method. For this simulated data, there is no solution to Problem \eqref{eq:db}, exhibited by the dashed vertical line.}
    \label{fig:method_fig}
\end{figure}

Figure \ref{fig:method_fig} presents a toy example of the differences in covariate balance and target population between balancing weights (Problem \eqref{eq:db}), minimal weights, and our proposed partially retargeted balancing weights (Problem \eqref{eq:our_db}); for this example, the constraints of the $w_i \geq 0$ and $\sum_{i=1}^n w_iZ_i = \sum_{i=1}^n w_i(1-Z_i)$ are added to the optimization problems. For this simulated data, there is no solution to Problem \eqref{eq:db}, such that weights that exactly balance the treated and control sample means of $b(\bm{x}) = \{c(\bm{x}), g(\bm{x})\}$ cannot be derived. Our proposed weights balance the treated and control sample means of $c(\bm{x})$ to the population sample mean, but the treated and control sample means of $g(\bm{x})$ are balanced to a different value. In contrast, minimal weights do not achieve covariate balance for either $c(\bm{x})$ or $g(\bm{x})$. In fact, the minimally weighted treated sample means are unchanged and the minimally weighted control sample mean of $c(\bm{x})$ is further away from the population sample mean of $c(\bm{x})$ than the unweighted control sample mean. Thus, while both the minimal weights optimization problem and our proposed Problem \eqref{eq:our_db} yield solutions for this data, they do this through different mechanisms such that each method yields a solution with a different interpretation.

To further explore the theoretical properties of our proposed weights, we identify the dual formulation of Problem \eqref{eq:our_db}. The addition of constraint \eqref{eq:db_b3} ensures that the optimization problem is not separable across the treatment groups such that there is correspondingly a single dual problem.

\begin{theorem}
    The dual of of Problem \eqref{eq:our_db} is equivalent to the 
    optimization problem
    \begin{equation}\label{eq:og_loss}
   \begin{split}
     \underset{\bm{\alpha}_0, \bm{\alpha}_1, \bm{\gamma}}{\textnormal{minimize}} \:\:\: &\frac1n \sum_{i=1}^n -Z_i\rho\{\bm{\alpha}_1^Tc(\bm{X}_i) + \bm{\gamma}^Tg(\bm{X}_i)\} - (1-Z_i)\rho\{\bm{\alpha}_0^Tc(\bm{X}_i) -\bm{\gamma}^Tg(\bm{X}_i)\} \\
     &+ \bm{\alpha}_1^Tc(\bm{X}_i)+\bm{\alpha}_0^Tc(\bm{X}_i)   
    \end{split}
\end{equation}
where $\bm{\alpha}_{0_{K\times 1}},\bm{\alpha}_{1_{K\times 1}},$ and $\bm{\gamma}_{L \times 1}$ are the dual variables associated with the balancing constraints in Equations \eqref{eq:db_b1}-\eqref{eq:db_b3}, $\rho(t) = t- t(h')^{-1}(t) + h[(h')^{-1}(t)] - h(1)$ is strictly concave, and $h(t) = D(1-t)$. 
Furthermore, the primal solution $w_{i}^* = Z_iw^*_{1i} + (1-Z_i)w^*_{0i}$ satisfies $w^*_{1i} = \rho'\{\hat{\bm{\alpha}}_1^Tc(\bm{X}_i) + \hat{\bm{\gamma}}^Tg(\bm{X}_i)\} \:\:\text{and}\:\: w^*_{0i} = \rho'\{\hat{\bm{\alpha}}_0^Tc(\bm{X}_i) - \hat{\bm{\gamma}}^Tg(\bm{X}_i)\}$
where $\hat{\bm{\alpha}}_0, \hat{\bm{\alpha}}_1, \hat{\bm{\gamma}}$ are the solution to the dual problem.
\label{thm:dual}
\end{theorem}

The proof of this theorem is in supplementary Section \ref{sec:proofs_supp}. 
In the following sections, we discuss the implications of Problem \eqref{eq:our_db}
in terms of positivity, the interpretation of the resulting estimator, and related estimation problems while drawing comparisons to direct balancing weights, minimal weights, and overlap weights.

\subsection{Positivity considerations}
\label{sec:positivity_considerations}
Problem \eqref{eq:our_db} is primarily motivated by the challenges posed by lack of overlap on causal effect estimation, generally, and direct balancing weights, specifically. In the following section, we show that a solution to Problem \eqref{eq:our_db}, with the additional constraints of $w_i \geq 0$ and $\sum_{i=1}^n w_iZ_i = \sum_{i=1}^n w_i(1-Z_i)$, exists under a relaxed positivity assumption. 
Furthermore, we determine the propensity score model implied by Problem \eqref{eq:our_db} and describe how this model connects to the relaxed positivity assumption and other methods. 
While positivity is required for the existence of a solution to Problem \eqref{eq:db} with the additional constraints of $w_i \geq 0$ and $\sum_{i=1}^n w_iZ_i = \sum_{i=1}^n w_i(1-Z_i)$ \citep{zhao_entropy_2017}, by not balancing $g(\bm{x})$ to the target population, the weights derived through Problem \eqref{eq:our_db} exist under a relaxed positivity assumption:
\begin{assumption}{2$^*$}{Relaxed positivity}\label{relaxed-positivity}
$0 < \text{Pr}\{Z_i | c(\bm{X}_i) = c(\bm{x})\} < 1$ for all $\bm{x}$ and there exists an $\bm{x}^*$ such that $0 < \text{Pr}\{Z_i | c(\bm{X}_i) = c(\bm{x}^*), g(\bm{X}_i) = g(\bm{x}^*)\} < 1$.
\end{assumption}
The following proposition states that this assumption is sufficient for the existence of the weights in Problem \eqref{eq:our_db} with the additional constraints of $w_i \geq 0$ and $\sum_{i=1}^n w_iZ_i = \sum_{i=1}^n w_i(1-Z_i)$ as $n \to \infty$. 
\begin{proposition}
    Suppose Assumption \ref{relaxed-positivity} is satisfied and the expectation of $c(\bm{X})$ exist, then $P(w^* \textnormal{ exists}) \to 1$ as $n \to \infty$.
    \label{thm:relaxed_pos}
\end{proposition}
We leave the proof of this proposition to supplementary Section \ref{sec:proofs_supp}. When
considering a finite sample, the weights in Problem  \eqref{eq:db} exist if the convex hulls generated by $\{c(\bm{X}_i), g(\bm{X}_i)\}_{Z_i =1}$  and $\{c(\bm{X}_i), g(\bm{X}_i)\}_{Z_i =0}$ contain $\frac1n\sum_{i=1}^n \{c(\bm{X}_i), g(\bm{X}_i)\}$ \citep{zhao_entropy_2017}. In contrast, the weights in Problem \eqref{eq:our_db} exist if 1) the convex hulls of $\{c(\bm{X}_i)\}_{Z_i =1}$ and $\{c(\bm{X}_i)\}_{Z_i =0}$ contain $\frac1n\sum_{i=1}^n c(\bm{X}_i)$; and 2) $\{c(\bm{X}_i), g(\bm{X}_i)\}_{Z_i =1} \cap \{c(\bm{X}_i), g(\bm{X}_i)\}_{Z_i =0} \neq \emptyset$.
This is a substantial relaxation of the requirements for weight existence, especially for scenarios where one treatment group is much smaller than the other. Returning to Figure \ref{fig:method_fig}, the sample mean of $\{c(\bm{x}), g(\bm{x})\}$ is not in the intersection of the two convex hulls and thus there is no solution to Problem \eqref{eq:db} as previously discussed. However, there are many observations in the intersection of the two convex hulls. Furthermore, when examining $c(\bm{x})$, its sample mean is within the observed range for both treatment groups. Thus, there is a solution to our proposed Problem \eqref{eq:our_db}, and in fact the corresponding weighted sample mean given by $\{w_i^*\}_{i=1}^n$ \textit{is} within the intersection of the treated and control convex hulls.

The relaxed positivity assumption for Problem \ref{eq:our_db} is directly related to the propensity score model implied by Problem \ref{eq:our_db}. 
Taking the expectation of population version of Equation \eqref{eq:og_loss} with respect to $\bm{x}$ yields the following first order conditions, 
\begin{subequations}\label{eq:ps_conditions}
\begin{align}
    \{e(\bm{x})\}^{-1} &= \rho'\{\bm{\alpha}_1^Tc(\bm{x}) + \bm{\gamma}^Tg(\bm{x})\} = w_1^*, \label{eq:ps1}\\
    \{1-e(\bm{x})\}^{-1} &= \rho'\{\bm{\alpha}_0^Tc(\bm{x}) - \bm{\gamma}^Tg(\bm{x})\} = w_0^*, \:\:\: \text{and} \label{eq:ps2}\\
    e(\bm{x}) &= \frac{\rho'\{\bm{\alpha}_0^Tc(\bm{x}) - \bm{\gamma}^Tg(\bm{x})\}}{\rho'\{\bm{\alpha}_0^Tc(\bm{x}) - \bm{\gamma}^Tg(\bm{x})\} + \rho'\{\bm{\alpha}_1^Tc(\bm{x}) + \bm{\gamma}^Tg(\bm{x})\}} =\frac{w_0^*}{w_0^*+w_1^*}. \label{eq:ps3}
\end{align}   
\end{subequations}
Thus, the weights $w_0^*$ and $w_1^*$ correspond to IPW for a specific propensity score model, which we reference as the implied propensity score model of Problem \eqref{eq:our_db}. In this implied propensity score model, the dual variables $\bm{\alpha}_0, \bm{\alpha}_1, \bm{\gamma}$ can be considered as the model coefficients. Furthermore, $\rho'(\cdot)$ may be regarded as the link function of the generalized linear propensity score model, where $D(\cdot)$ determines this link function \citep{wang_minimal_2020}.   

Problem \eqref{eq:our_db} restricts the coefficients corresponding to $g(\bm{x})$ in the models for $\{e(\bm{x})\}^{-1}$ and $\{1-e(\bm{x})\}^{-1}$ to be opposites of each other. This 
ensures that $w^*_0 = \{1 - w_1^{*^{-1}}\}^{-1}$ such that only a single propensity score model is specified by Problem \eqref{eq:our_db}. However, there are substantial restrictions on the propensity score model that satisfies conditions \eqref{eq:ps1}-\eqref{eq:ps3}. Specifically, $\Var(Z| \bm{X}=\bm{x})$ is restricted to be a function of only $c(\bm{x})$ 
for specific dispersion measures (see supplementary Section \ref{sec:implied_ps_supp} for further discussion). Since this generally does not hold when $E[Z|\bm{X}=\bm{x}]$ is a function of both $c(\bm{x})$ and $g(\bm{x})$, this model will only hold when $g(\bm{x})$ takes a specific functional form. 
However, it is explicitly through this restriction on the propensity score model that there is a solution to Problem \eqref{eq:our_db} in scenarios with a lack of overlap, as overlap is improved through a \textit{reduction} in the ability to predict treatment assignment \citep{clivio_towards_2024-1}. 
Rather than attempt to specify a correct implied propensity score model, Problem \eqref{eq:our_db} imposes restrictions on the propensity score model to improve overlap and yield a solution to the optimization problem. 

\subsubsection{Connections to other methods}
\label{sec:ps_connect}
Problem \eqref{eq:our_db} corresponds to a single propensity score model, while Problem  \eqref{eq:db} implies two incompatible propensity score models due to the separability of the optimization problem by treatment group.
By implying a single propensity score model, the weights derived by our proposed method maintain the standard philosophical interpretation of IPW that direct balancing weights lack for the ATE. Furthermore, this allows for the doubly robust property (i.e., the estimator is consistent if either the implied propensity score model or outcome model is properly specified) and semiparametric efficiency of our method given the correct model specification (see Section \ref{sec:asymptotics} for further discussion of these properties). 
Minimal weights also correspond to two incompatible propensity score models but additionally place an L1 penalty on the dual variables (i.e., propensity score model coefficients). 
While the restrictions imposed on the implied propensity score model of Problem \eqref{eq:our_db} are comparably less intuitive, by restricting only $\Var(Z|\bm{X}=\bm{x})$ rather than both $E[Z|\bm{X}=\bm{x}]$ and $\Var(Z|\bm{X}=\bm{x})$ our derived weights are able to achieve exact finite sample balance. Furthermore, 
in restricting $\Var(Z|\bm{X}=\bm{x})$ our proposed method precisely restricts the aspect the propensity score model that primarily influences estimator variance. 

When 
$c(\bm{x}) = \emptyset$ and
$g(\bm{x}) = b(\bm{x})$ in Problem \eqref{eq:our_db}, Assumption \ref{relaxed-positivity} becomes: ``There exists an $\bm{x}^*$ such that $0 < \text{Pr}\{Z_i | c(\bm{X}_i) = c(\bm{x}^*), g(\bm{X}_i) = g(\bm{x}^*)\} < 1$'' which is comparable to the lack of a positivity assumption for overlap weights.  Thus, Assumption \ref{relaxed-positivity} is an intermediate  between Assumption $2$ and no positivity assumption; as more covariate functions are added to the $g(\bm{x})$ set, Assumption \ref{relaxed-positivity} becomes more relaxed. To further connect overlap weights and Problem \eqref{eq:our_db}, when $c(\bm{x}) = \emptyset$ and
$g(\bm{x}) = b(\bm{x})$ the first order conditions in \eqref{eq:ps_conditions} simplify to $e(\bm{x})/\{1-e(\bm{x})\} = w_0^*/w_1^*$.
While the exact solution to the dual problem will depend on the specific dispersion measure, 
overlap weights satisfy this condition. Therefore, Problem \eqref{eq:our_db} with $c(\bm{x}) = \emptyset$ and
$g(\bm{x}) = b(\bm{x})$ can be considered the direct balancing equivalent of overlap weights; however, Problem \eqref{eq:our_db} will always yield exact balance while overlap weights only yields  exact balance when the propensity score is estimated with logistic regression. 

\subsection{Interpretation of the proposed estimator}
\label{sec:estimator_interp}
While Problem \eqref{eq:our_db} yields a solution when there is a lack of overlap, in doing so the population of $g(\bm{x})$ is modified away from the ATE target population, potentially resulting in estimator bias with respect to the ATE (termed bias due to estimand mismatch in \cite{barnard_unified_2025}). 
Thus, we provide model- and design- based perspectives on the conditions under which our proposed estimator has minimal error with respect to the ATE. In addition, we provide a characterization of the modified population and estimand that $\hat{\tau}_{w^*}$ targets when $\hat{\tau}_{w^*}$ does not have minimal error with respect to the ATE. 
We first note that for any weighted estimator for the ATE,
\begin{subequations}\label{eq:bias+decomp}
\begin{align}
    \hspace{-0.2in} \hat{\tau}_w- \tau &= 
    \frac1n \sum_{i=1}^n \{w_{i}Z_i\mu_{Z_i}(\bm{X}_i) - \mu_{1}(\bm{X}_i)\} -  \frac1n \sum_{i=1}^n \{w_{i}(1-Z_i)\mu_{Z_i}(\bm{X}_i) - \mu_{0}(\bm{X}_i)\}  \label{eq:bias_term} \\
    &- \int \{\mu_1(\bm{x}) - \mu_0(\bm{x})\}d\{F - F_n\}(\bm{x}) \label{eq:zero_bias1} \\
    &+ \frac{1}{n}\sum_{i=1}^n w_{i}\epsilon_iZ_i - \frac{1}{n}\sum_{i=1}^n w_{i}\epsilon_i(1 - Z_i),
    \label{eq:zero_bias2}
\end{align}
\end{subequations}
where $\epsilon_i = Y_i(Z_i) - \mu_{Z_i}(\bm{X}_i)$ and $F_n(\bm{x}) = \sum_{i=1}^n I(\bm{X}_i \leq \bm{x})/n$ is the empirical CDF. In this expression, term \eqref{eq:zero_bias1} goes to zero for a representative sample of the population and term \eqref{eq:zero_bias2} has expectation zero. Thus, term \eqref{eq:bias_term} is the primary contributor to estimator error; we focus on this term for the remainder of this section. 


To better understand the conditions under which our balancing equations are justified, consider the case where $\mu_z(\bm{x})$ is a linear function of $c(\bm{x})$ and $g(\bm{x})$ such that $\mu_1(\bm{x}) = \bm{\beta}_1^Tc(\bm{x}) + \bm{\lambda}_1^Tg(\bm{x})$ and $\mu_0(\bm{x}) = \bm{\beta}_0^Tc(\bm{x}) + \bm{\lambda}_0^Tg(\bm{x})$ for some $\bm{\beta}_z \in \mathbb{R}^K$ and $\bm{\lambda}_z \in \mathbb{R}^L$. When weights are derived through Problem \eqref{eq:db}, term \eqref{eq:bias_term} will be zero for the corresponding estimator. However when weights are derived with Problem \eqref{eq:our_db}, term \eqref{eq:bias_term} simplifies to $\frac{1}{n}\sum_{i=1}^n \bm{\lambda}_1^Tg(\bm{X}_i) - \bm{\lambda}_0^Tg(\bm{X}_i)$
such that term \eqref{eq:bias_term}
is zero when $\bm{\lambda}_1 = \bm{\lambda}_0$. Thus, term \eqref{eq:bias_term} is zero when $\mu_z(\bm{x}) = \bm{\beta}_z^Tc(\bm{x}) + \bm{\lambda}^Tg(\bm{x})$ and $\tau(\bm{x}) = (\bm{\beta}_1^T - \bm{\beta}_0^T)c(\bm{x})$.
These models imply that our proposed estimator has minimal error with respect to the ATE when the $c(\bm{x})$ set contains all treatment effect modifiers. Thus, as the set $c(\bm{x})$ gets smaller and correspondingly the set $g(\bm{x})$ gets larger, our proposed estimator requires a stronger assumption about the number of effect modifiers to be justified for the ATE. However, this is a relaxation of the assumption $\tau(\bm{x}) = \tau$ required for the overlap weights estimator to be consistent for the ATE.

We can also explore term \eqref{eq:bias_term} for our proposed estimator from a design-based, rather than modeling, prospective; 
an alternative simplification of term \eqref{eq:bias_term} is $(\bm{\lambda}_1^T -  \bm{\lambda}_0^T)\left\{\frac{1}{n}\sum_{i=1}^n w_iI[Z_i = z] \right.$ $\left. \times g(\bm{X}_i) -\frac{1}{n}\sum_{i=1}^ng(\bm{X}_i)\right\}$.
From a design-based perspective, this term tends to be small when
$1^T\left\lvert\frac{1}{n}\sum_{i=1}^n w_i  I[Z_i = z]g(\bm{X}_i) -\frac{1}{n}\sum_{i=1}^ng(\bm{X}_i)\right\rvert$ is small. However, in Problem \eqref{eq:our_db} we explicitly do not constrain $\frac{1}{n}\sum_{i=1}^n w_iI[Z_i = z]g(\bm{X}_i) -\frac{1}{n}\sum_{i=1}^ng(\bm{X}_i) = 0$ in order to yield a solution when there is a lack of overlap. Thus, $1^T\left\lvert\frac{1}{n}\sum_{i=1}^n w_i  I[Z_i = z]g(\bm{X}_i) -\frac{1}{n}\sum_{i=1}^ng(\bm{X}_i)\right\rvert$ will be non-zero when weights are derived using Problem \eqref{eq:our_db}. Despite this, the error of our proposed estimator with respect to the ATE is due to differences between the weighted and sample population of $g(\bm{x})$, rather than the imbalance between the weighted treated and control populations. Thus, when $c(\bm{x})$ does not contain all treatment effect modifiers, our proposed estimator has minimal error with respect to the average treatment effect estimand over a modified population of $g(\bm{x})$, which can be characterized by $\frac1n\sum_{i=1}^n w^*_iZ_i g(\bm{X}_i) = \frac1n\sum_{i=1}^n w^*_i(1-Z_i) g(\bm{X}_i)$. While the ATO population can also be characterized through weighted sample means, it can be challenging to conceptualize this modified population especially with high-dimensional covariates. 
We provide further guidance and discussion on how to select $g(\bm{x})$ given these model- and design-based perspectives in Section \ref{sec:implementation} and supplementary Section \ref{sec:guidance_implement_supp}.

\subsection{Related estimation problems}
\label{sec:related_estimation}
Our proposed constrained optimization problem for deriving weights for ATE estimation can be readily extended to variety of additional estimation problems.
We discuss extensions to the ATT and transporting effects below. Extensions to effect estimation with more than two treatment groups, weighted average treatment effects, and distributional balancing weights are discussed in supplementary Section \ref{sec:connect_methods_supp}. 
\subsubsection{ATT}
When deriving either IPW or direct balancing weights for the ATT, only weights for the control group are derived such that the corresponding estimator is $\frac1n \sum_{i=1}^n Y_i - \frac1n \sum_{i=1}^n w_i(1-Z_i)Y_i$. However, when adapting Problem \eqref{eq:our_db} for ATT estimation, weights must be derived for both the treated and control group. To derive these weights, $\{w_i^{*ATT}\}$, we can simply replace the right-hand side of constraint \eqref{eq:db_b1}  with $\frac1n\sum_{i=1}^n Z_ic_k(\bm{X}_i)$ and compute Equation \eqref{eq:estimator} with these weights. 
The design-based and model-based interpretations of the sets $c(\bm{x})$ and $g(\bm{x})$ remain the same as described in Section \ref{sec:estimator_interp}. However, the covariates that are effect modifiers may differ between the entire sample population and the treated population.

\subsubsection{Transporting effects}
\label{sec:transportability}
While clinical trials yield ``gold standard'' causal effect estimates through randomization, trial inclusion/exclusion criteria is often restrictive such that the trial population may not reflect a clinically relevant population. In these cases, it is desirable to transport the effect to a more meaningful population. However, there is often a lack of overlap between the trial population and this more meaningful target population. 
Our procedure allows the analyst to improve overlap by modifying the target population for only a subset of covariates, rather than modifying all covariates or discarding observations, which is inconsistent with the goal of estimating the effect in a clinically relevant population. To adapt Problem \ref{eq:our_db} for transportability, we first introduce some additional notation. Let $R_i = r$ be an indicator of whether an observation belongs to the trial population ($r = 1$) or the target population ($r=0$). We only observe $Y_i$ and $Z_i$ when $R_i = 1$ but observe $\bm{X}_i$ for both populations. Let $n_{r}$ indicate the sample size in population $r$. Given additional assumptions outlined in \cite{dahabreh_extending_2020}, the transported effect estimand is identifiable with observed data. We propose the following modification of Problem \eqref{eq:our_db} for transportability,
\begin{subequations}\label{eq:transport_db}
\begin{align}
   \{w^{*t}_i\}_{i=1}^n  = &\text{argmin}_{w_i} \sum_{i=1}^n D(w_i) \:\:\: \text{subject to}\\
    \frac1{n_1}\sum_{i=1}^n w_iI[Z_i = z]R_ic_k(\bm{X}_i) &= \frac1{n_0}\sum_{i=1}^n (1-R_i)c_k(\bm{X}_i), \:\:\: k = 1, \ldots K, \:\:\: z = 0,1 \\
    \frac1{n_1}\sum_{i=1}^n w_iZ_iR_ig_l(\bm{X}_i) &=   \frac1{n_1}\sum_{i=1}^n w_i(1-Z_i)R_ig_l(\bm{X}_i) , \:\:\: l = 1, \ldots L  
\end{align}
\end{subequations}
Then, our proposed estimator for the transported effect is $\tau_{w^{*t}} = \frac1n \sum_{i=1}^n w^{*t}_iZ_iR_iY_i - \frac1n \sum_{i=1}^n w^{*t}_i(1-Z_i)R_iY_i$. The interpretation of $c(\bm{x})$ and $g(\bm{x})$ remains the same as described in Section \ref{sec:estimator_interp}, however the selection of $g(\bm{x})$ may be of even more importance in order to maintain a target population of interest.

\section{Asymptotic properties of the proposed estimator}
\label{sec:asymptotics}
We derive asymptotic results for the solution to Equation \eqref{eq:og_loss} which we then use to develop asymptotic results for our proposed estimator, $\hat{\tau}_{w^*}$, under various conditions. All proofs of results in this section are in supplementary Section \ref{sec:proofs_supp}. We first require additional regularity conditions, similar to those in \cite{hirano_efficient_2003, chan_globally_2016} and \cite{kallberg_large_2023}, on the set balancing functions, $b(\bm{x})$.
\begin{assumption}{3}{Balance function conditions}\label{balanc_func_regularity}
There exists a finite $M$, such that $|b_j(\bm{X})| < M$ for $j = 1, \ldots, J$ almost surely and $B(\bm{X}) = [b_1(\bm{X}), \ldots,b_J(\bm{X})] $ has full rank.
\end{assumption}


For the remainder of this section, we notate $w_i^*$ as $ w^*(\bm{X}_i)$ to explicitly show the relationship to $\bm{X}_i$. Given Assumptions \ref{strong-ignorability}, \ref{relaxed-positivity}, and \ref{balanc_func_regularity}, $w_{1}^*(\bm{X}_i)$ and $w_{0}^*(\bm{X}_i)$ are consistent for their asymptotic counterparts, $\tilde{w}_1(\bm{X}_i)$ and $\tilde{w}_0(\bm{X}_i)$, which uniquely solve the population version of Equation \eqref{eq:og_loss} (supplementary Section \ref{sec:proofs_supp}, Proposition \ref{prop:supp}); a majority of the asymptotic results follow from this fact.
Since Problem \eqref{eq:our_db} implies a single propensity score model, our proposed estimator is doubly robust as stated below in Theorem \ref{theorem:db}.

\begin{theorem}
    Suppose Assumptions \ref{strong-ignorability} and \ref{balanc_func_regularity} hold and $\rho''\{\cdot\}$ is continuous over $\mathbb{R}$. Assume that the expectation of $c(\bm{X})$ exists and that $\Var\{Y(a)\} < \infty$ for $a = 0,1$. Then, $\hat{\tau}_{w^*}$ is doubly robust in the sense that if either 
    1) Assumption \ref{positivity} holds and $e(\bm{x})$ satisfies conditions \eqref{eq:ps1}-\eqref{eq:ps3} or; 2) Assumption \ref{relaxed-positivity} holds, $\mu_z(\bm{x})$ is linear in $c(\bm{x})$ and $g(\bm{x})$, and $\tau(\bm{x})$ is linear in $c(\bm{x})$, then $\hat{\tau}_{w^*}$ is consistent for the ATE. \label{theorem:db}
\end{theorem}

Note that $\rho''\{\cdot\}$ is continuous for common dispersion measures such as $w^2$ and $w\log(w)$. In contrast, $\hat{\tau}_{w^{bw}}$ is not doubly robust \citep{chan_globally_2016}. Since our estimator is motivated by scenarios where positivity does not hold, the second condition that guarantees consistency (i.e., $\tau(\bm{x})$ is linear in only $c(\bm{x})$) is of primary importance. 
Let $\tau_{g, \tilde{w}}(\bm{x}, z) = (\bm{\beta}_1 - \bm{\beta}_0)^Tc(\bm{x})  + \bm{\lambda}_1^T\tilde{w}_1(\bm{x})zg(\bm{x}) - \bm{\lambda}_0^T\tilde{w}_0(\bm{x})(1-z)g(\bm{x})$ for some $\bm{\beta}_z \in \mathbb{R}^K$ and $\bm{\lambda}_z \in \mathbb{R}^L$. When $\tau(\bm{x})$ is linear in both $c(\bm{x})$ and $g(\bm{x})$, our proposed estimator is consistent for the modified estimand $\tau_{g, \tilde{w}} = E[E[Y|Z = 1, c(\bm{X}), \tilde{w}_1(\bm{X})e(\bm{X})g(\bm{X})]] - E[E[Y|Z = 0, c(\bm{X}), \tilde{w}_0(\bm{X})\{1-e(\bm{X})\}g(\bm{X})]]= E[\tau_{g, \tilde{w}}(\bm{X}, Z)]$ as stated in Corollary \ref{coro:consistency}. Note that $\tau_{g, \tilde{w}}$ is indeed a causal estimand because $E[\tilde{w}_1(\bm{X})e(\bm{X})g(\bm{X})] = E[\tilde{w}_0(\bm{X})\{1-e(\bm{X})\}g(\bm{X})]$ is a first order condition of the population version of Equation \eqref{eq:og_loss}. 
\begin{corollary}
    Suppose Assumptions \ref{strong-ignorability} and \ref{relaxed-positivity} hold and $\rho''\{\cdot\}$ is continuous over $\mathbb{R}$. Assume that the expectation of $c(\bm{X})$ exists and that $\Var\{Y(a)\} < \infty$ for $a = 0,1$. Then, $\hat{\tau}_{w^*}$ is consistent for $\tau_{g,\tilde{w}}$ if $e(\bm{x})$ does not satisfy conditions \eqref{eq:ps1}-\eqref{eq:ps3} and $\mu_z(\bm{x})$ and $\tau(\bm{x})$ are both linear in $c(\bm{x})$ and $g(\bm{x})$.
    \label{coro:consistency}
\end{corollary}
When positivity does hold and all models implied by Problem \ref{eq:our_db} are correctly specified, our proposed estimator is semiparametrically efficient as stated in Theorem \ref{theorem:semipar_eff}.

\begin{theorem}
    Suppose Assumptions \ref{strong-ignorability}, \ref{positivity}, and \ref{balanc_func_regularity} hold and $\Var\{Y(a)|\bm{X} = \bm{x}\} < \infty$ for all $\bm{x}$ and $a = 0,1$. Suppose $e(\bm{x})$ satisfies conditions \eqref{eq:ps1}-\eqref{eq:ps3}, $\mu_z(\bm{x})$ is linear in $c(\bm{x})$ and $g(\bm{x})$, and $\tau(\bm{x})$ is linear in $c(\bm{x})$. Then,
    \begin{equation*}
        \sqrt{n}(\hat{\tau}_{w^*} - \tau) \to N(0, V_{opt})
    \end{equation*}
    where $V_{opt} = E\left[\frac{\Var(Y(1)|\bm{X})}{e(\bm{X})} + \frac{\Var(Y(0)|\bm{X})}{1-e(\bm{X})} + \{\tau(\bm{X}) - \tau\}^2 \right]$ is the semiparametric efficiency bound \citep{hahn_role_1998}.
    \label{theorem:semipar_eff}
\end{theorem}

However, it is rare that the implied propensity score model is correctly specified and our proposed estimator is motivated by scenarios where Assumption \ref{positivity} does not hold. Thus, we derive and present the asymptotic variance of our estimator when only the outcome and models are correctly specified in Corollary \ref{theorem:new_var}. 

\begin{corollary}
    Suppose Assumptions \ref{strong-ignorability}, \ref{relaxed-positivity}, and \ref{balanc_func_regularity} hold and $\Var\{Y(a)|\bm{X} = \bm{x}\} < \infty$ for all $\bm{x}$ and $a = 0,1$. Suppose $e(\bm{x})$ does not satisfy conditions \eqref{eq:ps1}-\eqref{eq:ps3}, $\mu_z(\bm{x})$ is linear in $c(\bm{x})$ and $g(\bm{x})$, and $\tau(\bm{x})$ is linear in $c(\bm{x})$. Then,
    \begin{equation*}
        \sqrt{n}(\hat{\tau}_{w^*}  - \tau) \to N(0, V_{opt}^*)
    \end{equation*}
    where $V_{opt}^* = E[\tilde{w}_1(\bm{X})^2e(\bm{X})\Var\{Y(1)|\bm{X}\} + \tilde{w}_0(\bm{X})^2\{1-e(\bm{X})\}\Var\{Y(0)|\bm{X}\} + \{\tau(\bm{X}) - \tau\}^2]$.
    \label{theorem:new_var}
\end{corollary}

Of key importance, this variance result only requires the relaxed positivity assumption, rather than the standard positivity assumption. However, to compare $V_{opt}$ and $V^*_{opt}$, consider the case where positivity does hold. Since the asymptotic weights tend to be less extreme (i.e., $\tilde{w}_1(\bm{X}_i) \leq \{e(\bm{X}_i)\}^{-1}$ and  $\tilde{w}_0(\bm{X}_i) \leq \{1-e(\bm{X}_i)\}^{-1}$) than IPW through the restriction on the implied propensity score model, $V^*_{opt} \leq V_{opt}$. Thus, our proposed estimator does indeed have reduced asymptotic variance, especially in scenarios where $e(\bm{x})$ is close to zero or positivity does not hold. 
Specifically, as the $g(\bm{x})$ set gets larger, $\tilde{w}_z(\bm{x})$ become less extreme and correspondingly $V^*_{opt}$ becomes smaller. We further derive the asymptotic variance of our proposed estimator with respect to the modified estimand, $\tau_{g, \tilde{w}}$, when $\tau(\bm{x})$ is a linear function of both $c(\bm{x})$ and $g(\bm{x})$ in Corollary \ref{theorem_g_new_var}.

\begin{corollary}
    Suppose Assumptions \ref{strong-ignorability}, \ref{relaxed-positivity}, and \ref{balanc_func_regularity} hold and $\Var\{Y(a)|\bm{X} = \bm{x}\} < \infty$ for all $\bm{x}$ and $a = 0,1$. Assume that $\rho''\{\cdot\}$ and $\rho'''\{\cdot\}$ are continuous over $\mathbb{R}$. Suppose $e(\bm{x})$ does not satisfy conditions \eqref{eq:ps1}-\eqref{eq:ps3} and $\mu_z(\bm{x})$ and $\tau(\bm{x})$ are linear in both $c(\bm{x})$ and $g(\bm{x})$. Then,
    \begin{align*}
        \sqrt{n}(\hat{\tau}_{w^*}  - \tau_{g,\tilde{w}}) \to N(0, V_{g, \tilde{w}})
    \end{align*}
    where $V_{g, \tilde{w}} = E[\tilde{w}_1(\bm{X})^2e(\bm{X})\Var\{Y(1)|\bm{X}\} + \tilde{w}_0(\bm{X})^2\{1-e(\bm{X})\}\Var\{Y(0)|\bm{X}\} + \{\tau_{g, \tilde{w}}(\bm{X}, Z) - \tau_{g, \tilde{w}} + r_1(\bm{X}, Z) - r_0(\bm{X}, Z)\}^2].$
    \label{theorem_g_new_var}
\end{corollary}
The functions $r_z(\bm{x}, z)$ account for the difference between the estimated and asymptotic weights. To avoid introducing substantial additional notation, we present the exact forms of these functions in the proof of Corollary \ref{theorem_g_new_var} in supplementary Section \ref{sec:proofs_supp}.  Note that the additional assumptions on $\rho\{\cdot\}$ hold for common dispersion measures such as $w^2$ and $w\log(w)$. To compare to Corollary \ref{theorem:new_var}, consider the case where the assumptions of Corollary \ref{theorem_g_new_var} hold and the direct balancing weights estimator is computed. In this case, the asymptotic variance of the direct balancing estimator takes the same form as $V_{opt}^*$ but with the asymptotic counterparts of direct balancing weights; we will reference this variance as $V_{opt}^{bw}$. Our proposed asymptotic weights are less extreme than the asymptotic direct balancing weights, however, $\Var\{\tau_{g, \tilde{w}}(\bm{X}, Z)\} \gtrsim  \Var\{\tau(\bm{X})\}$ as $\Var\{w_z(\bm{X})I[Z=z]g(\bm{X})\} \gtrsim \Var\{g(\bm{X})\}$. Thus, there is an unclear relationship between $V_{g, \tilde{w}}$ and $V^{bw}_{opt}$; in some scenarios $V_{g, \tilde{w}} > V^{bw}_{opt}$ while in others $V_{g, \tilde{w}} < V^{bw}_{opt}$. To further explore this relationship, we perform a simulation study to compare the asymptotic variance of these estimators (supplementary Section \ref{sec:supp_sim_var}). When overlap is poor and there is low to moderate treatment effect heterogeneity,  $V_{g, \tilde{w}}$ tends to be smaller than $V^{bw}_{opt}$ (supplementary Figure \ref{fig:asymp_var}). 
With the addition of the constraints $w_i \geq 0$ and $\sum_{i=1}^n w_iZ_i = \sum_{i=1}^n w_i(1-Z_i)$ to both optimization problems, the direct balancing weights are even more extreme such that $\Var(\hat{\tau}_{w^*}) \lesssim \Var(\hat{\tau}_{w^{wb}})$ across a majority of data generating scenarios (supplementary Figure \ref{fig:emp_sd_var}); in fact, when overlap is poor $\Var(\hat{\tau}_{w^*}) << \Var(\hat{\tau}_{w^{wb}})$.

\section{Implementation details}
\label{sec:implementation}
To select the optimal $g(\bm{x})$ set, in terms of error with respect to the ATE, one approach is to examine all possible covariate function sets and select the set that 1) yields a solution to Problem \eqref{eq:our_db}; and 2) minimizes term \eqref{eq:bias_term}. However, examining all possible covariate subsets is computationally prohibitive, especially as covariate dimension increases. Therefore, we propose a greedy algorithm (Algorithm \ref{alg:adaptive}) for identifying an approximately optimal $g(\bm{x})$ subset according to a specific metric 
aligned with either the model- or design-based perspective on estimator error. We also propose an additional algorithm (supplementary Algorithm \ref{alg:simple}) that is less computationally intensive than Algorithm \ref{alg:adaptive}; we found these algorithms resulted in estimators with similar performance  when there was no correlation between covariates (supplementary Figure \ref{fig:mse_alg_adapt}). However, we generally recommend the use of Algorithm \ref{alg:adaptive} if it is not computationally prohibitive.

\begin{algorithm}
	\caption{Identifying an approximately optimal $g(\bm{X})$ set given metric $m$} 
    \label{alg:adaptive}
     \hspace*{\algorithmicindent}
	\begin{algorithmic}[1]
            \State \textbf{Inputs:} The matrix $\mathbf{X}_{n\times p}$, set of covariate functions $b_{1}(\bm{X}), \ldots, b_{J}(\bm{X})$, and a metric $m$ that can be computed for each $b_j(\bm{X})$
            \State \textbf{Initialize:} $b(\bm{X})^{-} = \emptyset $
		\While {$\{w^*_i\}_{i=1}^n$ given $g(\bm{X}) = b(\bm{X})^{-}$ does not exist}
                \State Compute $m$ for $\{b_{1}(\bm{X}), \ldots, b_{j}(\bm{X})\} \setminus b(\bm{X})^{-}\}$ and identify $b^*(\bm{X})$, the covariate with the minimum or maximum value of the metric $m$
                \State Let $b(\bm{X})^{-} = \{b(\bm{X})^{-},  b^*(\bm{X}) \}$
			\State Solve Problem \eqref{eq:our_db} with $g(\bm{X}) = b(\bm{X})^{-}$ and $c(\bm{X}) = \{b_{1}(\bm{X}), \ldots, b_{j}\} \setminus b(\bm{X})^{-}$
		\EndWhile
            \State \textbf{Outputs:} 
            \State \hspace{1.5em} $b(\bm{X})^{-}$ \Comment{Smallest set that yields a solution to Problem \eqref{eq:our_db} when $g(\bm{X}) = b(\bm{X})^{-}$} 
            \State \hspace{1.5em} $\{w^*_i\}_{i=1}^n$ \Comment{Resulting weights from Problem \eqref{eq:our_db} with $g(\bm{X}) = b(\bm{X})^{-}$}
	\end{algorithmic} 
\end{algorithm}

From a design-based perspective, selecting the covariate functions that are each conditionally the most predictive of treatment will tend to result in the smallest $g(\bm{x})$ set that yields a solution to Problem \eqref{eq:our_db}. While selecting this set may not minimize the design-based simplification of estimator error, this set does tend to result in smaller design-based estimator error and simpler estimator interpretation. Thus, for the design-based metric, we propose and implement Spearman semipartial correlation; 
in step four, the maximum correlation is used to identify $b^*(\bm{x})$. We choose this metric because it is rank-based and reflects the impact of removing a covariate in a hypothetical propensity score model, however other metrics could work well here. Our design-based estimator then takes the same form as Equation \eqref{eq:estimator} with weights computed with the $g(\bm{x})$ set selected by Algorithm \ref{alg:adaptive}. 

From a model-based perspective, selecting the covariate functions that do not modify the treatment effect as $g(\bm{x})$ will tend to minimize estimator error. Thus, we implement the absolute value of the one-step doubly robust estimator of treatment effect modification,
proposed by \cite{boileau_nonparametric_2025}, as the model-based metric to use within Algorithm \ref{alg:adaptive}; 
 in step four, the minimum of this metric is used to identify $b^*(\bm{x})$. Our model-based estimator is a cross-fit style estimator to preserve downstream inference; we provide further discussion of this model-based metric and estimator in supplementary Section \ref{sec:extra_mb_estimator_supp}. For either metric/estimator, Algorithm \ref{alg:adaptive} can be initiated with a non-empty $b(\bm{x})^-$ set (step 2) if there are covariates that one wants to include in the $g(\bm{x})$ set \textit{a priori}. If minimizing error with respect to the ATE is not a priority, $g(\bm{x})$ could be chosen such that the modified $g(\bm{x})$ population is of scientific interest which we discuss further in supplementary Section \ref{sec:guidance_implement_supp}.


\section{Simulation experiments}
\label{sec:sim}
We generate a variety of simulation scenarios that vary in the number of confounders (i.e., covariates), level of overlap, level of treatment effect heterogeneity, and proportion of covariates that are effect modifiers. Through these scenarios, we first aim to validate the estimator corresponding to the weights derived through Problem \eqref{eq:our_db} when the true $\tau(\bm{x})$ function is known. In addition, we seek to determine in which data generating scenarios the proposed methods for a completely data-driven selection of $g(\bm{x})$ in Section \ref{sec:implementation} yield estimators with improved performance when compared to IPW, minimal weights, and overlap weights estimators when there is no solution to Problem \eqref{eq:db}. 

\subsection{Data generation}
For clarity in presentation, we drop the observation index on the variables. For $p = 20$ and $p=100$ we simulate $V_1$ - $V_p$ from a truncated normal distribution where $E[V_i] = 0$, $\Var(V_i) = 1$ and $\text{Cov}(V_i, V_j) = 0$ for all $i = 1, \ldots, p$ and $i \neq j$. To obtain a mix of continuous and categorial coefficients, we take $X_1$ - $X_{p/2} =$ $V_1$ - $V_{p/2}$ and $X_j = I[V_j < 0]$ for $j = p/2+1, \ldots p$. We generate the true propensity scores with a logistic model, $e(\bm{X}) = \{1 - \exp(-\sum_{i=0}^p \alpha_i)\}^{-1}$. In this model, $\alpha_0$ is selected to obtain the desired percent of treated observations (20\%, 40\%) and all other coefficients are determined as described in supplementary Section \ref{sec:sim_supp}. 
We simulate treatment assignment as $Z \sim \text{Bernoulli}\{e(\bm{X})\}$ and the outcome as $Y = \mu_z(\bm{X}) + N(0,1)$ for $Z=z$. Both $\mu_0(\bm{x})$ and $\mu_1(\bm{x})$ are modeled as linear combinations of $X_1, \ldots, X_p$. For $\theta = 25\%, 75\%$ of covariates, the corresponding $\mu_0(\bm{x})$ and $\mu_1(\bm{x})$ coefficients are equal  (i.e., these covariates are not effect modifiers). The remaining $\mu_0(\bm{x})$ and $\mu_1(\bm{x})$ coefficients are determined as described in supplementary Section \ref{sec:sim_supp}.

We generate a total of 72 simulation scenarios of all combinations $p=20, 100$, 20\% and 40\% treated, 25\% and 75\% of covariates being effect modifiers, three levels of overlap, and three levels of treatment effect heterogeneity. For each simulation scenario we generate 1,000 independent datasets with either $n = 1000$ ($p=20$) or $n = 2000$ ($p=100$) observations. 
We compute 1) IPW, direct balancing weights, and minimal weights estimators for the ATE; 2) our proposed estimator given the true $g(\bm{X})$ and the design- and model-based estimators proposed in Section \ref{sec:implementation}; and 3) the IPW and balancing weights (i.e., all covariates in the set $g(\bm{x})$) estimators for the ATO. All balancing weights are computed with the stable balancing weights deviance measure \citep{zubizarreta_stable_2015} and with the additional constraints of $w_i \geq 0$ and $\sum_{i=1}^n w_iZ_i = \sum_{i=1}^n w_i(1-Z_i)$. We compare the mean squared error (MSE) of these estimators with respect to the true ATE.

\subsection{Results}
Our proposed estimator implemented with the true $\tau(\bm{x})$ function has the minimum estimator MSE with respect to the ATE across almost all simulation scenarios where there is a solution to Problem \eqref{eq:our_db} for the true $\tau(\bm{x})$ function (Figure \ref{fig:mse_sim} and supplementary Figure \ref{fig:mse_sim_supp}). Specifically, this estimator tends to be minimally biased and have comparable variance to the minimal weights and ATO estimators (supplementary Figures \ref{fig:p20_sim}-\ref{fig:p100_sim_supp}). Furthermore, when there is a lack of overlap and no solution to Problem \eqref{eq:db}, our proposed estimator is the only estimator that has minimal bias for the ATE.
However, it is rare that the true $\tau(\bm{x})$ is known, so we focus on the performance of our proposed design- and model-based estimators that use a data-driven procedure to select $g(\bm{x}).$

\begin{figure}
    \centering
    \includegraphics[width=0.9\linewidth]{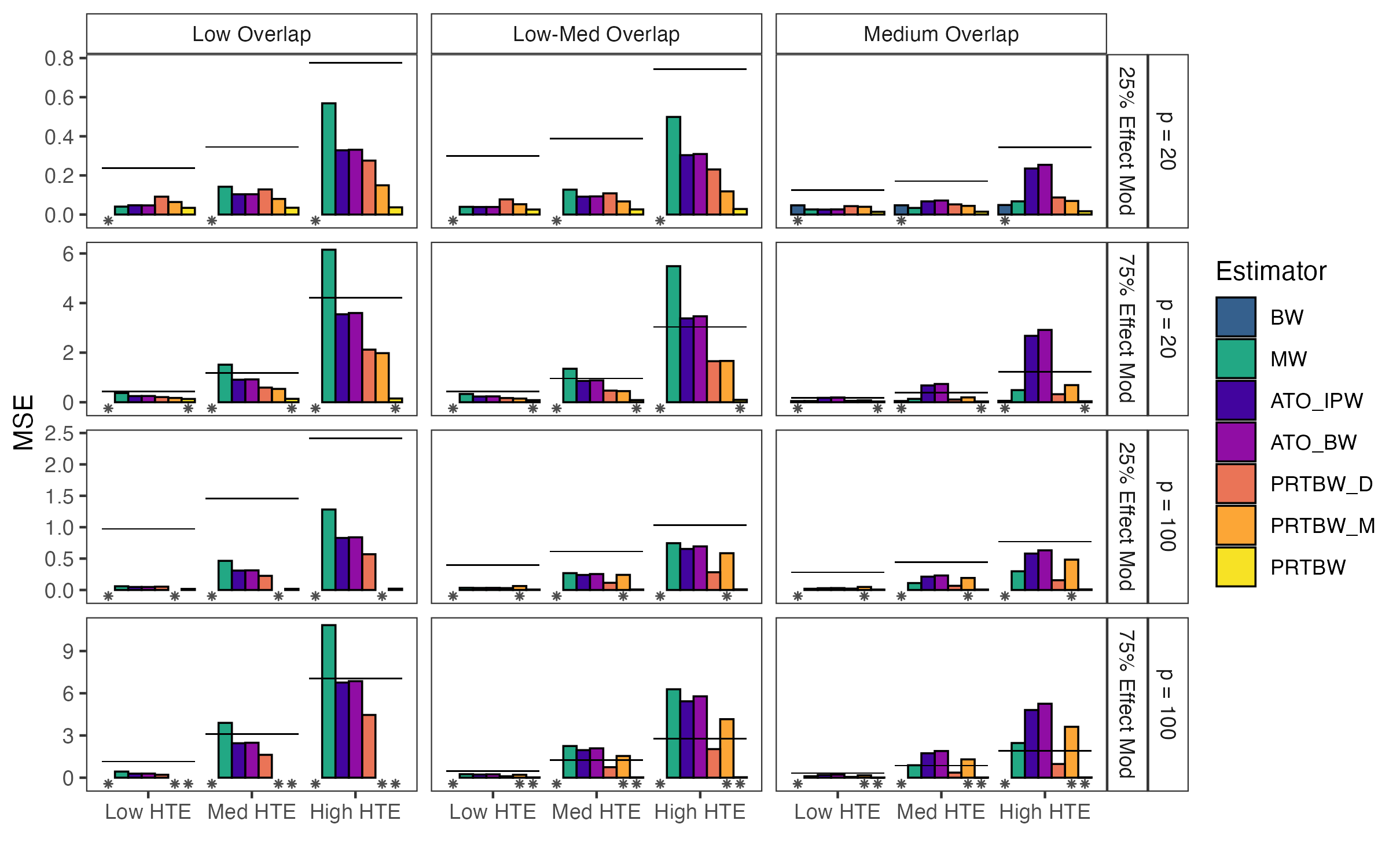}
    \caption{\footnotesize Estimator mean squared error (MSE) for all 20\% treated simulation scenarios. The line indicates the MSE for the IPW ATE estimator for each scenario. The asterisk, $*$, indicates estimators where a solution to the balancing optimization problem did not exist for all datasets. The horizontal axis indicates levels of treatment effect heterogeneity while each panel is labeled by different overlap levels, number of predictors, and percent of covariates that are treatment effect modifiers.}
    \label{fig:mse_sim}
\end{figure}

Across almost all simulation scenarios, either our proposed design- or model-based estimator has minimum estimator MSE with respect to the ATE compared to either the minimal weights or ATO estimators (Figure \ref{fig:mse_sim} and supplementary Figure \ref{fig:mse_sim_supp}). Furthermore, the design-based estimator is the only estimator that has 1) a solution to Problem \eqref{eq:our_db}; and 2) lower MSE than the IPW ATE estimator across all scenarios. The minimal weights estimator tends to have the largest MSE, even larger than the ATE and ATO estimators, when there is moderate-high treatment effect heterogeneity and low overlap. When there is low treatment effect heterogeneity, the minimal weights estimator tends to have the smallest MSE. However, this is due to the estimator achieving minimal variance but possessing substantial statistical bias; thus, our proposed estimators may still be preferable in these scenarios (supplementary Figure \ref{fig:p20_sim}). Our model-based estimator tends to perform better than the design-based estimator for $p=20$, while the reverse is true when $p=100$. We likely see this trend because there may not be sufficient sample size for estimating treatment effect modification and computing the weights with sample splitting for $p=100$. 
We see similar trends in estimator MSE between the 20\% treated (Figure \ref{fig:mse_sim}) and 40\% treated scenarios (supplementary Figure \ref{fig:mse_sim_supp}), though for the 40\% treated scenarios all estimators substantially outperform the IPW estimator.

\section{Two illustrative applications}
\label{sec:illustrative_app}
We apply our proposed weighted estimators for ATE estimation to an EHR study and adapt Problem \eqref{eq:our_db} for transportability to transport a treatment effect from an RCT to a Midwestern academic health center population. We compare our proposed estimators to IPW, direct balancing, and minimal weights estimators for the ATE and ATO and demonstrate their superior performance in terms of covariate balance and estimator variance for these two applications.
\subsection{EHR study on the effect of indwelling arterial catheters on mortality}
\label{sec:real_data_mimic}
We perform a reanalysis of \cite{hsu_association_2015} which estimates the ATE of indwelling arterial catheters (IACs) on mortality in patients with respiratory failure using observational data from the MIMIC-III v1.4 critical care database \citep{johnson_mimic-iii_2016}. 
The data has 2,522 observations of mechanically ventilated patients, where 51.5\% of patients received IAC. The outcome of interest is an indicator of mortality 28 days within hospital admission. Missing data is imputed with single imputation. Pre-treatment covariates include demographic information, baseline measurements (e.g., blood pressure, lab values), risk scores, and missing data indicators. In total, we identify 72 pre-treatment covariates of interest.  There is a lack of overlap in the estimated propensity scores in this data with many propensity scores close to 1 (supplementary Figure \ref{fig:MIMIC_ps}). Due to this lack of overlap, there is no solution to Problem \eqref{eq:db} with the additional constraints of $w_i \geq 0$ and $\sum_{i=1}^n w_iZ_i = \sum_{i=1}^n w_i(1-Z_i)$. We estimate the ATE with IPW, minimal weights, and our proposed design- and model-based estimators that use Algorithm \eqref{alg:adaptive} to select $g(\bm{x})$. Since rare binary covariates tend to be challenging to exactly balance, we also compute our proposed estimators when initializing Algorithm \eqref{alg:adaptive} with the set of binary covariates whose outcomes occur for < 5\% of observations. We also estimate the ATO using IPW and balancing weights (i.e., all covariates are in the $g(\bm{x})$ set).  All balancing weights are computed with the additional constraints of $w_i \geq 0$ and $\sum_{i=1}^n w_iZ_i = \sum_{i=1}^n w_i(1-Z_i)$. We compare the covariate balance, treatment effect estimates, and Wald-type bootstrapped 95\% confidence intervals (CIs) between all methods.

The treatment effect estimates and 95\% CIs  are similar for all estimators except the IPW estimator, which corresponds to a more negative estimate and a larger 95\% CI (Figure \ref{fig:MIMIC}). However, all estimates and 95\% CIs indicate that IAC non-significantly reduces 28 day mortality by a small percentage. The fact that the minimal weights, overlap weights, and PRTBW estimates are similar likely indicates that there is limited treatment effect heterogeneity. Thus, 
these estimators primarily differ in their target populations and the covariate balance achieved. Figure \ref{fig:MIMIC} shows the weighted treated and control sample means for three covariates in comparison to the sample mean (the black horizontal line). IPW and minimal weights do not exactly balance any of the covariates, with corresponding average standardized mean differences (SMD) of $0.161$ and $0.008$, respectively. 
In contrast, all other weights exactly balance the treated and control distributions. 
The smallest $g(\bm{x})$ set is selected for our method when Algorithm \ref{alg:adaptive} is initiated with rare binary covariates.
In Figure \ref{fig:MIMIC}, the corresponding design-based weights balance both diastolic pressure and the SOFA score to the target population, which is not true of any of the other weights. Supplementary Figures \ref{fig:MIMIC_tc}-\ref{fig:MIMIC_cp} show the SMD for all covariates 1) between treated and control groups; and 2) between treated/control groups and the target population. While these estimators yield similar estimates and 95\% CIs for this application, our proposed estimators tend to have improved covariate balance and correspond to a minimally modified target population.

\begin{figure}
    \centering
    \includegraphics[width=0.7\linewidth]{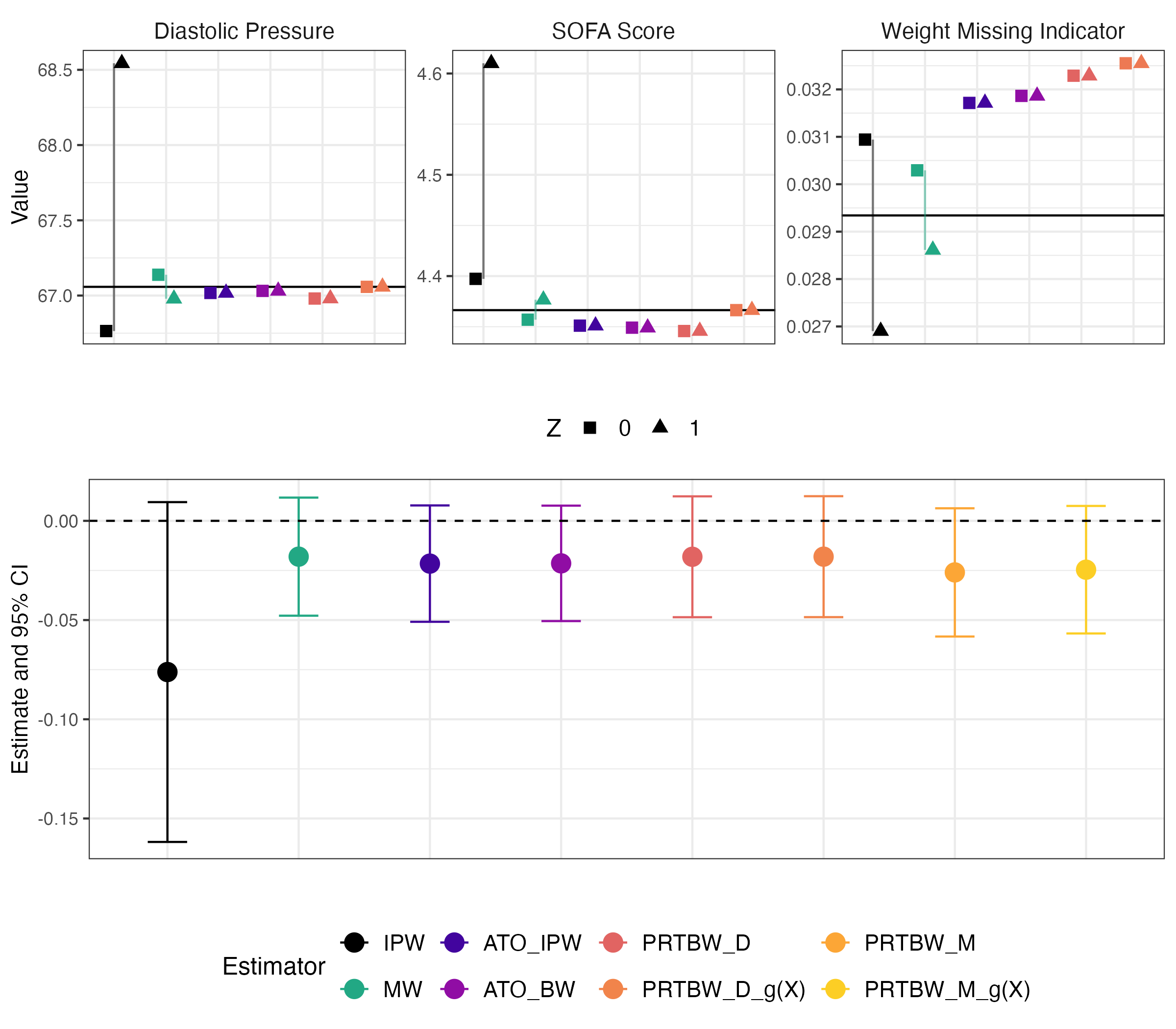}
    \caption{\footnotesize ATE estimates and 95\% confidence intervals of the effect of IAC on 28-day mortality (bottom figure) and weighted covariate balance of three covariates (top figures) for eight different weighting/estimation methods. For the top figures, the horizontal line indicates the sample mean of the covariates in the target population and vertical lines indicate the extent of covariate imbalance for each method.}
    \label{fig:MIMIC}
\end{figure}

\subsection{Transporting a “health care hotspotting” RCT to a Midwestern U.S. academic health center population}
\label{sec:real_data_hotspot}
A recent intervention of popular interest, termed ``health care hotspotting'', aimed to improve healthcare delivery for individuals who interact substantially with the healthcare system through intensive, targeted follow-up after admission. Despite the promising nature of this intervention, an RCT in Camden, NJ found a null effect of this intervention on readmission rates \citep{finkelstein_health_2020}. However, there is evidence that the treatment effect varies by population \citep{yang_hospital_2023} such that we are interested in determining the effect of this intervention in a Midwestern U.S. academic health center population. 
We perform a complete case analysis (< 30 missing observations) with a total of 781 trial population observations and 1305 target population observations. The outcome of interest is a readmission within 30 days. We include all common pre-treatment covariates between both datasets, which includes demographics and hospital stay characteristics, for a total of 21 covariates. There is minimal overlap in the estimated probability of being in the target population primarily due to race; the trial populations is 84.9\% non-white while the target population is 13.6\% non-white (supplementary Figure \ref{fig:hotspot_ps}).  
Thus, we implement our proposed design-based estimator and our weighted estimator when $g(\bm{x})$ \textit{a priori} only includes race. We compare this estimator to the normalized inverse odds weighted (IOW) estimator \citep{dahabreh_extending_2020}, the IOW estimator with population probabilities outside of $[0.1, 0.9]$ trimmed \citep{crump_dealing_2009}, and the direct balancing weights estimator. All balancing weights are computed with the additional constraints of $w_i \geq 0$ and $\sum_{i=1}^n w_iZ_i = \sum_{i=1}^n w_i(1-Z_i)$. We report covariate balance, transported treatment effect estimates, and Wald-type bootstrapped 95\% CIs for these methods.

 The IOW and direct balancing weights estimates and 95\% CIs indicate a non-significant reduction in 30-day readmission (Figure \ref{fig:hotspot}); however, there is no solution to the direct balancing weights optimization problem for many of the bootstrapped datasets such that a true 95\% bootstrapped CI cannot be derived. The trimmed estimator has the largest treatment effect estimate, but also the largest 95\% CI, likely due to the small sample size after trimming. While the 95\% CIs for both of our proposed estimators include zero, the intervals are substantially smaller and the estimates indicate that the intervention may actually increase hospital readmission rates in the target population. Both the IOW and trimmed estimators achieve poor covariate balance with average SMDs of $0.140$ and $0.382$. Supplementary Figures \ref{fig:hotspsot_tc}-\ref{fig:hotspot_cp} show the SMD for all covariates 1) between treated and control groups; and 2) between treated/control groups and the target population. Direct balancing weights and our proposed weights achieve exact covariate balance; however, our proposed weights do modify the sample mean of the race covariate such that the estimate applies to a population that is approximately 80\% non-white. Therefore, the PRTBW estimates indicate that the intervention may in fact increase 30-day readmission rates for either 1) this slightly modified target population; or 2) for the original target population when race is not a treatment effect modifier.

\begin{figure}
    \centering
    \includegraphics[width=0.6\linewidth]{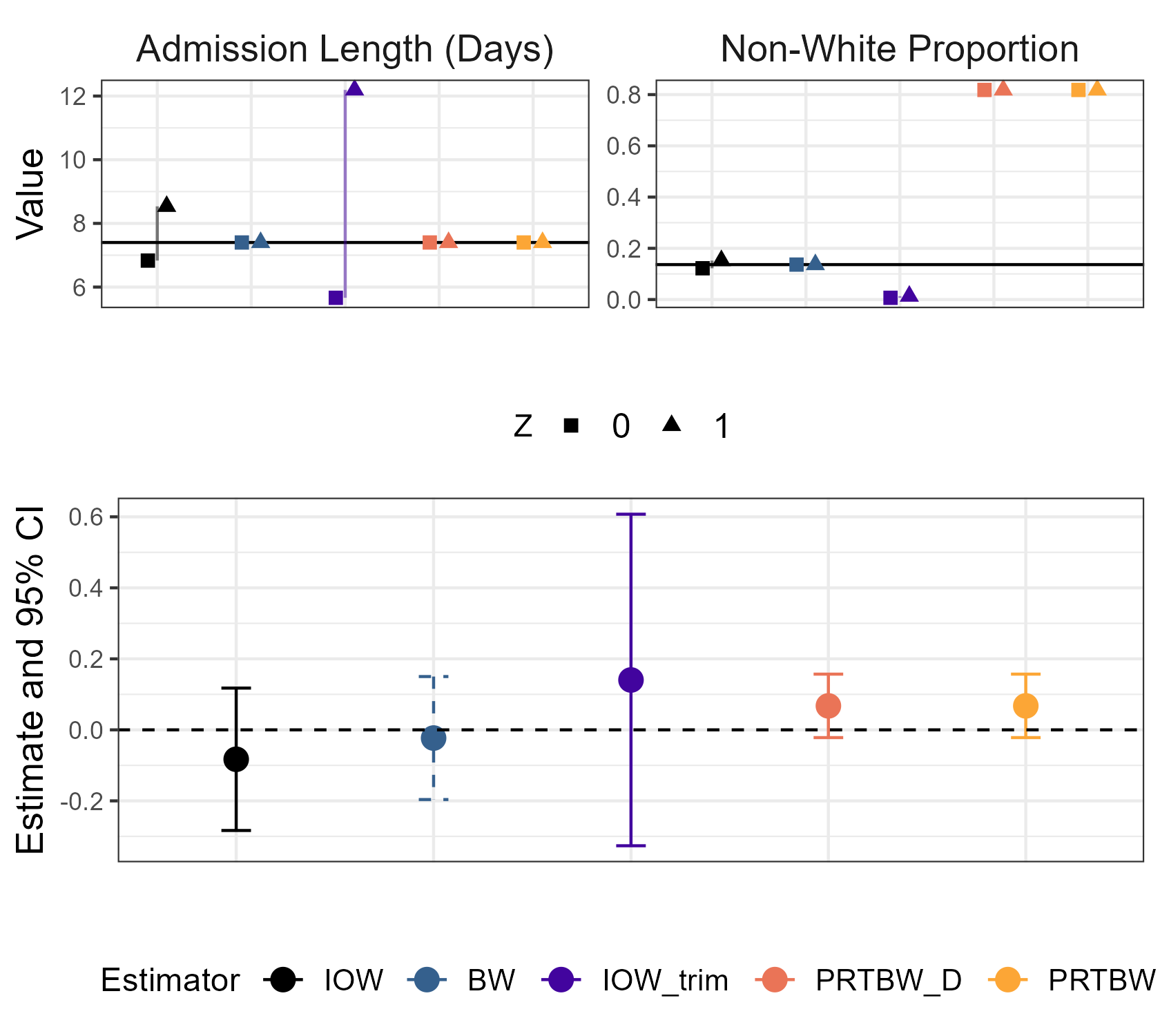}
    \caption{\footnotesize Transported treatment effect estimates and 95\% confidence intervals of the effect of ``health care hotspotting'' on 30-day readmission rates for a Midwestern healthcare center population (bottom figure) and weighted covariate balance of two covariates (top figures) for five different weighting/estimation methods. For the top figures, the horizontal line indicates the sample mean of the covariates in the target population and vertical lines indicate the extent of covariate imbalance for each method.}
    \label{fig:hotspot}
\end{figure}

\section{Discussion}
\label{sec:discussion}

Positivity violations create substantial challenges when estimating causal effects with observational data. These challenges necessitate researchers to confront the trade-offs between estimator bias, variance, and maintaining the original target population of interest. In this work, we proposed a novel weighting procedure that achieves reduced estimator bias and variance through relaxing the balancing constraints to the target population. In doing so, our proposed weights can be derived under an relaxed positivity assumption. Our approach offers a choice in between direct balancing and overlap weights (which are special cases of our method), allowing the analyst to navigate between preservation of the original estimand and mitigation of positivity violations. We have shown that our proposed estimator 1) is consistent for the original estimand when either the implied propensity score model is correct or the set of treatment effect modifiers is properly specified; and 2) is consistent for a slightly modified estimand that is simple to characterize when these conditions do not hold. Furthermore, we have shown that our proposed estimator achieves reduced asymptotic variance under the relaxed positivity assumption. To aid in the practical use of our method, we proposed an algorithm for identifying the set of covariates balanced to a modified population; the corresponding design- and model-based estimators perform well across applications to synthetic data, EHR data, and when transporting RCT effects.

While we focus on estimating the ATE, our proposed weighting procedure can be readily extended to other causal effect estimation problems as discussed in Section \ref{sec:related_estimation} and supplementary Section \ref{sec:connect_methods_supp}. Our proposed extension to distributional balancing weights may be especially promising as these weights do not require strong parametric assumptions on the outcome models. Future work includes exploring the theoretical properties of these extensions. There is also potential for combining the intuition behind our proposed method and modified treatment policies to obtain estimators with reduced bias and variance that correspond to both a treatment and population of interest. While our weighting procedure does perform well when there are positivity violations and higher dimensional covariates (e.g, $p=100$), there is a need for additional methods that address the combined challenges of variable selection (in the context of unmeasured confounding) and positivity violations in $p >> n$ scenarios. Further, while the infeasibility of Problem \eqref{eq:db} is a clear indication that our proposed method or minimal weights need to be used to derive weights, it is generally unclear when estimator performance may be impacted by overlap levels for a given research problem. Thus, future work includes developing measures of overlap that are directly related the impact of overlap on estimator MSE.

\section*{Funding}
This material is based upon work supported by the National Science Foundation Graduate Research Fellowship Program under Grant No. 2237827.

\section*{Disclosure Statement}
The authors report there are no competing interests to declare.

\section*{Data Availability Statement}
The data that supports Section \ref{sec:illustrative_app} are 1) MIMIC-III data: available with queries provided at \url{https://github.com/MIT-LCP/mimic-code}; 2) RCT data: available at \url{https://doi.org/10.7910/DVN/ZJVVQZ}; and 3) Health system target population data: not publicly available.

\putbib
\end{bibunit}


\pagebreak
\center{
\textbf{\Large Supplementary Materials for Partially Retargeted Balancing Weights for Causal Effect Estimation Under Positivity Violations}
}
\setcounter{section}{0}
\setcounter{algorithm}{0}
\setcounter{figure}{0}
\setstretch{1.2}

\begin{bibunit}

\justifying
\section{Technical results and proofs}
\label{sec:proofs_supp}
For all proofs, we notate $|| \cdot ||$ as the L2 norm.
\subsection{Theorem \ref{thm:dual}}
\begin{proof}
We rewrite the problem in matrix notation:
\begin{align*}
    \text{argmin}_{w_i}& \: \: \sum_{i=1}^n h(s_i) \\
    \text{subject to}& \:\: Q_{2K + L \times 2n}s_{2n \times 1} = 0_{2K + L \times 1}
\end{align*}
where 
\begin{align*}
   s_{2n \times 1} &= \begin{bmatrix}
           (1- Z_iw_i)_{n\times 1} \\
           (1- (1-Z_i)w_i)_{n\times 1}
         \end{bmatrix}  \\[8pt]
    Q_{2K + L \times 2n}  &=   \begin{bmatrix}
           c_1(\bm{X}_1) & c_1(\bm{X}_2) & \cdots & c_1(\bm{X}_n) & 0 & 0 & \cdots & 0 \\
           c_K(\bm{X}_1) & c_K(\bm{X}_2) & \cdots & c_K(\bm{X}_n) & 0 & 0& \cdots & 0 \\
           0 & 0 & \cdots & 0 &c_1(\bm{X}_1) & c_1(\bm{X}_2) & \cdots & c_1(\bm{X}_n) &  \\
           0 & 0 & \cdots & 0 &c_K(\bm{X}_1) & c_K(\bm{X}_2) & \cdots & c_K(\bm{X}_n) &  \\
           g_1(\bm{X}_1) & g_1(\bm{X}_2) & \cdots & g_1(\bm{X}_n) & -g_1(\bm{X}_1) & -g_1(\bm{X}_2) & \cdots & -g_1(\bm{X}_n) \\
           g_L(\bm{X}_1) & g_L(\bm{X}_2) & \cdots & g_L(\bm{X}_n) & -g_L(\bm{X}_1) & -g_L(\bm{X}_2) & \cdots & -g_L(\bm{X}_n) \\
         \end{bmatrix}.
\end{align*}
This problem is now in the form of \cite{tseng_relaxation_1991}. The dual of this problem is
\begin{align*}
    \text{minimize}_{\lambda}& \:\:\: q(\lambda) \\
    \text{subject to}& \:\: \: \text{no constraint on}\:\: \lambda,
\end{align*}
where $q(\lambda) = \sum_{j=1}^{2n} h^*_j(Q_j^T\lambda)$ and $h^*_j(t) = \sup_{s_j}\{ts_j - h(s_j)\}$. Consider $j < n+1$. Then we have that
\begin{align*}
   h^*_j(t) &= \sup_{s_j}\{ts_j - h(s_j)\} \\
   &= \sup_{w_j} \{t - tZ_jw_j - h(1- Z_jw_j)\} \\
    &= \sup_{w_j} \{t - tZ_jw_j - Z_jh(1- w_j) - (1-Z_j)h(1)\} \\
    &= t- tZ_jw_j^* - Z_jh(1-w_j^*) - (1-Z_j)h(1).
\end{align*}
Next we have that
\begin{align*}
   h^*_{j+n}(t) &= \sup_{s_{j+n}}\{ts_{j+n} - h(s_{j +n})\} \\
   &= \sup_{w_j} \{t- t(1-Z_j)w_j - h[1- (1-Z_j)w_j]\} \\
    &= t- t(1-Z_j)w_j^* - (1-Z_j)h(1- w_j^*) - Z_jh(1),
\end{align*}
where $w_j^*$ satisfies the first order conditions:
\begin{align*}
    -tZ_j + Z_jh'(1- w^*_j) &= 0  \\
    -t(1-Z_j) + (1-Z_j)h(1- w_j^*) &= 0 \\
   \implies \:\: 1- (h')^{-1}(t)  &= w^*_j.
\end{align*}
Then we have that
\begin{align*}
    h^*_j(t) &= t- tZ_j(1- (h')^{-1}(t))  - Z_jh(1 - 1 + (h')^{-1}(t)) - (1-Z_j)h(1) \\
     &= t- tZ_j\{1- (h')^{-1}(t)\}  - Z_jh\{(h')^{-1}(t)\} - (1-Z_j)h(1) \\ 
     &= -Z_j\{t - t(h')^{-1}(t) +h[(h')^{-1}(t)]- h(-1)\} - h(1) + t \\ \text{and} \\
     h^*_{j+n}(t) &= t- t(1-Z_j)\{1- (h')^{-1}(t)\} - (1-Z_j)h\{1-1+ (h')^{-1}(t)\} - Z_jh(1) \\
     &= t- t(1-Z_j)\{1- (h')^{-1}(t)\} - (1-Z_j)h\{(h')^{-1}(t)\} - Z_jh(1) \\
     &= -(1-Z_j)\{t- t(h')^{-1}(t) + h[(h')^{-1}(t)] - h(1)\} - h(1) + t.
\end{align*}
Let $\rho(t) = t- t(h')^{-1}(t) + h[(h')^{-1}(t)] - h(1)$. This gives
\begin{align*}
    h^*_j(t) &= -Z_j\rho(t) - h(1) + t \\
    h^*_{j+n}(t) &= -(1-Z_j)\rho(t) - h(1) + t.
\end{align*}
Now, note that
\begin{align*}
    \rho'(t) &= 1 - (h')^{-1}(t) - t\{(h')^{-1}(t)\}' + h'[(h')^{-1}(t)] \times \{(h')^{-1}(t)\}' \\
    &= 1 - (h')^{-1}(t) - t\{(h')^{-1}(t)\}'  + t\{(h')^{-1}(t)\}' \\
    &= 1 - (h')^{-1}(t) \\
    \implies \rho'(t) &= w_j^*.
\end{align*}
Thus, the dual formulation becomes
\begin{align*}
    \text{minimize}_{\lambda}& \:\:\: q(\lambda) \\
    \text{subject to}& \:\: \: \text{no constraint on}\:\: \lambda,
\end{align*}
where 
\begin{align*}
    q(\lambda) &= \frac1n\sum_{j=1}^n -Z_j\rho(Q_j^T\lambda) - (1-Z_j)\rho(Q_{j+n}^T\lambda) + Q_j^T\lambda + Q_{j+n}^T\lambda - 2nh(1) \\
    &= \frac1n\sum_{j=1}^n -Z_j\rho(Q_j^T\lambda) - (1-Z_j)\rho(Q_{j+n}^T\lambda) + Q_j^T\lambda + Q_{j+n}^T\lambda.
\end{align*}
Let $B^+(\bm{X}_j) = \{c_1(\bm{X}_j), \ldots, c_K(\bm{X}_j), 0_K,  g_1(\bm{X}_j), \ldots, g_L(\bm{X}_j)\}$ \\ and $B^-(\bm{X}_j) = \{0_K, c_1(\bm{X}_j), \ldots, c_K(\bm{X}_j), -g_1(\bm{X}_j), \ldots, -g_L(\bm{X}_j)\}$. Then, the primal solution $w_j^*$ satisfies
\begin{align*}
    w^*_j &= \rho'\{B^+(\bm{X}_j)^T\hat{\lambda}\} \:\:\text{for}\:\: Z_j = 1 \:\: ; \:\: j = 1, \ldots n \\
    w^*_j &= \rho'\{B^-(\bm{X}_j)^T\hat{\lambda}\} \:\:\text{for}\:\: Z_j = 0 \:\:; \:\:j = 1, \ldots n,
\end{align*}
where $\hat{\lambda}$ is the solution to the dual problem.
Let $\lambda = \begin{bmatrix}
           \bm{\alpha}_{0_{K \times 1}} \\
           \bm{\alpha}_{1_{K \times 1}} \\
           \bm{\gamma}_{L \times 1}
         \end{bmatrix} $ 
and let $c(\bm{X}_j) = \{c_1(\bm{X}_j), \ldots c_K(\bm{X}_j)\}$ and $g(\bm{X}_j) = \{g_1(\bm{X}_j), \ldots, g_L(\bm{X}_j)\}$. Then the dual formulation simplifies to
\begin{align*}
    \underset{\bm{\alpha}_0, \bm{\alpha}_1, \bm{\gamma}}{\textnormal{minimize}}& \:\:\: q(\bm{\alpha}_0, \bm{\alpha}_1, \bm{\gamma}) \\
    \text{subject to}& \:\: \: \text{no constraint on}\:\: \bm{\alpha}_0, \bm{\alpha}_1, \bm{\gamma},\\ \\
    \text{where}\: \:  q(\alpha_0, \alpha_1, \gamma) &= \frac1n \sum_{i=1}^n -Z_i\rho\{\bm{\alpha}_1^Tc(\bm{X}_i) + \bm{\gamma}^Tg(\bm{X}_i)\} - (1-Z_i)\rho\{\bm{\alpha}_0^Tc(\bm{X}_i) -\bm{\gamma}^Tg(\bm{X}_i)\} \\
     &\quad+ \bm{\alpha}_1^Tc(\bm{X}_i) + \bm{\gamma}^Tg(\bm{X}_i) +\bm{\alpha}_0^Tc(\bm{X}_i) -  \bm{\gamma}^Tg(\bm{X}_i)  \\
    &= \frac1n \sum_{i=1}^n -Z_i\rho\{\bm{\alpha}_1^Tc(\bm{X}_i) + \bm{\gamma}^Tg(\bm{X}_i)\} - (1-Z_i)\rho\{\bm{\alpha}_0^Tc(\bm{X}_i) -\bm{\gamma}^Tg(\bm{X}_i)\} \\
     &\quad+ \bm{\alpha}_1^Tc(\bm{X}_i)+\bm{\alpha}_0^Tc(\bm{X}_i)  
\end{align*}

and the primal solution $w_{i}^* = Z_iw^*_{1i} + (1-Z_i)w^*_{0i}$ satisfies $w^*_{1i} = \rho'\{\hat{\bm{\alpha}}_1^Tc(\bm{X}_i) + \hat{\bm{\gamma}}^Tg(\bm{X}_i)\} \:\:\text{and}\:\: w^*_{0i} = \rho'\{\hat{\bm{\alpha}}_0^Tc(\bm{X}_i) - \hat{\bm{\gamma}}^Tg(\bm{X}_i)\}$
where $\hat{\bm{\alpha}}_0, \hat{\bm{\alpha}}_1, \hat{\bm{\gamma}}$ is the solution to the dual problem.
\end{proof}

\subsection{Proposition \ref{thm:relaxed_pos}}

\begin{proof}

We use similar proof techniques to Appendix A of \cite{zhao_entropy_2017}. 
Let $B_{\epsilon}(\bm{x}^*) = \{\bm{x} : || \bm{x} - \bm{x}^*||_\infty \leq \epsilon \}$ for some $\epsilon > 0$ and let $\text{CH}\{\bm{X}_i\}_{Z_i = z}$ be the convex hull generated by $\{\bm{X}_i\}_{Z_i = z}$. For the weights to exist, we need 1) $0 < \text{Pr}\{Z_i | c(\bm{X}_i) = c(\bm{x})\} < 1$ for all $\bm{x}$ as shown by \citep{zhao_entropy_2017}; and 2) $\text{CH}\{\bm{X}_i\}_{Z_i = 1} \cap \text{CH}\{\bm{X}_i\}_{Z_i = 0} \neq \emptyset$. It remains to determine when $\text{CH}\{\bm{X}_i\}_{Z_i = 1} \cap \text{CH}\{\bm{X}_i\}_{Z_i = 0} \neq \emptyset$ holds. We show that if there exists $\bm{x}^*$ such that that $0 < P(Z_i|\bm{X}_i=\bm{x}^*) < 1$, a slightly stronger claim holds: $P\{\text{CH}\{\bm{X}_i\}_{Z_i = 1} \cap \text{CH}\{\bm{X}_i\}_{Z_i = 0} = B_{\epsilon}(\bm{x}^*)\} = P\{B_{\epsilon}(\bm{x}^*) \in \text{CH}\{\bm{X}_i\}_{Z_i = 1}, B_{\epsilon}(\bm{x}^*) \in \text{CH}\{\bm{X}_i\}_{Z_i = 0}\} \to 1$ as $n \to \infty$.

Let $R_i(\bm{x}^*) = i = 1, \ldots 3^p$ be the $3^p$ boxes centered at $\bm{x}^* + \frac{3}{2}\epsilon b$ where $b \in \mathbb{R}^p$ contains entries of $0,-1,1$. Then,  $P(B_{\epsilon}(\bm{x}^*) \in \text{CH}\{\bm{X}_i\}_{Z_i = 1}, B_{\epsilon}(\bm{x}^*) \in \text{CH}\{\bm{X}_i\}_{Z_i = 0}) \geq P(\exists \bm{X}_i \in R_i(\bm{x}^*), Z_i = 0, \exists \bm{X}_j \in R_j(\bm{x}^*), Z_j = 1, i,j = 1, \ldots, 3^p)$. 

Assume that for some $\bm{x}^*$, for all $\bm{x} \in B_{2\epsilon}(\bm{x}^*), \:\:\: 0 < P(Z_i|\bm{X}_i = \bm{x}) < 1$. 
This assumption implies that $\rho_1 = \min_{i} P(\bm{X} \in R_i(\bm{x}^*)|Z = 1) > 0$ and $\rho_0 = \min_{i} P(\bm{X} \in R_i(\bm{x}^*)|Z = 0) > 0$. Then we have the following,
\begin{align*}
    P(\exists \bm{X}_i \in R_i(\bm{x}^*), Z_i = 0, &\exists \bm{X}_j \in R_j(\bm{x}^*), Z_j = 1, i,j = 1, \ldots, 3^p) \\
    &= P(\exists \bm{X}_i \in R_i(\bm{x}^*), Z_i = 0, i = 1, \ldots, 3^p) \\
    &\quad\times P(\exists \bm{X}_j \in R_j(\bm{x}^*), Z_j = 1, j = 1, \ldots, 3^p) \\
    &\geq \{1 - \sum_{i=1}^{3^p} P(X \not\in R_i(\bm{x}^*)|Z=0)^{n}\} \\
    &\quad\times \{1 - \sum_{i=1}^{3^p} P(X \not\in R_j(\bm{x}^*)|Z=1)^{n}\} \\
    &= \{1 - 3^p(1-\rho_0)^n\} \{1 - 3^p(1-\rho_1)^n\} \\
    \lim_{n\to \infty} &\to 1,
\end{align*}
where the first equality holds because each $(\bm{X}_i, Z_i)$ is an iid sample. We see that letting $\epsilon \to 0$, the only assumption required  is the existence of some $\bm{x}^*$ such that $0 < P(Z_i|\bm{X}_i=\bm{x}^*) < 1$.

\end{proof}

\subsection{Proposition \ref{prop:supp}}
\begin{proposition}
Suppose Assumptions \ref{strong-ignorability}, either \ref{positivity} or \ref{relaxed-positivity}, and \ref{balanc_func_regularity} hold and $\rho$ has a continuous first derivative. Let $\bm{\theta} = (\bm{\alpha}_0, \bm{\alpha}_1, \bm{\gamma}).$ Then, \\
i) $\hat{\bm{\theta}} \to_p \tilde{\bm{\theta}}$ where $\tilde{\bm{\theta}} = (\tilde{\bm{\alpha}}_0, \tilde{\bm{\alpha}}_1, \tilde{\bm{\gamma}}).$ is the unique minimizer of $E[-e(\bm{X})\rho\{\bm{\alpha}_1^Tc(\bm{X}) + \bm{\gamma}^Tg(\bm{X})\} - \{1-e(\bm{X})\}\rho\{\bm{\alpha}_0^Tc(\bm{X}) - \bm{\gamma}^Tg(\bm{X})\}  + \bm{\alpha}_1^Tc(\bm{X}) + \bm{\alpha}_0^Tc(\bm{X})]$ and  \\
ii) $w^*_1(\bm{X}) = \rho'\{\hat{\bm{\alpha}}^T_1c(\bm{X}) + \hat{\bm{\gamma}}^Tg(\bm{X})\} \to_p \tilde{w}_1(\bm{X}) = \rho'\{\tilde{\bm{\alpha}}^T_1c(\bm{X}) + \tilde{\bm{\gamma}}^Tg(\bm{X})\}$ and $w^*_0(\bm{X}) = \rho'\{\hat{\bm{\alpha}}_0^Tc(\bm{X}) +-\hat{\bm{\gamma}}^Tg(\bm{X})\} \to_p \tilde{w}_0(\bm{X}) = \rho'\{\tilde{\bm{\alpha}}^T_0c(\bm{X}) - \tilde{\bm{\gamma}}^Tg(\bm{X})\}.$
\label{prop:supp}
\end{proposition}

We use similar proof techniques to the proof of proposition 1 in \cite{kallberg_large_2023}. First, we prove the following lemma.
\begin{lemma}
    Given Assumption \ref{balanc_func_regularity}, $l(\bm{\theta}) = E[-e(\bm{X})\rho\{\bm{\alpha}_1^Tc(\bm{X}_i) + \bm{\gamma}^Tg(\bm{X})\} - \{1-e(\bm{X})\}\rho\{\bm{\alpha}_0^Tc(\bm{X}) - \bm{\gamma}^Tg(\bm{X})\}  + \bm{\alpha}_1^Tc(\bm{X}) + \bm{\alpha}_0^Tc(\bm{X})]$ has a unique minimizer $\tilde{\bm{\theta}}$ that satisfies 
    $\inf_{\bm{\theta}} l(\bm{\theta}) = l(\tilde{\bm{\theta}})$.
    \label{lemma}
\end{lemma}
\begin{proof}

    Let $\bm{\theta}^* = \bm{\theta}^*(t)$, $t \geq 1$, be a sequence of random vectors where $\lim_{t \to \infty} l(\bm{\theta}^*) = l(\tilde{\bm{\theta}}) < \infty$. Without loss of generality assume $\bm{\theta}^*/||\bm{\theta}^*|| \to_p \bm{a}$ where $\bm{a}$ is a constant vector. If $||\bm{\theta}^*||$ is bounded, $\bm{\theta}^* \to_p \tilde{\bm{\theta}}$ for some $\tilde{\bm{\theta}}$ where $\inf_{\bm{\theta}} l(\bm{\theta}) = l(\tilde{\bm{\theta}})$ by the continuity of $\rho$ such that $\tilde{\bm{\theta}}$ is a minimum of $l(\bm{\theta})$. The solution $\tilde{\bm{\theta}}$ is unique if $l(\bm{\theta})$ is strictly convex. Thus, it is sufficient to show that $||\bm{\theta}^*||$ is bounded and $l(\bm{\theta})$ is strictly convex.
    
    We must first show that $||\bm{\theta}^*||$ is bounded. We will use proof by contradiction and assume that $||\bm{\theta}^*|| \to \infty$. Since the covariance matrix of $b(\bm{X}) = \{c(\bm{X}), g(\bm{X})\}$ is non-singular by Assumption \ref{balanc_func_regularity}, $\bm{\alpha}_1^{T}c(\bm{X}_i) + \bm{\gamma}^{T}g(\bm{X})\} \to \pm \infty$ and $\bm{\alpha}_0^{T}c(\bm{X}_i) - \bm{\gamma}^{T}g(\bm{X})\} \to \pm \infty$ with non-zero probability. Then since $\rho$ is strictly concave, this implies that $-\rho\{\bm{\alpha}_1^{T}c(\bm{X}_i) + \bm{\gamma}^{T}g(\bm{X})  \to \infty$ and $-\rho\{\bm{\alpha}_0^{T}c(\bm{X}_i) -\bm{\gamma}^{T}g(\bm{X})  \to \infty$. However, this implies that $\lim_{t \to \infty} l(\bm{\theta}^*) \to \infty$ with non-zero probability, a contradiction to $\lim_{t \to \infty} l(\bm{\theta}^*) = l(\tilde{\bm{\theta}})$. Thus, $||\bm{\theta}^*||$ is bounded and $\bm{\theta}^* \to_p \tilde{\bm{\theta}}$ for some $\tilde{\bm{\theta}}$ such that $\inf_{\bm{\theta}} l(\bm{\theta}) = l(\tilde{\bm{\theta}})$. It remains to show that this solution is unique by proving the strict convexity of $l(\bm{\theta})$. We will show that the Hessian of $l(\bm{\theta})$, $\bm{H}(\bm{\theta})$, is positive-definite. Now $\bm{H}(\bm{\theta})$ is given by the following (row, column), 
    \begin{align*}
        \bm{H}_{k_1, k_2}(\bm{\theta}) &= E[-e(\bm{X})\rho''\{\bm{\alpha}_1^{T}c(\bm{X}) + \bm{\gamma}^{T}g(\bm{X})\}c_{k_1}(\bm{X})c_{k_2}(\bm{X}) \\
        &\quad- \{1-e(\bm{X})\}\rho''\{\bm{\alpha}_0^{T}c(\bm{X}) - \bm{\gamma}^{T}g(\bm{X})\}c_{k_1}(\bm{X})c_{k_2}(\bm{X})], \\[6pt]
        \bm{H}_{l_1, l_2}(\bm{\theta}) &= E[-e(\bm{X})\rho''\{\bm{\alpha}_1^{T}c(\bm{X}) + \bm{\gamma}^{T}g(\bm{X})\}g_{l_1}(\bm{X})g_{l_2}(\bm{X}) \\
        &\quad- \{1-e(\bm{X})\}\rho''\{\bm{\alpha}_0^{T}c(\bm{X}_i) - \bm{\gamma}^{T}g(\bm{X})\}g_{l_1}(\bm{X})g_{l_2}(\bm{X})], \\[6pt]
        \bm{H}_{k_1, l_1}(\bm{\theta}) &= E[-e(\bm{X})\rho''\{\bm{\alpha}_1^{T}c(\bm{X}_i) + \bm{\gamma}^{T}g(\bm{X})\}c_{k_1}(\bm{X})g_{l_1}(\bm{X}) \\
        &\quad- \{1-e(\bm{X})\}\rho''\{\bm{\alpha}_0^{T}c(\bm{X}_i) - \bm{\gamma}^{T}g(\bm{X})\}c_{k_1}(\bm{X})g_{l_1}(\bm{X})].
    \end{align*}
     We can construct an estimator for $\bm{H}(\bm{\theta})$, $\hat{\bm{H}}(\bm{\theta})$ using the sample equivalents of $\bm{H}_{k_1, k_2}(\bm{\theta}), \bm{H}_{l_1, l_2}(\bm{\theta}), $ and $\bm{H}_{k_1, l_1}(\bm{\theta})$. Let $c_i = -e(\bm{X}_i)\rho''\{\bm{\alpha}_1^{T}c(\bm{X}_i) + \bm{\gamma}^{T}g(\bm{X}_i)\} + -\{1-e(\bm{X_i})\}\rho''\{\bm{\alpha}_0^{T}c(\bm{X}_i) - \bm{\gamma}^{T}g(\bm{X_i})\}$. By the strict concavity of $\rho$, $c_i > 0$. Let $v^* = [\sqrt{c_1}, \ldots, \sqrt{c_n}]$. Let $v$ be any vector of length $n$. Then $v^T\hat{\bm{H}}(\bm{\theta})v = v^T \odot v^{*^T} B(\bm{X})^TB(\bm{X}) v \odot v^* > 0$ since $B(\bm{X})$ is full rank by Assumption \ref{balanc_func_regularity}. Then $\hat{H}$ is positive definite and  since $E[\hat{\bm{H}}(\bm{\theta})] = \bm{H}(\bm{\theta})$, $\hat{\bm{H}}(\bm{\theta}) \to_p \bm{H}(\bm{\theta})$, $\bm{H}(\bm{\theta})$ is positive definite and $l(\bm{\theta})$ is strictly convex as desired.
\end{proof}
Now we will prove Proposition \ref{prop:supp}.
\begin{proof}

By Theorem 6.3 in \cite{boos_essential_2013} and the boundedness of $b_j(\bm{X})$ (Assumption \ref{balanc_func_regularity}), $\hat{l}(\bm{\theta}) \to l(\bm{\theta})$ uniformly almost surely for all $\bm{\theta}$ in every compact subset of $\mathbb{R}^{2K+L}$ where $\hat{l}(\bm{\theta}) = \frac1n \sum_{i=1}^n -Z_i\rho\{\bm{\alpha}_1^Tc(\bm{X}_i) + \bm{\gamma}^Tg(\bm{X}_i)\} - (1-Z_i)\rho\{\bm{\alpha}_0^Tc(\bm{X}_i) - \bm{\gamma}^Tg(\bm{X}_i)\}  +\bm{\alpha}_1^Tc(\bm{X}_i) + \bm{\alpha}_0^Tc(\bm{X}_i)$. By Lemma 6, we know the Hessian of $\hat{l}(\bm{\theta})$ is positive definite and thus that there is a unique solution to $\hat{l}(\bm{\theta})$. Then, let $R_{m, r}$ be the indicator that the solution is within distance $r > 0$ of $\tilde{\bm{\theta}}$ for sample size $n = m$. We want to show that $\{R_{m, r}\}_{m \geq 1}$ has finitely many zeros such that $\sum_{m=1}^\infty (1-R_{m, r}) < \infty$ for an arbitrary $r$; this is equivalent to $\lim_{n \to \infty} P(|\hat{\bm{\theta}} - \tilde{\bm{\theta}}| > \epsilon) = 0$ for any $\epsilon > 0$ as desired. Let $\Delta = \{\bm{\theta}: ||\bm{\theta} - \tilde{\bm{\theta}}|| \leq r\}$ and  $\Delta^* = \{\bm{\theta}: ||\bm{\theta} - \tilde{\bm{\theta}}|| =r\}$. Then, $\hat{l}(\bm{\theta}) \to l(\bm{\theta})$ implies that there exists a number $n^*$ such that $|\hat{l}(\bm{\theta}) - l(\bm{\theta})| < r/2$ and  $|\hat{l}(\tilde{\bm{\theta}}) - l(\tilde{\bm{\theta}})| < r/2$ almost surely for $n \geq n^*$. Since $\tilde{\bm{\theta}}$ minimizes $l(\bm{\theta})$, $\hat{l}(\bm{\theta})$ has a minima in $\Delta$ and a solution for $n \geq n^*$. This implies that $R_{m, r} = 1$ for $m \geq n^*$ and thus that $\sum_{m=1}^\infty (1-R_{m, r}) < \infty$. Then, $\hat{\bm{\theta}} \to_p \tilde{\bm{\theta}}$ where $\tilde{\bm{\theta}}$ is the unique minimizer of $l(\bm{\theta})$ by Lemma 6.

Then since $\rho'$ is continuous, by the continuous mapping theorem we have that $w^*_1(\bm{X}) = \rho'\{\hat{\bm{\alpha}}^T_1c(\bm{X}) + \hat{\bm{\gamma}}^Tg(\bm{X})\} \to_p \tilde{w}_1(\bm{X}) = \rho'\{\tilde{\bm{\alpha}}^T_1c(\bm{X}) + \tilde{\bm{\gamma}}^Tg(\bm{X})\}$ and $w^*_0(\bm{X}) = \rho'\{\hat{\bm{\alpha}}_0^Tc(\bm{X}) -\hat{\bm{\gamma}}^Tg(\bm{X})\} \to_p \tilde{w}_0(\bm{X}) = \rho'\{\tilde{\bm{\alpha}}^T_0c(\bm{X}) - \tilde{\bm{\gamma}}^Tg(\bm{X})\}.$
\end{proof}

\subsection{Theorem \ref{theorem:db}}
\begin{proof}{(1)}

Assume that the true models are $\{e(\bm{x})\}^{-1} = \rho'\{\tilde{\bm{\alpha}}_1^{T}c(\bm{x}) + \tilde{\bm{\gamma}}^{T}g(\bm{x})\} $ and  $\{1-e(\bm{x})\}^{-1} = \rho'\{\tilde{\bm{\alpha}}_0^{T}c(\bm{x}) - \tilde{\bm{\gamma}}^{T}g(\bm{x})\}$. Then given Assumptions \ref{strong-ignorability} and \ref{positivity},
\begin{align*}
    \tau &= E\left[\frac{Z_iY_i}{e(\bm{X}_i)}\right] - E\left[\frac{(1-Z_i)Y_i}{1-e(\bm{X}_i)}\right] \\[6pt]
    &= E[\rho'\{\tilde{\bm{\alpha}}_1^{T}c(\bm{x}) + \tilde{\bm{\gamma}}^{T}g(\bm{x})\}Z_iY_i] - E[\rho'\{\tilde{\bm{\alpha}}_0^{T}c(\bm{x}) - \tilde{\bm{\gamma}}^{T}g(\bm{x})\}(1-Z_i)Y_i].
\end{align*}
Further,
\begin{align*}
    \tau_{w^*} &= \sum_{i=1}^n w^*(\bm{X})Z_iY_i - \sum_{i=1}^n w^*(\bm{X})(1-Z_i)Y_i \\
    &= \sum_{i=1}^n \rho'\{\hat{\bm{\alpha}}^T_1c(\bm{X}) + \hat{\bm{\gamma}}^Tg(\bm{X})\} Z_iY_i - \sum_{i=1}^n \rho'\{\hat{\bm{\alpha}}_0^Tc(\bm{X}) -\hat{\bm{\gamma}}^Tg(\bm{X})\}(1-Z_i)Y_i 
\end{align*}
Then, given Proposition \ref{prop:supp}, Assumption \ref{balanc_func_regularity}, and the assumption that $\rho''\{\cdot\}$ is continuous over $\mathbb{R}$, the law of large numbers for averages with estimated parameters applies (\cite{boos_essential_2013} Theorem 7.3) and
\begin{align*}
    \sum_{i=1}^n \rho'\{\hat{\bm{\alpha}}^T_1c(\bm{X}) + \hat{\bm{\gamma}}^Tg(\bm{X})\} Z_iY_i &\to_p E[\rho'\{\tilde{\bm{\alpha}}_1^{T}c(\bm{x}) + \tilde{\bm{\gamma}}^{T}g(\bm{x})\}Z_iY_i], \\
    \sum_{i=1}^n \rho'\{\hat{\bm{\alpha}}_0^Tc(\bm{X}) -\hat{\bm{\gamma}}^Tg(\bm{X})\}(1-Z_i)Y_i &\to_p E[\rho'\{\tilde{\bm{\alpha}}_0^{T}c(\bm{x}) - \tilde{\bm{\gamma}}^{T}g(\bm{x})\}(1-Z_i)Y_i].
\end{align*}
It follows directly that $\tau_{w^*} \to_p \tau.$
\end{proof}

\begin{proof}{(2)}
    
Assume $\mu_1(\bm{x}) = \bm{\beta}_1^Tc(\bm{x}) + \bm{\lambda}^Tg(\bm{x})$ and $\mu_0(\bm{x}) = \bm{\beta}_0^Tc(\bm{x}) + \bm{\lambda}^Tg(\bm{x})$ for some $\bm{\beta}_z \in \mathbb{R}^K$ and $\bm{\lambda} \in \mathbb{R}^L$. We will use a similar proof technique by \cite{zhao_entropy_2017} to show that augmenting the our proposed estimator with estimated outcome regressions does not change the estimator. Let $\hat{\mu}_0(\bm{x})$,  $\hat{\mu}_1(\bm{x})$, and $\hat{e}(\bm{x})$ be outcome regression and propensity score estimates. Then the standard doubly robust estimator for the ATE is,
\begin{equation*}
    \hat{\tau}_{DR} = \frac1{n_1} \sum_{Z_i =1}\left[\frac{Y_i - \hat{\mu}_1(\bm{X}_i)}{\hat{e}(\bm{X}_i)} -  \hat{\mu}_1(\bm{X}_i)\right] - \frac1{n_0} \sum_{Z_i =1}\left[\frac{Y_i - \hat{\mu}_0(\bm{X}_i)}{1 - \hat{e}(\bm{X}_i)} -  \hat{\mu}_0(\bm{X}_i)\right] 
\end{equation*}
Now consider using $w^*$ to estimate the propensity score. Then,
\begin{align*}
    \hat{\tau}_{DR} - \hat{\tau}_{w^*} &= -\sum_{Z_i =1} w^*(\bm{X}_i)\hat{\mu}_1(\bm{X}_i) + \frac1n \sum_i \hat{\mu}_1(\bm{X}_i)  + \sum_{Z_i =0} w(\bm{X}_i)^*\hat{\mu}_0(\bm{X}_i)- \frac1n \sum_i \hat{\mu}_0(\bm{X}_i) \\
    &= - \sum_{Z_i =1} w^*(\bm{X}_i)[\hat{\bm{\beta}}_1^Tc(\bm{X}_i) + \hat{\bm{\lambda}}^Tg(\bm{X})_i] + \frac1n \sum_i \hat{\bm{\beta}}_1^Tc(\bm{X}_i) + \hat{\bm{\lambda}}^Tg(\bm{X})_i \\
   &\quad+ \sum_{Z_i =0} w^*(\bm{X}_i)[\hat{\bm{\beta}}_0^Tc(\bm{X}_i) + \hat{\bm{\lambda}}^Tg(\bm{X})_i] - \frac1n \sum_i \hat{\bm{\beta}}_0^Tc(\bm{X}_i) + \hat{\bm{\lambda}}^Tg(\bm{X})_i \\
   &= \hat{\bm{\beta}}_1^T [\sum_i c(\bm{X}_i) - \sum_i w^*(\bm{X}_i)Z_ic(\bm{X}_i)] \\
   &\quad-  \hat{\bm{\beta}}_0^T [\sum_i c(\bm{X}_i) - \sum_i w^*(\bm{X}_i)(1-Z_i)c(\bm{X}_i)] \\
   &\quad+ \hat{\bm{\lambda}}^T[\sum_i w^*(\bm{X}_i)(1-Z_i)g(\bm{X}_i) - \sum_i w^*(\bm{X}_i)Z_ig(\bm{X}_i)] \\
   &= 0,
\end{align*}
and $\hat{\tau}_{DR} = \hat{\tau}_{w^*}$. Since $\hat{\tau}_{DR} \to_p \tau$ \citep{bang_doubly_2005} we also have $ \hat{\tau}_{w^*} \to_p \tau$.
\end{proof}

\subsection{Corollary \ref{coro:consistency}}
\begin{proof}
Assume $\mu_1(\bm{x}) = \bm{\beta}_1^Tc(\bm{x}) + \bm{\lambda}_1^Tg(\bm{x})$ and $\mu_0(\bm{x}) = \bm{\beta}_0^Tc(\bm{x}) + \bm{\lambda}_0^Tg(\bm{x})$ for some $\bm{\beta}_z \in \mathbb{R}^K$ and $\bm{\lambda}_z \in \mathbb{R}^L$. Then, given Proposition \ref{prop:supp}, Assumption \ref{balanc_func_regularity}, and the assumption that $\rho''\{\cdot\}$ is continuous over $\mathbb{R}$, the law of large numbers for averages with estimated parameters applies (\cite{boos_essential_2013} Theorem 7.3) and
\begin{align*}
    \frac1n \sum_{i=1}^n w^*(\bm{X})Z_iY_i &\to_p E[\tilde{w}(\bm{X}_i)Z_iY_i] \\
    &= E[\tilde{w}(\bm{X}_i)Z_i\{\bm{\beta}_1^Tc(\bm{X}_i) + \bm{\lambda}_1^Tg(\bm{X}_i)\}] \\
    &= E[\bm{\beta}_1^Tc(\bm{X}_i) + \tilde{w}(\bm{X}_i)Z_i\bm{\lambda}_1^Tg(\bm{X}_i)] \\
    &= E[\bm{\beta}_1^Tc(\bm{X}_i) + \bm{\lambda}_1^T\tilde{w}(\bm{X}_i)e(\bm{X}_i)g(\bm{X}_i)] \\
    &= E[E[Y_i(1)|c(\bm{X}_i), \tilde{w}(\bm{X}_i)e(\bm{X}_i)g(\bm{X}_i)]] \\
    &= E[E[Y_i|Z_i=1, c(\bm{X}_i), \tilde{w}(\bm{X}_i)e(\bm{X}_i)g(\bm{X}_i)]]. 
\end{align*}
The second equality holds by $\tilde{w}$ satisfying the first order conditions of the population version of Equation \eqref{eq:og_loss}. Similarly, $\frac1n \sum_{i=1}^n w^*(\bm{X})(1-Z_i)Y_i \to E[E[Y_i|Z_i=0, c(\bm{X}_i), \tilde{w}(\bm{X}_i)\{1-e(\bm{X}_i)\}g(\bm{X}_i)]]$. Then $\hat{\tau}_{w^*} \to \tau_{g, \tilde{w}}$ as desired.
\end{proof}

\subsection{Theorem \ref{theorem:semipar_eff}}
\begin{proof}

We can decompose $\hat{\tau}_{w^*} - \tau$ into the following:
\begin{align*}
    \hat{\tau}_{w^*} - \tau &= \frac1n \sum_{i=1}^n w^*(\bm{X}_i)Z_iY_i - w^*(\bm{X}_i)(1-Z_i)Y_i - \tau \\
    &= \frac1n\sum_{i=1}^n w^*_1(\bm{X}_i)Z_i\{Y_i - \mu_1(\bm{X}_i)\} + \frac1n\sum_{i=1}^n (w^*_1(\bm{X}_i)Z_i -1)\mu_1(\bm{X}_i) \\
    &\quad- \frac1n\sum_{i=1}^n w^*_0(\bm{X}_i)(1-Z_i)\{Y_i - \mu_0(\bm{X}_i)\} + \frac1n\sum_{i=1}^n \{w^*_0(\bm{X}_i)(1-Z_i) -1\}\mu_0(\bm{X}_i) \\
    &\quad+ \frac1n\sum_{i=1}^n \tau(\bm{X}_i) - \tau \\
    &= \frac1n \sum_{i=1}^n S_i + R_1 + R_2,
\end{align*}
where
\begin{align*}
    S_i &= \frac{Z_i}{e(\bm{X}_i)}\{Y_i - \mu_1(\bm{X}_i)\} - \frac{1-Z_i}{1-e(\bm{X}_i)}\{Y_i - \mu_0(\bm{X}_i)\} + \{\tau(\bm{X}_i) - \tau\} \\
    R_1 &= \frac1n \sum_{i=1}^n (w^*_1(\bm{X}_i) - \frac1{e(\bm{X}_i)})Z_i\{Y_i - \mu_1(\bm{X}_i)\} - (w^*_0(\bm{X}_i) - \frac1{1-e(\bm{X}_i)})(1-Z_i)\{Y_i - \mu_0(\bm{X}_i)\} \\
    R_2 &= \frac1n\sum_{i=1}^n (w_1^*(\bm{X}_i)Z_i -1)\mu_1(\bm{X}_i)\}] - \{w_0^*(\bm{X}_i)(1-Z_i) -1\}\mu_0(\bm{X}_i) \\
     &= \frac1n \sum_{i=1}^n (w^*_1(\bm{X}_i)Z_i - 1)\{\bm{\beta}_1^Tc(\bm{X}_i) +\bm{\lambda}^Tg(\bm{X}_i)\} - \{w^*_0(\bm{X}_i)(1-Z_i) - 1\}\{\bm{\beta}_1^Tc(\bm{X}_i) +\bm{\lambda}^Tg(\bm{X}_i)\} \\
    &= \frac1n \sum_{i=1}^n (w^*_1(\bm{X}_i)Z_i - 1)\bm{\beta}_1^Tc(\bm{X}_i) - \sum_{i=1}^n \{w_0^*(\bm{X}_i)(1-Z_i) - 1\}\bm{\beta}_0^Tc(\bm{X}_i) \\
    &\quad+ \sum_{i=1}^n \{w_1^*(\bm{X}_i)Z_i - w_0^*(\bm{X}_i)(1-Z_i)\}\bm{\lambda}^Tg(\bm{X}_i).
\end{align*}
Note that $R_2 = 0$ by the constraints in the optimization problem. Then, since $S_i$ takes the same form as the efficient score for the ATE in \cite{hahn_role_1998}, $\hat{\tau}_{w^*}$ is asymptotically normal and semiparametrically efficient as long as $R_1$ is $o_p(n^{-1/2}).$

We will focus on the first component of $R_1$, which we decompose as $R_1 = R_1^1 + R_1^2$. We can bound the first term as
\begin{align*}
   |R_1^1| &= \left\lvert\frac1n \sum_{i=1}^n \left(w^*_1(\bm{X}_i) - \frac1{e(\bm{X}_i)}\right)Z_i\{Y_i - \mu_1(\bm{X}_i)\}\right\rvert \\[4pt]
    &= \left\lvert\int \left(w^*_1(\bm{X}_i) - \frac{1}{e(\bm{X}_i)}\right)\{Z_iY_i - Z_i\mu_1(\bm{X}_i)\} dF_n(\bm{x})\right\rvert\\[4pt]
    &\leq \left\lVert w^*_1(\bm{X}_i) - \frac{1}{e(\bm{X}_i)}\right\rVert \cdot \left\lVert\{Z_iY_i - Z_i\mu_1(\bm{X}_i)\}\right\rVert \\[4pt]
    &= o_p(1)O_p(n^{-1/2}) = o_p(n^{-1/2}),
\end{align*}
where the third inequality follows from Cauchy-Schwartz and the fourth equality follows from Theorem 1 and and the central limit theorem. Similar arguments hold for $R_1^2$ and we get that 
\begin{align*}
    R_1 &\leq |R_1^1| + |R_1^2| \\
    &= o_p(n^{-1/2}) + o_p(n^{-1/2}) = o_p(n^{-1/2})
\end{align*}
as desired. 
\end{proof}

\subsection{Corollary \ref{theorem:new_var}}
\begin{proof}
    
When the propensity score model is not correctly specified, $R_1$ will no longer go to zero. Then, the variance of $\hat{\tau}_{w^*}$ will be $\Var(S_i + R_{1_i})$ where
\begin{align*}
    S_i + R_{1_i} &= w_1^*(\bm{X}_i)Z_i\{Y_i - \mu_1(\bm{X}_i)\} - w_0^*(\bm{X}_i)(1-Z_i)\{Y_i - \mu_0(\bm{X}_i)\} + \{\tau(\bm{X}_i) - \tau\} \\
    &= w_1^*(\bm{X}_i)Z_i\{Y_i(1) - \mu_1(\bm{X}_i)\} - w_0^*(\bm{X}_i)(1-Z_i)\{Y_i(0) - \mu_0(\bm{X}_i)\} + \{\tau(\bm{X}_i) - \tau\} \\
    &= \{w_1^*(\bm{X}_i) - \tilde{w}_1(\bm{X}_i)\}Z_i\{Y_i(1) - \mu_1(\bm{X}_i)\} - \{w_0^*(\bm{X}_i) - \tilde{w}_0(\bm{X}_i)\}(1-Z_i)\{Y_i(0) - \mu_0(\bm{X}_i)\} \\
    &\quad+ \tilde{w}_1(\bm{X}_i)Z_i\{Y_i(1) - \mu_1(\bm{X}_i)\} - \tilde{w}_0(\bm{X}_i)(1-Z_i)\{Y_i(0) - \mu_0(\bm{X}_i)\} + \{\tau(\bm{X}_i) - \tau\}.
\end{align*}
Using similar arguments used for $R_1$ above and Proposition 3,
\begin{align*}
    \{w_1^*(\bm{X}_i) - \tilde{w}_1(\bm{X}_i)\}Z_i\{Y_i(1) - \mu_1(\bm{X}_i)\} - \{w_0^*(\bm{X}_i) - \tilde{w}_0(\bm{X}_i)\}(1-Z_i)\{Y_i(0) - \mu_0(\bm{X}_i)\} = o_p(n^{-1/2})
\end{align*}
such that 
\begin{align*}
    \Var(S_i + R_{1i}) = \Var[\tilde{w}_1(\bm{X}_i)Z_i\{Y_i(1) - \mu_1(\bm{X}_i)\} - \tilde{w}_0(\bm{X}_i)(1-Z_i)\{Y_i(0) - \mu_0(\bm{X}_i)\} + \{\tau(\bm{X}_i) - \tau\}].
\end{align*}

Then,
\begin{align*}
   \Var(S_i + R_{1_i}) &= \Var[\tilde{w}_1(\bm{X}_i)Z_i\{Y_i(1) - \mu_1(\bm{X}_i)\}] + \Var[\tilde{w}_0(\bm{X}_i)(1-Z_i)\{Y_i(0) - \mu_0(\bm{X}_i)\}] \\
   &\quad+ \Var\{\tau(\bm{X}_i) - \tau\} \\
&\quad- 2\Cov[\tilde{w}_1(\bm{X}_i)Z_i\{Y_i(1) - \mu_1(\bm{X}_i)\}, \tilde{w}_0(\bm{X}_i)(1-Z_i)\{Y_i(0) - \mu_0(\bm{X}_i)\}] \\
   &\quad+ 2\Cov[\tilde{w}_1(\bm{X}_i)Z_i\{Y_i(1) - \mu_1(\bm{X}_i)\}, \{\tau(\bm{X}_i) - \tau\}] \\
   &\quad- 2\Cov[\tilde{w}_0(\bm{X}_i)(1-Z_i)\{Y_i(0) - \mu_0(\bm{X}_i)\}, \{\tau(\bm{X}_i) - \tau\}].
\end{align*}
We have that
\begin{align*}
    - 2\Cov[\tilde{w}_1(\bm{X}_i)Z_i\{Y_i(1) - &\mu_1(\bm{X}_i)\}, \tilde{w}_0(\bm{X}_i)(1-Z_i)\{Y_i(0) - \mu_0(\bm{X}_i)\}]  \\
    &= -2E[\tilde{w}_1(\bm{X}_i)Z_i\{Y_i(1) - \mu_1(\bm{X}_i)\}] \\
    &\quad\times E[\tilde{w}_0(\bm{X}_i)(1-Z_i)\{Y_i(0) - \mu_0(\bm{X}_i)\}]] \\
    &= -2E(E[\tilde{w}_1(\bm{X}_i)Z_i\{Y_i(1) - \mu_1(\bm{X}_i)\}|\bm{X}_i])\\
    &\quad\times E(E[\tilde{w}_0(\bm{X}_i)(1-Z_i)\{Y_i(0) - \mu_0(\bm{X}_i)\}|\bm{X}_i]) \\
    &= -2E(\tilde{w}_1(\bm{X}_i)E[Z_i\{Y_i(1) - \mu_1(\bm{X}_i)\}|\bm{X}_i]) \\
    &\quad\times E(\tilde{w}_0(\bm{X}_i)E[(1-Z_i)\{Y_i(0) - \mu_0(\bm{X}_i)\}|\bm{X}_i]) \\
    \text{(by ignorability)} \:\: &= -2E[\tilde{w}_1(\bm{X}_i)e(\bm{X}_i)\{\mu_1(\bm{X}_i) - \mu_1(\bm{X}_i)\}] \\
    &\quad\times E[\tilde{w}_0(\bm{X}_i)\{1-e(\bm{X}_i)\}\{\mu_0(\bm{X}_i) - \mu_0(\bm{X}_i)\}]  \\
    &= 0 \times 0 = 0.
\end{align*}
By a similar argument, $E[\tilde{w}_1(\bm{X}_i)Z_i\{Y_i(1) - \mu_1(\bm{X}_i)\}]E[\{\tau(\bm{X}_i) - \tau\}] = E[\tilde{w}_0(\bm{X}_i)(1-Z_i)\{Y_i(0) - \mu_0(\bm{X}_i)\}]E[\{\tau(\bm{X}_i) - \tau\}] = 0$ so it remains to show that $E[\tilde{w}_1(\bm{X}_i)Z_i\{Y_i(1) - \mu_1(\bm{X}_i)\}\{\tau(\bm{X}_i) - \tau\}] = E[\tilde{w}_0(\bm{X}_i)(1-Z_i)\{Y_i(0) - \mu_0(\bm{X}_i)\}\{\tau(\bm{X}_i) - \tau\}] = 0$. Now,
\begin{align*}
   E[\tilde{w}_1(\bm{X}_i)Z_i\{Y_i(1) - \mu_1(\bm{X}_i)\}\{\tau(\bm{X}_i) - \tau\}] &= E(E[\tilde{w}_1(\bm{X}_i)Z_i\{Y_i(1) - \mu_1(\bm{X}_i)\}\{\tau(\bm{X}_i) - \tau\}|\bm{X}_i]) \\
   &= E(\tilde{w}_1(\bm{X}_i)\{\tau(\bm{X}_i) - \tau\}E[Z_i\{Y_i(1) - \mu_1(\bm{X}_i)\}|\bm{X}_i]) \\
   \text{(by ignorability)} &= E[\tilde{w}_1(\bm{X}_i)\{\tau(\bm{X}_i) - \tau\}e(\bm{X}_i)\{\mu_1(\bm{X}_i) - \mu_1(\bm{X}_i)\}] \\
   &= 0
\end{align*}
and by a similar argument $E[\tilde{w}_0(\bm{X}_i)(1-Z_i)\{Y_i(0) - \mu_0(\bm{X}_i)\}\{\tau(\bm{X}_i) - \tau\}] = 0$ such that all covariance terms are zero. Next,
\begin{align*}
    \Var[\tilde{w}_1(\bm{X}_i)Z_i\{Y_i(1) - \mu_1(\bm{X}_i)\}] &= E(\Var[\tilde{w}_1(\bm{X}_i)Z_i\{Y_i(1) - \mu_1(\bm{X}_i)\}|\bm{X}_i]) \\
    &\quad+ \Var(E[\tilde{w}_1(\bm{X}_i)Z_i\{Y_i(1) - \mu_1(\bm{X}_i)\}|\bm{X}_i]) \\
    \text{(by similar argument to above)}\:\: &= E(\Var[\tilde{w}_1(\bm{X}_i)Z_i\{Y_i(1) - \mu_1(\bm{X}_i)\}|\bm{X}_i]) + \Var(0) \\
    &= E(\tilde{w}_1(\bm{X}_i)^2\Var[Z_i\{Y_i(1) - \mu_1(\bm{X}_i)\}|\bm{X}_i]).
\end{align*}
Now, note that 
\begin{align*}
    \Var\left[Z_i\{Y_i(1) - \mu_1(\bm{X}_i)\}|\bm{X}_i\right] &= \Var(Z_i |\bm{X}_i)\Var[\{Y_i(1) - \mu_1(\bm{X}_i)\}|\bm{X}_i] \\
    &\quad+ \Var(Z_i |\bm{X}_i) E[\{Y_i(1) - \mu_1(\bm{X}_i)\}|\bm{X}_i]^2 \\
    &\quad+ \Var[\{Y_i(1) - \mu_1(\bm{X}_i)\}|\bm{X}_i]E[Z_i|\bm{X}_i]^2 \\
    &= e(\bm{X}_i)\{1-e(\bm{X}_i)\}\Var\{Y_i(1)|\bm{X}_i\} \\
    &\quad+ e(\bm{X}_i)^2\{1-e(\bm{X}_i)\}\Var\{Y_i(1)|\bm{X}_i\} \\
    &= e(\bm{X}_i)(1-e(\bm{X}_i) + e(\bm{X}_i))\Var\{Y_i(1)|\bm{X}_i\} \\
    &= e(\bm{X}_i)\Var\{Y_i(1)|\bm{X}_i\}.
\end{align*}
Then,
\begin{align*}
    \Var[\tilde{w}_1(\bm{X}_i)Z_i\{Y_i(1) - \mu_1(\bm{X}_i)\}] &= E[\tilde{w}_1(\bm{X}_i)^2e(\bm{X}_i)\Var\{Y_i(1)|\bm{X}_i\}]
\end{align*}
and by a similar argument,
\begin{align*}
    \Var[\tilde{w}_0(\bm{X}_i)(1-Z_i)\{Y_i(0) - \mu_0(\bm{X}_i)\}] &= E[\tilde{w}_0(\bm{X}_i)^2\{1-e(\bm{X}_i)\}\Var\{Y_i(0)|\bm{X}_i\}].
\end{align*}
Finally, we have that
\begin{align*}
    \Var\{\tau(\bm{X}_i) - \tau\} &= E[\{\tau(\bm{X}_i) - \tau\}^2] - E[\tau(\bm{X}_i) - \tau]^2 \\
    &= E[\{\tau(\bm{X}_i) - \tau\}^2] - (E[\tau(\bm{X}_i)] - \tau)^2 \\
    &= E[\{\tau(\bm{X}_i) - \tau\}^2].
\end{align*}
Then, we have that the asymptotic variance of the estimator when the propensity score model is misspecified is
\begin{align*}
    E[\tilde{w}_1(\bm{X}_i)^2e(\bm{X}_i)\Var\{Y_i(1)|\bm{X}_i\} + \tilde{w}_0(\bm{X}_i)^2\{1-e(\bm{X}_i)\}\Var\{Y_i(0)|\bm{X}_i\} + \{\tau(\bm{X}_i) - \tau\}^2].
\end{align*}
\end{proof}

\newpage
\subsection{Corollary \ref{theorem_g_new_var}}
\begin{proof}

Assume that $\mu_z(\bm{x})$ and $\tau(\bm{x})$ are linear in both $c(\bm{x})$ and $g(\bm{x})$ such that $\mu_1(\bm{x}) = \bm{\beta}_1^Tc(\bm{x}) + \bm{\lambda}_1^Tg(\bm{x})$ and $\mu_0(\bm{x}) = \bm{\beta}_0^Tc(\bm{x}) + \bm{\lambda}_0^Tg(\bm{x})$ for some $\bm{\beta}_z \in \mathbb{R}^K$ and $\bm{\lambda}_z \in \mathbb{R}^L$ where $\bm{\beta}_0 \neq \bm{\beta}_1$ and $\bm{\lambda}_0 \neq \bm{\lambda}_1$. Instead of targeting the ATE, $\tau$, $\hat{\tau}_{w^*}$ targets $\tau_{g, \tilde{w}} = E[E[Y|Z = 1, c(\bm{X}), \tilde{w}(\bm{X})e(\bm{X})g(\bm{X_i})]] - E[E[Y|Z = 0, c(\bm{X}), \tilde{w}(\bm{X})\{1 - e(\bm{X})\}g(\bm{X})]] = E[\tau_{g, \tilde{w}}(\bm{X}, Z)]$ where $\tau_{g, \tilde{w}}(\bm{x}, z) = (\bm{\beta}_1 - \bm{\beta}_0)^Tc(\bm{x})  + \bm{\lambda}_1^T\tilde{w}_1(\bm{x})zg(\bm{x}) - \bm{\lambda}_0^T\tilde{w}_0(\bm{x})(1-z)g(\bm{x})$. 
We have the following decomposition of estimator error:
\begin{align*}
    \hat{\tau}_{w^*} - \tau_{g, \tilde{w}} &= \frac1n\sum_{i=1}^n w^*_1(\bm{X}_i)Z_iY_i - \frac1n\sum_{i=1}^n w^*_0(\bm{X}_i)(1-Z_i)Y_i -\tau_{g, \tilde{w}}\\
    &= \frac1n\sum_{i=1}^n \tilde{w}_1(\bm{X}_i)Z_i\{Y_i(1) - \mu_1(\bm{X}_i)\} - \tilde{w}_0(\bm{X}_i)(1-Z_i)\{Y_i(0) - \mu_0(\bm{X}_i)\} - \tau_{g, \tilde{w}} \\
    &\quad + \frac1n\sum_{i=1}^n\{w_1^*(\bm{X}_i) - \tilde{w}_1(\bm{X}_i)\}Z_i\{Y_i(1) - \mu_1(\bm{X}_i)\} \\
    &\quad -\frac1n\sum_{i=1}^n \{w_0^*(\bm{X}_i) - \tilde{w}_0(\bm{X}_i)\}(1-Z_i)\{Y_i(0) - \mu_0(\bm{X}_i)\} \\
    &\quad + \frac1n\sum_{i=1}^n w^*_1(\bm{X}_i)Z_i\mu_1(\bm{X}_i)  - \frac1n\sum_{i=1}^n w^*_0(\bm{X}_i)(1-Z_i)\mu_0(\bm{X}_i) \\
    &= \frac1n\sum_{i=1}^n \tilde{w}_1(\bm{X}_i)Z_i\{Y_i(1) - \mu_1(\bm{X}_i)\} - \tilde{w}_0(\bm{X}_i)(1-Z_i)\{Y_i(0) - \mu_0(\bm{X}_i)\} \\
    &\quad + \frac1n\sum_{i=1}^n \bm{\beta}_1^Tc(\bm{X}_i) - \bm{\beta}_0^Tc(\bm{X}_i) \\ &\quad\quad +E[\bm{\lambda}_1^T\tilde{w}_1(\bm{X})Zg(\bm{X})] - E[\bm{\lambda}_0^T\tilde{w}_0(\bm{X})(1-Z)g(\bm{X})] - \tau_{g, \tilde{w}} \\
    &\quad + \frac1n\sum_{i=1}^n\{w_1^*(\bm{X}_i) - \tilde{w}_1(\bm{X}_i)\}Z_i\{Y_i(1) - \mu_1(\bm{X}_i)\} \\
    &\quad - \frac1n\sum_{i=1}^n\{w_0^*(\bm{X}_i) - \tilde{w}_0(\bm{X}_i)\}(1-Z_i)\{Y_i(0) - \mu_0(\bm{X}_i)\} \\
    &\quad + \frac1n\sum_{i=1}^n \bm{\lambda}_1^Tw^*_1(\bm{X}_i)Z_ig(\bm{X}_i) - E[\bm{\lambda}_1^T\tilde{w}_1(\bm{X})Zg(\bm{X})] \\
    &\quad - \frac1n\sum_{i=1}^n \bm{\lambda}_0^Tw^*_0(\bm{X}_i)(1-Z_i)g(\bm{X}_i) 
    + E[\bm{\lambda}_0^T\tilde{w}_0(\bm{X})(1-Z)g(\bm{X})].
\end{align*}
By Assumption \ref{balanc_func_regularity} and the assumption that $\rho''\{\cdot\}$ and $\rho'''\{\cdot\}$ are continuous over $\mathbb{R}$, the results of Theorems 5.28 and 7.2 in \cite{boos_essential_2013} apply. Thus, there exist some $q_1(\bm{X}_i, Z_i)$ and $q_0(\bm{X}_i, Z_i)$ with $E[q_z(\bm{X}_i, Z_i)] = 0$, $z=0,1$ such that
\begin{align*}
    \frac1n\sum_{i=1}^n \bm{\lambda}_1^Tw^*_1(\bm{X}_i)Z_ig(\bm{X}_i) - E[\bm{\lambda}_1^T\tilde{w}_1(\bm{X}_i)Z_ig(\bm{X}_i)] &= \frac1n\sum_{i=1}^n q_1(\bm{X}_i, Z_i) + R_2 \\
    \frac1n\sum_{i=1}^n \bm{\lambda}_0^Tw^*_0(\bm{X}_i)(1-Z_i)g(\bm{X}_i)-  E[\bm{\lambda}_0^T\tilde{w}_0(\bm{X}_i)(1-Z_i)g(\bm{X}_i)]&= \frac1n\sum_{i=1}^n q_0(\bm{X}_i, Z_i) + R_3
\end{align*}
where $R_2, R_3 = o_p(n^{-1/2})$. Note that we know the exact form of these two functions, $q_1(\bm{X}_i, Z_i)$ and $q_0(\bm{X}_i, Z_i)$, as specified in Theorem 5.28 in \cite{boos_essential_2013}.  Then $\hat{\tau}_w - \tau_{g, \tilde{w}} = S_i + R_1 + R_2 + R_3$ where
\begin{align*}
    S_i &= \tilde{w}_1(\bm{X}_i)Z_i\{Y_i(1) - \mu_1(\bm{X}_i)\} - \tilde{w}_0(\bm{X}_i)(1-Z_i)\{Y_i(0) - \mu_0(\bm{X}_i)\} \\
    &\quad+ \bm{\beta}_1^Tc(\bm{X}_i) - \bm{\beta}_0^Tc(\bm{X}_i) - E[\bm{\beta}_1^Tc(\bm{X}_i) - \bm{\beta}_0^Tc(\bm{X}_i)] + q_1(\bm{X}_i, Z_i) - q_0(\bm{X}_i, Z_i) \\
    R_1 &= \frac1n\sum_{i=1}^n\{w_1^*(\bm{X}_i) - \tilde{w}_1(\bm{X}_i)\}Z_i\{Y_i(1) - \mu_1(\bm{X}_i)\} \\ 
    &\quad- \frac1n\sum_{i=1}^n\{w_0^*(\bm{X}_i) - \tilde{w}_0(\bm{X}_i)\}(1-Z_i)\{Y_i(0) - \mu_0(\bm{X}_i)\} 
\end{align*}
where $R_1$ has previously been shown to be $o_p(n^{-1/2})$. Then, 
\begin{align*}
    \Var(\hat{\tau}_w - \tau_{g, \tilde{w}}) &= \Var(S_i) \\
    &= E[\tilde{w}_1(\bm{X}_i)^2e(\bm{X}_i)\Var\{Y_i(1)|\bm{X}_i\} + \tilde{w}_0(\bm{X}_i)^2\{1-e(\bm{X}_i)\}\Var\{Y_i(0)|\bm{X}_i\}] \\
    &\quad+ E\left[\left\{\bm{\beta}_1^Tc(\bm{X}_i) - \bm{\beta}_0^Tc(\bm{X}_i) - E[\bm{\beta}_1^Tc(\bm{X}_i) - \bm{\beta}_0^Tc(\bm{X}_i)] \right.\right. 
    \\ &\quad\quad\quad\quad \left.\left. + q_1(\bm{X}_i, Z_i) - q_0(\bm{X}_i, Z_i)\right\}^2\right]
\end{align*}
as all covariance terms cancel as shown in Corollary \ref{theorem:new_var}. It remains to determine $q_1(\bm{X}_i, Z_i)$ and $q_0(\bm{X}_i, Z_i).$ By Theorem 5.28 in \cite{boos_essential_2013},
\begin{align*}
    q_1(\bm{X}_i, Z_i) &= \bm{\lambda}_1^T\tilde{w}_1(\bm{X}_i)Z_ig(\bm{X}_i) - E[\bm{\lambda}_1^T\tilde{w}_1(\bm{X}_i)Z_ig(\bm{X}_i)] + r_1(\bm{X}_i, Z_i) \:\:\: \text{and}\\
    q_0(\bm{X}_i, Z_i) &= \bm{\lambda}_0^T\tilde{w}_0(\bm{X}_i)(1-Z_i)g(\bm{X}_i) - E[\bm{\lambda}_0^T\tilde{w}_0(\bm{X}_i)(1-Z_i)g(\bm{X}_i)] + r_0(\bm{X}_i, Z_i),
\end{align*}
where $r_z(\bm{X}_i, Z_i)$ depend on the asymptotic distribution of $w^*_1(\bm{X}_i)$ and $w^*_0(\bm{X}_i)$. Therefore, to derive $r_z(\bm{X}_i, Z_i)$, we first derive the asymptotic distribution of $\hat{\bm{\alpha}}_1, \hat{\bm{\alpha}}_0, \hat{\bm{\gamma}}$ using the theory of M-estimation.

Let $m_k = E[c_k(\bm{X}_i)]$ and $\bm{m} = (m_1, \ldots m_K)$. Let $\bm{\theta} = (\bm{m}, \bm{\alpha}_1, \bm{\alpha}_0, \bm{\gamma})$. Let
\begin{align*}
    \phi_k(\bm{X}, \bm{m}) &= c_k(\bm{X}) - m_k  \\
    \varphi^{(1)}_k(\bm{X}, Z, \bm{m}, \bm{\alpha}_1, \bm{\gamma}) &= \rho'\{\bm{\alpha}_1^{T}c(\bm{X}) + \bm{\gamma}^{T}g(\bm{X})\}Zc_k(\bm{X}) - m_k  \\
     \varphi^{(0)}_k(\bm{X}, Z, \bm{m}, \bm{\alpha}_0, \bm{\gamma}) &= \rho'\{\bm{\alpha}_0^{T}c(\bm{X}) - \bm{\gamma}^{T}g(\bm{x})\}(1-Z)c_k(\bm{X}) - m_k \\
     \varphi^{(2)}_l(\bm{X}, Z, \bm{\alpha}_0, \bm{\alpha}_1 \bm{\gamma})&=\rho'\{\bm{\alpha}_1^{T}c(\bm{X}) + \bm{\gamma}^{T}g(\bm{X})\}Zg_l(\bm{X}) - \rho'\{\bm{\alpha}_0^{T}c(\bm{X}) - \bm{\gamma}^{T}g(\bm{X})\}(1-Z)g_l(\bm{X}) \\
\end{align*}
Let $\psi(\bm{X}, Z, \bm{\theta}) = (\phi, \varphi^{(1)}, \varphi^{(0)}, \varphi^{(2)})$. Then $\hat{\alpha}_0, \hat{\alpha}_1, \hat{\gamma}$ is the solution to 
\begin{align*}
    \frac1n\sum_{i=1}^n \psi(\bm{X}_i, Z_i, \bm{\theta}) = 0.
\end{align*}
Given the assumptions on $\rho\{\cdot\}$, Assumption \ref{balanc_func_regularity}, and similar arguments to the proof of Lemma \ref{lemma}, by the theory of M-estimation and Theorem 7.2 in \cite{boos_essential_2013},
\begin{align*}
    \hat{\bm{\theta}} - \tilde{\bm{\theta}} &= \sum_{i=1}^n \bm{A}^{-1}(\tilde{\bm{\theta}})\psi(\bm{X}_i, Z_i, \bm{\tilde{\theta})} + R_4
\end{align*}
where $R_4 = o_p(n^{-1/2})$. We will now determine $\bm{A}(\tilde{\bm{\theta}})$. For ease in notation, let $h'_1 = \rho''\{\tilde{\bm{\alpha}}_1^{T}c(\bm{x}) + \tilde{\bm{\gamma}}^{T}g(\bm{x})\}Z$ and $h'_0 = \rho''\{\tilde{\bm{\alpha}}_0^{T}c(\bm{x}) - \tilde{\bm{\gamma}}^{T}g(\bm{x})\}(1-Z)$ and $\bm{H}_{\bm{a},\bm{b}} = E(\bm{ab}^T) = \Cov(\bm{a},\bm{b}) + E[\bm{a}]E[\bm{b}]^T$.
Then,
\begin{align*}
    A(\tilde{\bm{\theta}}) &= E\begin{bmatrix}
    \bm{I}_K & \bm{0}_K & \bm{0}_K & \bm{0}_{K\times L}\\
    \bm{I}_K & -h'_1c(\bm{X})c(\bm{X})^T & \bm{0}_K & -h'_1c(\bm{X})g(\bm{X})^T\\
    \bm{I}_K & \bm{0}_K & -h'_0c(\bm{X})c(\bm{X})^T & -h'_0c(\bm{X})g(\bm{X})^T\\
    \bm{0}_{L \times K} & -h'_1g(\bm{X})c(\bm{X})^T &h'_0g(\bm{X})c(\bm{X})^T  & \{-h'_1 -h'_0\}g(\bm{X})g(\bm{X})^T \\
    \end{bmatrix} \\[12pt]
    &= \begin{bmatrix}
    \bm{I}_K & \bm{0}_K & \bm{0}_K & \bm{0}_{K \times L}\\
    \bm{I}_K & -\bm{H}_{h'_1c(\bm{X}),c(\bm{X})} & \bm{0}_K & -\bm{H}_{h'_1c(\bm{X}),g(\bm{X})}\\
    \bm{I}_K & \bm{0}_K & -\bm{H}_{h'_0c(\bm{X}),c(\bm{X})}  & -\bm{H}_{h'_0c(\bm{X}),g(\bm{X})}\\
    \bm{0}_{L \times K} & -\bm{H}_{h'_1g(\bm{X}),c(\bm{X})}  & \bm{H}_{h'_0g(\bm{X}),c(\bm{X})}  & -\bm{H}_{(h'_1+ h'_0)g(\bm{X}),g(\bm{X})}  \\
    \end{bmatrix}
\end{align*}
Symbolically inverting $\bm{A}(\tilde{\bm{\theta}})$ is non-trivial, so we proceed with $\bm{A}^{-1}(\tilde{\bm{\theta}})\psi(\bm{X}, Z, \bm{\tilde{\theta})}$ rather than the simplified form. Let $\bm{A}^{-1}(\tilde{\bm{\theta}})\psi(\bm{X}, Z, \bm{\tilde{\theta})}_{\tilde{\bm{\delta}}}$ be the rows of $\bm{A}^{-1}(\tilde{\bm{\theta}})\psi(\bm{X}, Z, \bm{\tilde{\theta})}$ corresponding to any $\tilde{\bm{\delta}} \subset \tilde{\bm{\theta}}$. Then, by Theorem 5.28 in \cite{boos_essential_2013},
\begin{align*}
    q_1(\bm{X}_i, Z_i) &= \bm{\lambda}_1^T\tilde{w}_1(\bm{X}_i)Z_ig(\bm{X}_i) - E[\bm{\lambda}_1^T\tilde{w}_1(\bm{X}_i)Z_ig(\bm{X}_i)] \\
    &\quad+ E\left[\frac{\partial}{\partial (\tilde{\bm{\alpha}}_1, \tilde{\bm{\gamma}})}\bm{\lambda}_1^T\tilde{w}_1(\bm{X}_i)Z_ig(\bm{X}_i)\right]\bm{A}^{-1}(\tilde{\bm{\theta}})\psi(\bm{X}_i, Z_i, \bm{\tilde{\theta})}_{(\tilde{\bm{\alpha}}_1, \tilde{\bm{\gamma}})} \\[6pt]
    q_0(\bm{X}_i, Z_i) &= \bm{\lambda}_0^T\tilde{w}_0(\bm{X}_i)(1-Z_i)g(\bm{X}_i) - E[\bm{\lambda}_0^T\tilde{w}_0(\bm{X}_i)(1-Z_i)g(\bm{X}_i)] \\
    &\quad+ E\left[\frac{\partial}{\partial (\tilde{\bm{\alpha}}_0, \tilde{\bm{\gamma}})}\bm{\lambda}_0^T\tilde{w}_0(\bm{X}_i)(1-Z_i)g(\bm{X}_i)\right]\bm{A}^{-1}(\tilde{\bm{\theta}})\psi(\bm{X}_i, Z_i, \bm{\tilde{\theta})}_{(\tilde{\bm{\alpha}}_0, \tilde{\bm{\gamma}})},
\end{align*}
where 
\begin{align*}
    E\left[\frac{\partial}{\partial (\tilde{\bm{\alpha}}_1, \tilde{\bm{\gamma}})}\bm{\lambda}_1^T\tilde{w}_1(\bm{X}_i)Z_ig(\bm{X}_i)\right] &= [\bm{\lambda}_1^T\bm{H}_{h'_1g(\bm{X}),c(\bm{X})}, \bm{\lambda}_1^T\bm{H}_{h'_1g(\bm{X}),g(\bm{X})}] \text{ and} \\[6pt]
    E\left[\frac{\partial}{\partial (\tilde{\bm{\alpha}}_0, \tilde{\bm{\gamma}})}\bm{\lambda}_0^T\tilde{w}_0(\bm{X}_i)(1-Z_i)g(\bm{X}_i)\right] &= [\bm{\lambda}_0^T\bm{H}_{h'_0g(\bm{X}),c(\bm{X})}, -\bm{\lambda}_0^T\bm{H}_{h'_0g(\bm{X}),g(\bm{X})}].
\end{align*}
 Then, we have that
\begin{align*}
    \Var(\hat{\tau}_w - \tau_{g, \tilde{w}}) &= E[\tilde{w}_1(\bm{X}_i)^2e(\bm{X}_i)\Var\{Y_i(1)|\bm{X}_i\} + \tilde{w}_0(\bm{X}_i)^2\{1-e(\bm{X}_i)\}\Var\{Y_i(0)|\bm{X}_i\}] \\
    &\quad+ E[\{\tau_{g, \tilde{w}}(\bm{X}_i, Z_i) - \tau_{g, \tilde{w}} + r_1(\bm{X}_i, Z_i) - r_0(\bm{X}_i, Z_i)\}^2] \\
    &= V_{g, \tilde{w}}
\end{align*}
for $r_1(\bm{X}_i, Z_i) = [\bm{\lambda}_1^T\bm{H}_{h'_1g(\bm{X}),c(\bm{X})}, \bm{\lambda}_1^T\bm{H}_{h'_1g(\bm{X}),g(\bm{X})}]\bm{A}^{-1}(\tilde{\bm{\theta}})\psi(\bm{X}_i, Z_i, \bm{\tilde{\theta})}_{(\tilde{\bm{\alpha}}_1, \tilde{\bm{\gamma}})}$ and $r_0(\bm{X}_i, Z_i) = [\bm{\lambda}_0^T\bm{H}_{h'_0g(\bm{X}),c(\bm{X})}, -\bm{\lambda}_0^T\bm{H}_{h'_0g(\bm{X}),g(\bm{X})}]\bm{A}^{-1}(\tilde{\bm{\theta}})\psi(\bm{X}_i, Z_i, \bm{\tilde{\theta})}_{(\tilde{\bm{\alpha}}_0, \tilde{\bm{\gamma}})}$. Thus, 
\begin{align*}
    \sqrt{n}(\hat{\tau}_w - \tau_{g, \tilde{w}}) \to_d N(0, V_{g, \tilde{w}})
\end{align*}
as desired.
\end{proof}

\section{Restrictions on the implied propensity score model}
\label{sec:implied_ps_supp}
We explore the restrictions on the implied propensity score model for two commonly used measures of dispersion: 1) entropy weights \citep{hainmueller_entropy_2012}; and 2) stable balancing weights \citep{zubizarreta_stable_2015}. The first order conditions in \eqref{eq:ps_conditions} imply the following condition:
 \begin{align*}
    (\rho')^{-1}(\{e(\bm{x})\}^{-1}) +  (\rho')^{-1}(\{1-e(\bm{x})\}^{-1}) &= (\bm{\alpha}_0 + \bm{\alpha}_1)^Tc(\bm{x}) 
 \end{align*}
Without a defining a specific measure of dispersion $D(\cdot)$ (and therefore a specific $\rho(\cdot)$) it is hard to connect this to $\Var(Z|\bm{X})$. However, we explore two common deviance measures \citep{wang_minimal_2020}:
\begin{enumerate}
    \item Entropy weights: $D(t) = t\log(t)$ $\implies$ $\rho'(t) = e^{-t-1}$  $\implies$ $(\rho')^{-1}(t) = -\log(t) - 1$
    \item Stable balancing weights: $D(t) = t^2$ $\implies$ $\rho'(x) = - \frac{t}{2} $ $\implies$ $(\rho')^{-1}(x) = -2t$
\end{enumerate}
 Consider the entropy measure of weight dispersion first. Applying this to the first order conditions in \eqref{eq:ps_conditions}, we get the following
 \begin{align*}
     \log\{e(\bm{x})\} = \bm{\alpha}_1^Tc(\bm{x}) + \bm{\gamma}^Tg(\bm{x}) + 1 \\
     \log\{1-e(\bm{x})\} = \bm{\alpha}_1^Tc(\bm{x}) - \bm{\gamma}^Tg(\bm{x}) + 1 \\  \\
     \log\{e(\bm{x})\} + \log\{1-e(\bm{x})\} = (\bm{\alpha}_0 + \bm{\alpha}_1)^Tc(\bm{x}) + 2  \\
     \log\{e(\bm{x})(1-e(\bm{x})\}= (\bm{\alpha}_0 + \bm{\alpha}_1)^Tc(\bm{x})  + 2 \\ 
     \Var(Z|\bm{X} = \bm{x}) = \text{exp}\{(\bm{\alpha}_0 + \bm{\alpha}_1)^Tc(\bm{x})  + 2 \}
 \end{align*}
 such that $\Var(Z|\bm{X} = \bm{x})$ only varies by $c(\bm{x})$ in the implied propensity score model.
 Now, consider the dispersion measure used for stable balancing weights:
 \begin{align*}
     \frac{1}{e(\bm{x})} &= -2\bm{\alpha}_1^Tc(\bm{x}) -2\bm{\gamma}^Tg(\bm{x})  \\
      \frac{1}{1-e(\bm{x})} &= -2\bm{\alpha}_0^Tc(\bm{x}) + 2\bm{\gamma}^Tg(\bm{x})  \\ \\
      \frac{1}{e(\bm{x})} + \frac{1}{1-e(\bm{x})}&= -2\bm{\alpha}_1^Tc(\bm{x}) -2\bm{\alpha}_0^Tc(\bm{x})  \\
       \frac{1}{e(\bm{x})(1-e(\bm{x}))}&= -2\bm{\alpha}_1^Tc(\bm{x}) -2\bm{\alpha}_0^Tc(\bm{x}) \\
       \Var(Z|\bm{X} = \bm{x}) &= \frac{1}{-2\bm{\alpha}_1^Tc(\bm{x}) -2\bm{\alpha}_0^Tc(\bm{x})}  
 \end{align*}
and again, $\Var(Z|\bm{X} = \bm{x})$ only varies by $c(\bm{x})$ in the implied propensity score model. While this relationship may not hold for every possible dispersion measure, for these two dispersion measures the implied propensity score model of Problem \eqref{eq:our_db} allows $E[Z|\bm{X} = \bm{x}]$ to vary with all functions of covariates but $\Var(Z|\bm{X} = \bm{x})$ to vary only with the set of covariates that are balanced to the population. Since this relationship does not hold for any arbitrary $c(\bm{x})$ and $g(\bm{x})$ sets, this model places substantial restrictions on what forms these sets can take for the implied propensity score model to be true. Given a specific $c(\bm{x})$ set, the $g(\bm{x})$ set that satisfies the implied propensity score model can be identified for each $D(\cdot)$ by using the first order conditions \eqref{eq:ps1}-\eqref{eq:ps3} to solve for $\bm{\gamma}^Tg(\bm{x})$ as a function of $\bm{\alpha}_z^Tc(\bm{x})$.

\pagebreak
\section{Connections to other estimation problems}
\label{sec:connect_methods_supp}
\subsection{Multiple treatment groups}
Consider the case where there are $M > 2$ treatment groups such that $Z_i = m$ for $m = 0, \ldots, M-1$ and where we want to estimate the average treatment effect between treatments $m_1,m_2 = 0, \ldots, M -1$, $\tau_{m_1, m_2} = E[\mu_{m_1}(\bm{x}) - \mu_{m_2}(\bm{x})]$. Adapting Problem \eqref{eq:our_db} to derive weights for estimating $\tau_{m_1, m_2}$ simply requires replacing $z=0,1$ with $z= m_1, m_2$ in Problem \eqref{eq:our_db}. However, if weights are derived in this way for a series of average treatment effects (e.g., $\tau_{m_1, m_2}, \tau_{m_2, m_3}, \tau_{m_1, m_3}$ etc.,), each average treatment effect will correspond to a different population for the covariate set $g(\bm{x})$ or even a different $g(\bm{x})$ set. To maintain the same $g(\bm{x})$ covariate population for average treatment effects corresponding all contrasts for $M_1 \leq M$ treatments, we would include $M_1$ constraints balancing $c(\bm{x})$ in treatment groups to the sample population (e.g., constraint \eqref{eq:db_b1}) and $M_1 - 1$ constraints balancing $g(\bm{x})$ between treatment groups (e.g., constraint \eqref{eq:db_b3}). As $M_1$ increases, generally the set $g(\bm{x})$ will need to be larger for the existence of a solution. However, for there to be no estimator error for every contrast, $g(\bm{x})$ must contain no effect modifiers for \textit{all} treatment groups. Thus, deriving weights for all contrasts separately may allow for a smaller $g(\bm{x})$ set, but this set and corresponding population may vary for each contrast. In contrast, deriving weights for all contrasts together will result in the same target population for all average treatment effect estimands but the set $g(\bm{x})$ may need to be large to derive the weights.

\subsection{Weighted average treatment effects}
Consider the estimand $E[h(\bm{X})\tau(\bm{X})]$ for some function $h(\bm{x})$. This as a weighted average treatment effect (WATE) where the population covariate density $f(\bm{x})$ is modified to $h(\bm{x})f(\bm{x})$. The ATO is the WATE when $h(\bm{x}) = e(\bm{x})\{1 - e(\bm{x})\}$. When adapting our proposed constrained optimization problem to derive weights for estimating the WATE, $\{w_i^{WATE}\}$, we can simply replace the right-hand side of constraint \eqref{eq:db_b1}  with $\sum_{i=1}^n h(\bm{X}_i)c_k(\bm{X}_i)/\sum_{i=1}^n h(\bm{X}_i)$. The resulting weights will balance $c(\bm{x})$ to the weighted population with density $h(\bm{x})f(\bm{x})$ while $g(\bm{x})$ will only be balanced between the treated and control groups. Similar to the ATT, all discussion and methods proposed in Sections \ref{sec:estimator_interp} and \ref{sec:implementation} can be used for the WATE. However, if the population defined by $h(\bm{x})f(\bm{x})$ has improved overlap compared to the sample population, the set $g(\bm{x})$ required to yield a solution will be smaller than the set required to yield a solution for the ATE. Thus, various $h(\bm{x})$ functions could be used in combination with our proposed constrained optimization problem to improve overlap and yield a solution. Yet, this will modify a large portion of the sample covariate population to the population described by $h(\bm{x})f(\bm{x})$ such that this modified population ideally is of specific interest to the given research question.

\subsection{Distributional balancing weights}
Direct balancing weights tend to correspond to estimators with reduced variance in part because they focus only on balancing linear combinations of the set of basis functions $b(\bm{x}).$ However, if the CATE is not only linear in this set of basis functions, the direct balancing weights estimator is generally biased for the ATE. This has motivated the development of distributional balancing weights which are derived to minimize the distance between each of the weighted treatment group empirical distributions and the population empirical distribution \citep{huling_energy_2024}. For distances that are integral probability metrics, the corresponding distributional balancing weights balance all functions in a given class (e.g, weights proposed by \cite{huling_energy_2024} balance all functions in a particular Sobolev space). Thus, these weights balance a much wider class of functions than direct balancing weights, which is desirable as $\mu_z(\bm{X})$ is often unknown. Let $\mathcal{D}$  a distributional distance measure, $F_n(\bm{x})$ be the empirical covariate distribution, and $F_{n, z, \bm{w}}(\bm{x})$ be the weighted empirical covariate distribution in treatment group $z$. Then distributional weights commonly take the follow form,
\begin{align*}
    \bm{w}_n = \text{argmin}_{\bm{w}} &\:\:\: \mathcal{D}\{F_{n, 0, \bm{w}}(\bm{x}), F_n(\bm{x})\} + \mathcal{D}\{F_{n, 1, \bm{w}}(\bm{x}), F_n(\bm{x})\} \label{eq:dist_bal_w}
\end{align*}
where $\mathcal{D}\{F_{n, 0, \bm{w}}(\bm{x}), F_{n, 1, \bm{w}}(\bm{x})\}$ can also be added to minimization objective. While there is always a solution to this optimization problem, when there is a lack of overlap, the corresponding estimator may have inflated variance. To address this, \cite{chen_robust_2024} propose weighting $\mathcal{D}\{F_{n, 0, \bm{w}}(\bm{x}), F_n(\bm{x})\} + \mathcal{D}\{F_{n, 1, \bm{w}}(\bm{x}), F_n(\bm{x})\}$ by a hyperparameter $0 \geq \alpha \geq 1$ and $\mathcal{D}\{F_{n, 0, \bm{w}}(\bm{x}), F_{n, 1, \bm{w}}(\bm{x})\}$ by $1-\alpha$; increasing the relative weight of $\mathcal{D}\{F_{n, 0, \bm{w}}(\bm{x}), F_{n, 1, \bm{w}}(\bm{x})\}$ tends to decrease variance but increase bias due to deviation from the target population. When $\alpha = 0$, the corresponding weights target the ATO. However, this results in the entire covariate population changing for any $\alpha > 0$. Thus, we propose the following distributional weights,
\begin{align*}
    \bm{w}_n = \text{argmin}_{\bm{w}} &\:\:\: \mathcal{D}[F_{n, 0, \bm{w}}\{c(\bm{x})\}, F_n\{c(\bm{x})\}] + \mathcal{D}[F_{n, 1, \bm{w}}\{c(\bm{x})\}, F_n\{c(\bm{x})\}] \\
    &+ \mathcal{D}\{F_{n, 0, \bm{w}}(\bm{x}), F_{n, 1, \bm{w}}(\bm{x})\},
\end{align*}
where the distribution of only a subset of the covariates (i.e., $c(\bm{x})$) is balanced to the empirical distribution, but all covariates are balanced between treated and control groups. In addition, $\mathcal{D}\{F_{n, 0, \bm{w}}(\bm{x}), F_{n, 1, \bm{w}}(\bm{x})\}$ could be replaced by $\mathcal{D}[F_{n, 0, \bm{w}}\{g(\bm{x})\}, F_{n, 1, \bm{w}}\{g(\bm{x})\}]$ if desired. Alternatively, \cite{mak_projected_2018} define distributional distances with a kernel that allows for the variable influence of each covariate on the distance. This could be used within a distributional balancing weights set-up to weaken the impact of the $g(\bm{x})$ on the distributional distance between the weighted treatment groups and the sample population.

\section{Simulation study on estimator variance}
\label{sec:supp_sim_var}
We generate data using the same methods described in Section \ref{sec:sim} and supplementary Section \ref{sec:sim_supp} with $\theta = 100\%$ of covariates being treatment effect modifiers. Thus, we generate a total of 36 simulation scenarios that are all combinations p = 20, 100, 20\% and 40\% treated, three levels of overlap,
and three levels of treatment effect heterogeneity. We compute a Monte Carlo estimate of the asymptotic variance of the direct balancing weights estimator (Corollary \ref{theorem:new_var}) and the asymptotic variance of our proposed estimator when $c(\bm{x})$ includes 75\% and 25\% of all covariates (Corollary \ref{theorem_g_new_var}). We also compute the empirical variance of these estimators when computing the weights with and without the additional constraints of $w_i \geq 0$ and $\sum_{i=1}^n w_iZ_i = \sum_{i=1}^n w_i(1-Z_i)$.

\begin{figure}
    \centering
    \includegraphics[width=0.9\linewidth]{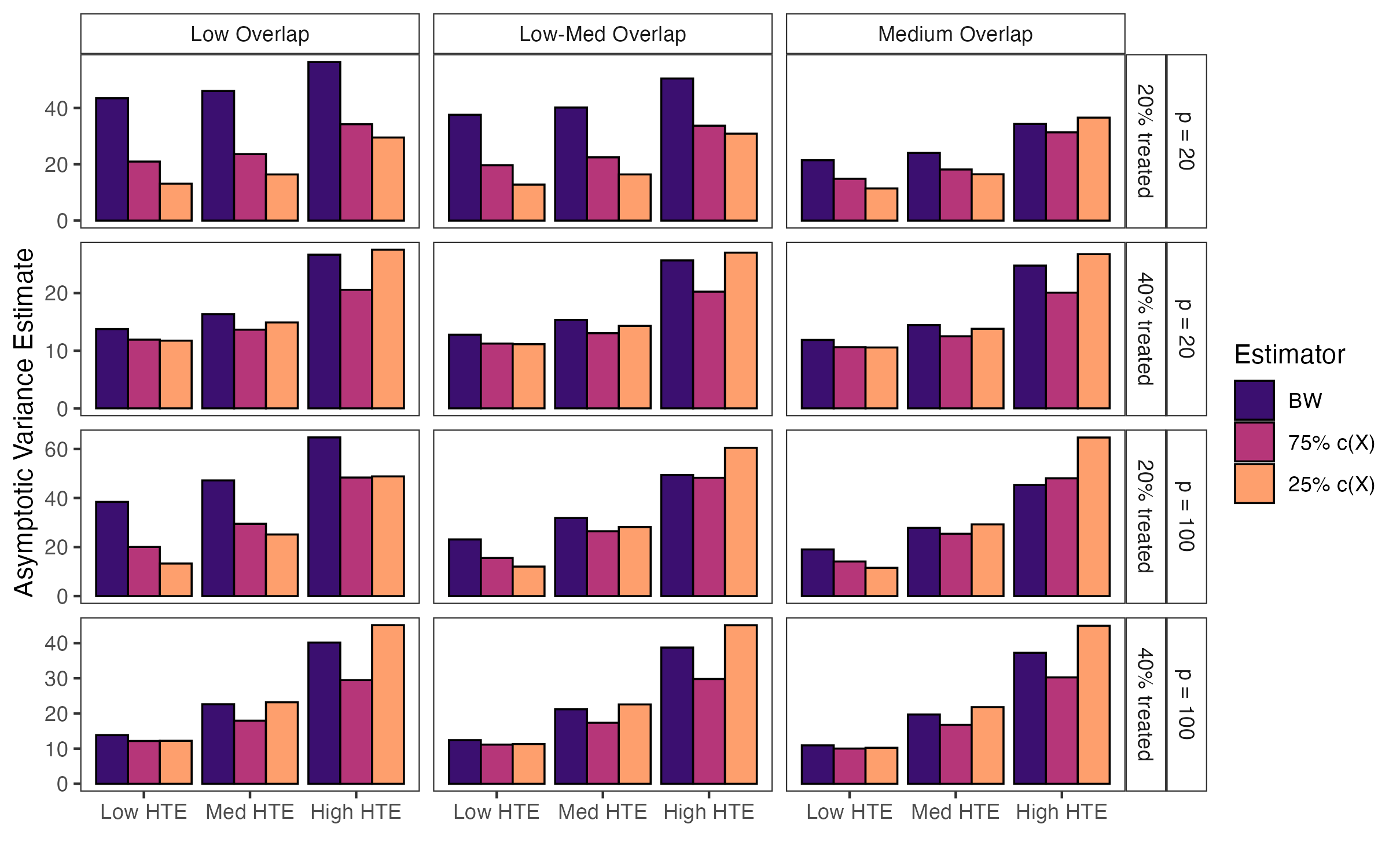}
    \caption{Estimated asymptotic variance estimates for all simulation scenarios. In the legend, ``X\% c(X)'' indicates that X\% of covariates are included in the $c(x)$ set for the PRTBW weighted estimators. The horizontal axis indicates levels of treatment effect heterogeneity while the simulation scenarios are labeled on the top and right of the panels.}
    \label{fig:asymp_var}
\end{figure}

Across a majority of scenarios, our proposed estimators have smaller estimated asymptotic variance than the direct balancing weights estimator (supplementary Figure \ref{fig:asymp_var}). However, when there is high treatment effect heterogeneity, our proposed estimator with only 25\% of covariates in the $c(\bm{x})$ set tends to have larger estimated asymptotic variance than the direct balancing weights estimator, likely due to the large variance of $\tau_{g, \tilde{w}}(\bm{X},Z)$ compared to $\tau(\bm{X})$. Despite this, our proposed estimator with 75\% of covariates in the $c(\bm{x})$ set has smaller asymptotic variance than the direct balancing weights estimator across almost all scenarios. Thus, while our proposed estimator will not always have smaller asymptotic variance than the direct balancing estimator, our proposed estimator does have similar or lower variance across many scenarios. The empirical estimator variance is similar the estimated asymptotic variance for all estimators across all scenarios which further validates the estimated asymptotic variances (supplementary Figure \ref{fig:asymp_emp_var}).

\begin{figure}
    \centering
    \includegraphics[width=0.9\linewidth]{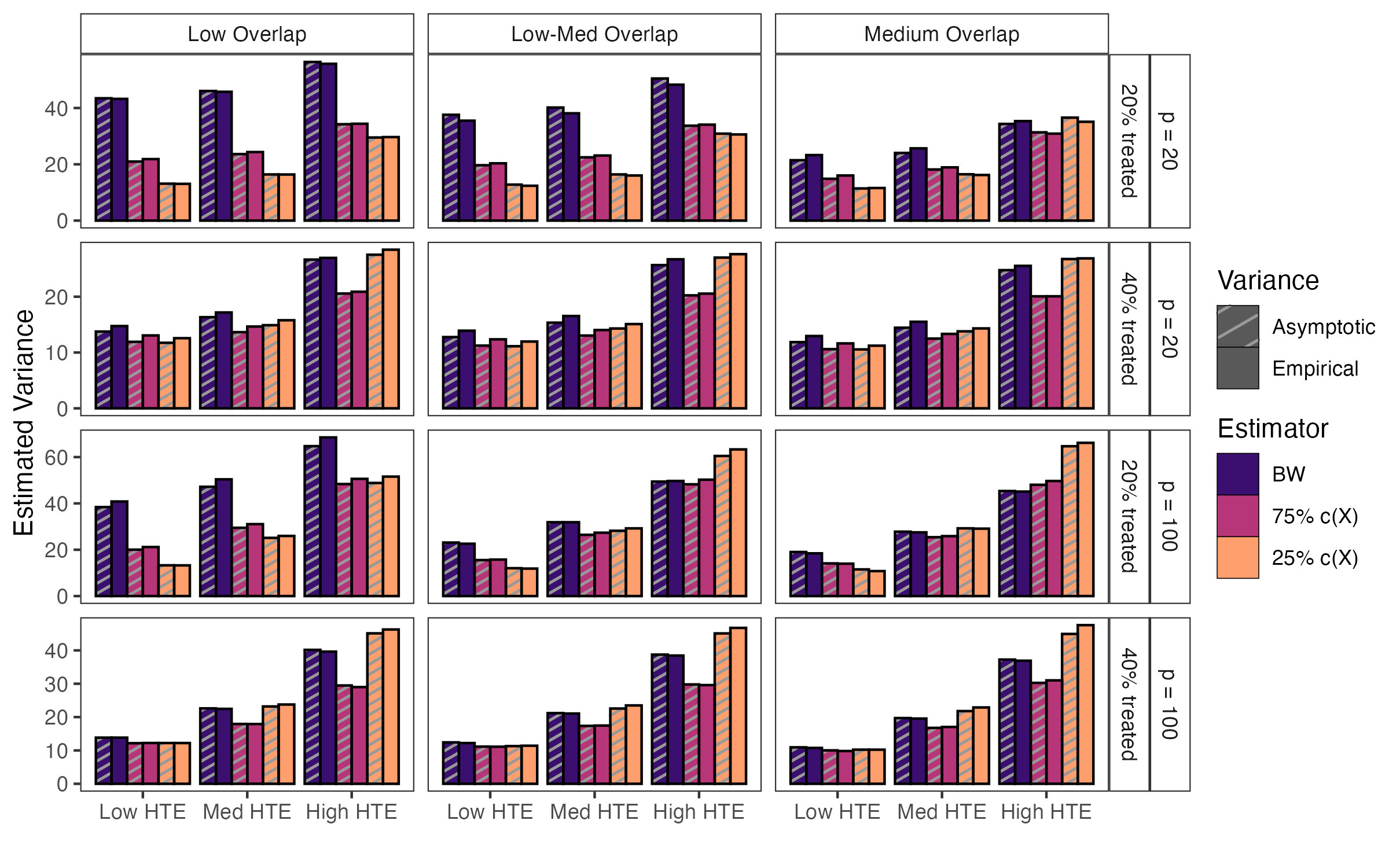}
    \caption{Estimated asymptotic variance estimates and empirical estimator variance for all simulation scenarios. In the legend, ``X\% c(X)'' indicates that X\% of covariates are included in the $c(x)$ set for the PRTBW weighted estimators. The horizontal axis indicates levels of treatment effect heterogeneity while the simulation scenarios are labeled on the top and right of the panels.}
    \label{fig:asymp_emp_var}
\end{figure}

While the asymptotic results do not directly apply with the additional constraints of $\sum_{i=1}^n w_iZ_i = \sum_{i=1}^n w_i(1-Z_i)$ to Problem \eqref{eq:our_db}, we can explore the behavior of the estimators in this case with empirical estimator variance. In many scenarios, there is no solution to the direct balancing optimization problem with these additional constraints such that empirical estimator variance cannot be computed (supplementary Figure \ref{fig:emp_sd_var}). Across all scenarios where direct balancing weights could be calculated, our proposed estimators have smaller or similar empirical variance to the direct balancing weights estimator; when overlap is poor our proposed estimators have substantially smaller empirical variance. Therefore, our proposed estimators tend to outperform the direct balancing weights estimator in terms of variance across almost all scenarios when the weights are constrained to not extrapolate.

\begin{figure}
    \centering
    \includegraphics[width=0.9\linewidth]{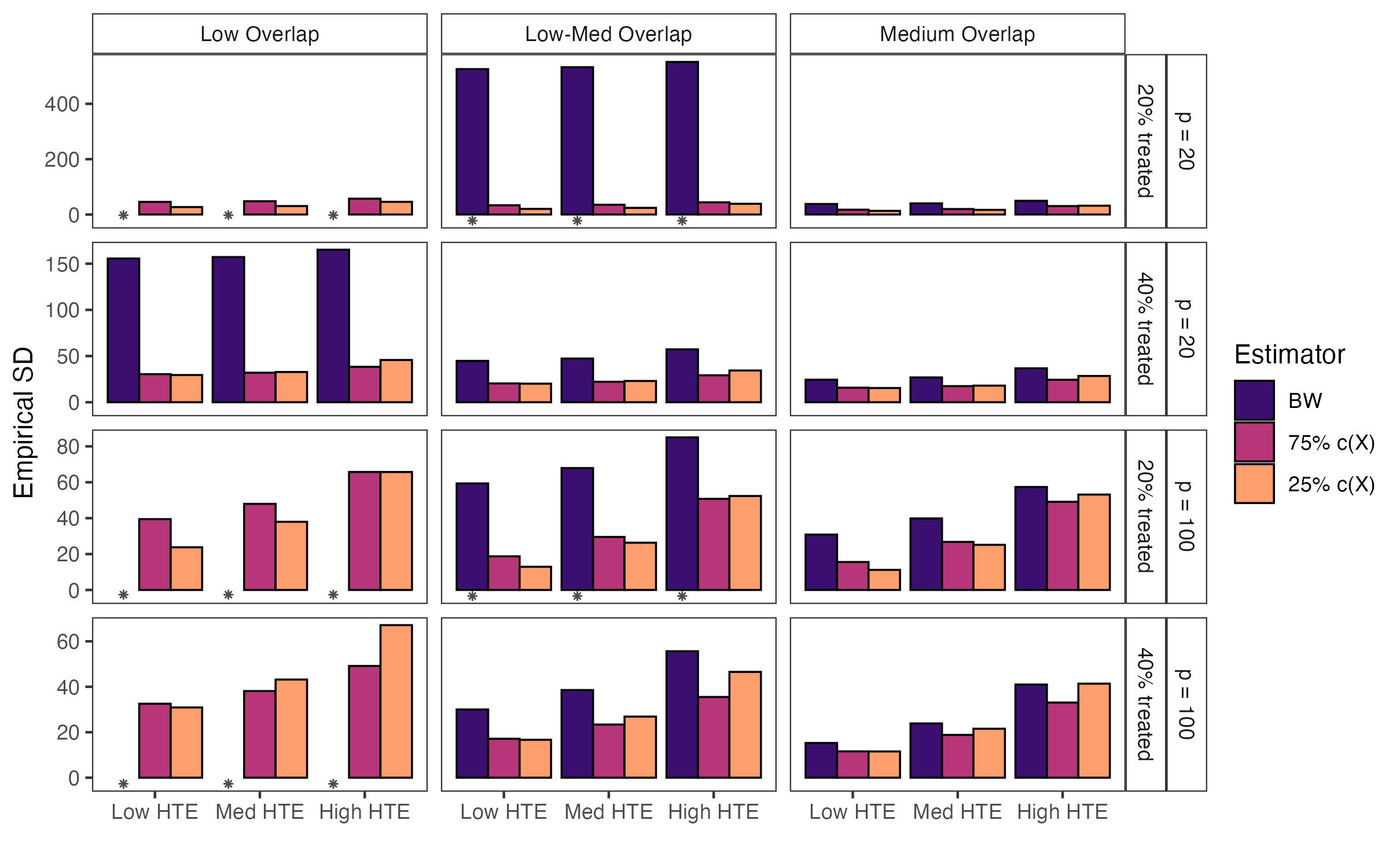}
    \caption{Empirical estimator variance for all simulation scenarios when including additional constraints of $w_i \geq 0$ and $\sum_{i=1}^n w_iZ_i = \sum_{i=1}^n w_i(1-Z_i)$ in Problem \eqref{eq:our_db}. In the legend, ``X\% c(X)'' indicates that X\% of covariates are included in the $c(x)$ set for the weighted PRTBW estimators. The asterisk, $*$, indicates estimators where a solution to the balancing weights optimization problem did not exist for all 1,000 datasets. The horizontal axis indicates levels of treatment effect heterogeneity while the simulation scenarios are labeled on the top and right of the panels.}
    \label{fig:emp_sd_var}
\end{figure}

\section{Additional guidance and implementation details}
\subsection{Guidance on selecting the \texorpdfstring{$g(\bm{x})$}{TEXT} set}
\label{sec:guidance_implement_supp}
Selecting the $g(\bm{x})$ set requires researchers to evaluate the trade-offs between estimator variance and estimator bias due to the modified $g(\bm{x})$ population for their particular application.  While these trade-offs are unavoidable when there are positivity violations, it is still desirable to select the $c(\bm{x})$ and $g(\bm{x})$ sets such that the corresponding estimator has minimal error with respect to an estimand of scientific interest, whether this the ATE or the average treatment effect for a modified population. Both subject matter knowledge as well as data driven techniques can be used to identify the $c(\bm{x})$ and $g(\bm{x})$ set; here, we provide further guidance on how a variety of factors and preferences may influence the selection of these two sets. 

As mentioned in the previous section, the corresponding estimator for Problem \eqref{eq:our_db} has minimal error with respect to the ATE when $c(\bm{x})$ contains all treatment effect modifiers effect modifiers. When the ATE is of scientific interest, it is therefore desirable to determine the set $g(\bm{x})$ that 1) yields a solution to Problem \eqref{eq:our_db}; and 2) minimally modifies the treatment effect. There are a variety of scenarios where set of effect modifiers is likely to be smaller than the set of confounders. Of particular interest are scenarios with high dimensional confounders, which tend to yield overlap issues \citep{damour_overlap_2021}, but where the set of effect modifiers potentially has moderate to low dimension.
Subject matter knowledge may indicate which covariates are unlikely to be effect modifiers and thus which covariates should be in the $g(\bm{x})$.  A data-driven procedure could also be used to select $g(\bm{x})$ by estimating how much each covariate influences treatment effect heterogeneity; in fact, we propose an algorithm and corresponding model-based estimator to do this Section \ref{sec:implementation}. 

In other scenarios, either 1) all covariates of interest may be effect modifiers or 2) the size of $g(\bm{x})$ required to ensure a solution to Problem \eqref{eq:our_db} may be large such that some effect modifiers need to be included in the set. From a design-based perspective, $g(\bm{x})$ is the set of covariates that may differ from the original sample population; thus in these cases, it is desirable to choose $g(\bm{x})$ such that the resulting weighted population is still meaningful within the given scientific context. 
In these cases, $g(\bm{x})$ could be chosen as the set of covariates for which their population is of less interest or relevance to the original research question. However, when there are high dimensional covariates in may be challenging to determine which subset yields a solution to Problem \eqref{eq:our_db} \textit{and} whose corresponding population may be modified without substantially altering the original research question. 
As described in Section \ref{sec:positivity_considerations}, the set $g(\bm{x})$ needs to be predictive of treatment assignment $Z$ in order to substantially relax the positivity assumption. As such, there is some level of predictive capacity, $\eta$, that the covariate functions in $g(\bm{x})$ need to jointly satisfy in order for Problem \eqref{eq:our_db} to yield a solution; thus, there are many possible $g(\bm{x})$ sets from which to select.
However, selecting the covariate functions that are individually the most predictive of treatment will tend to result in the smallest of such sets that have predictive capacity $\eta$. While selecting the smallest $g(\bm{x})$ set that yields a solution to Problem \eqref{eq:our_db} may not minimize the design-based component of estimator error, it tends to both 1) simplify the interpretation of the resulting modified population; and 2) be more computationally efficient to operationalize. Thus, in Section \ref{sec:implementation}, we proposed a design-based algorithm for identifying this set. 

Furthermore, both model- and design-based perspectives can be used together to select a $g(\bm{x})$ set that best matches the original research objective (e.g., $g(\bm{x})$ contains a mix of covariates that are either unlikely to be effect modifiers or highly predictive of $Z$). Therefore, while a lack of overlap forces researchers to confront trade-offs between estimator variance, statistical bias, and bias due population modification, we propose a variety of intuitive arguments and tools such that researchers can incorporate their subject-matter knowledge and research priorities when navigating these tradeoffs within our proposed optimization procedure.

\subsection{More details on the proposed model-based estimator}
\label{sec:extra_mb_estimator_supp}
In Section \ref{sec:implementation}, we propose a model-based estimator. To preserve downstream inference we use the following cross-fit style estimator,
\begin{align}
    \hat{\tau}^{CF}_{w^*} &= \frac{|\mathcal{I}_1|}{n}\hat{\tau}^{\mathcal{I}_1} + \frac{|\mathcal{I}_2|}{n}\hat{\tau}^{\mathcal{I}_2}, \:\:\: 
    \hat{\tau}^{\mathcal{I}_1} = \frac{1}{|\mathcal{I}_1|}\sum_{i \in \mathcal{I}_1} \{\hat{w}_i^{*\mathcal{I}_2}Z_iY_i - \hat{w}_i^{*\mathcal{I}_2}(1-Z_i)Y_i\}
\end{align}
where $\mathcal{I}_1$ and $\mathcal{I}_2$ each contain half of the data, randomly split, $\hat{w}_i^{*\mathcal{I}_2}$ are weights computed with $\mathcal{I}_2$ and Algorithm \ref{alg:adaptive} with a metric of treatment effect modification, and $\hat{\tau}^{\mathcal{I}_2}$ is defined in the same manner by swapping $\mathcal{I}_1$ and $\mathcal{I}_2$. We implement the absolute value of the one-step doubly robust estimator of treatment effect modification,
\begin{align*}
       TEM_j &= \left|\frac{1}{\sum_i X_{ij}^2} \sum_{i=1}^n X_{ij}\left\{(Y_i - \hat{\mu}_{Z_i}(\bm{X}_i))\frac{2Z_i -1}{Z_i\hat{e}(\bm{X}_i) + (1-Z_i)(1-\hat{e}(\bm{X}_i))} \right.\right.
       \\
       &\quad\quad\quad\quad\quad\quad\quad\quad+ \left.\left. \vphantom{\sum_{i=1}^n} \hat{\mu}_1(\bm{X}_i) - \hat{\mu}_0(\bm{X}_i)\right\}\right|,
\end{align*}
proposed by \cite{boileau_nonparametric_2025}, as this treatment effect modification measure resulted in estimators with the lowest bias across a majority of the simulation scenarios, compared to the \cite{hines_variable_2023} metric and the coefficients from a linear CATE model (supplementary Figure \ref{fig:mse_alg_adapt}). 
\begin{figure}
    \centering
    \includegraphics[width=0.9\linewidth]{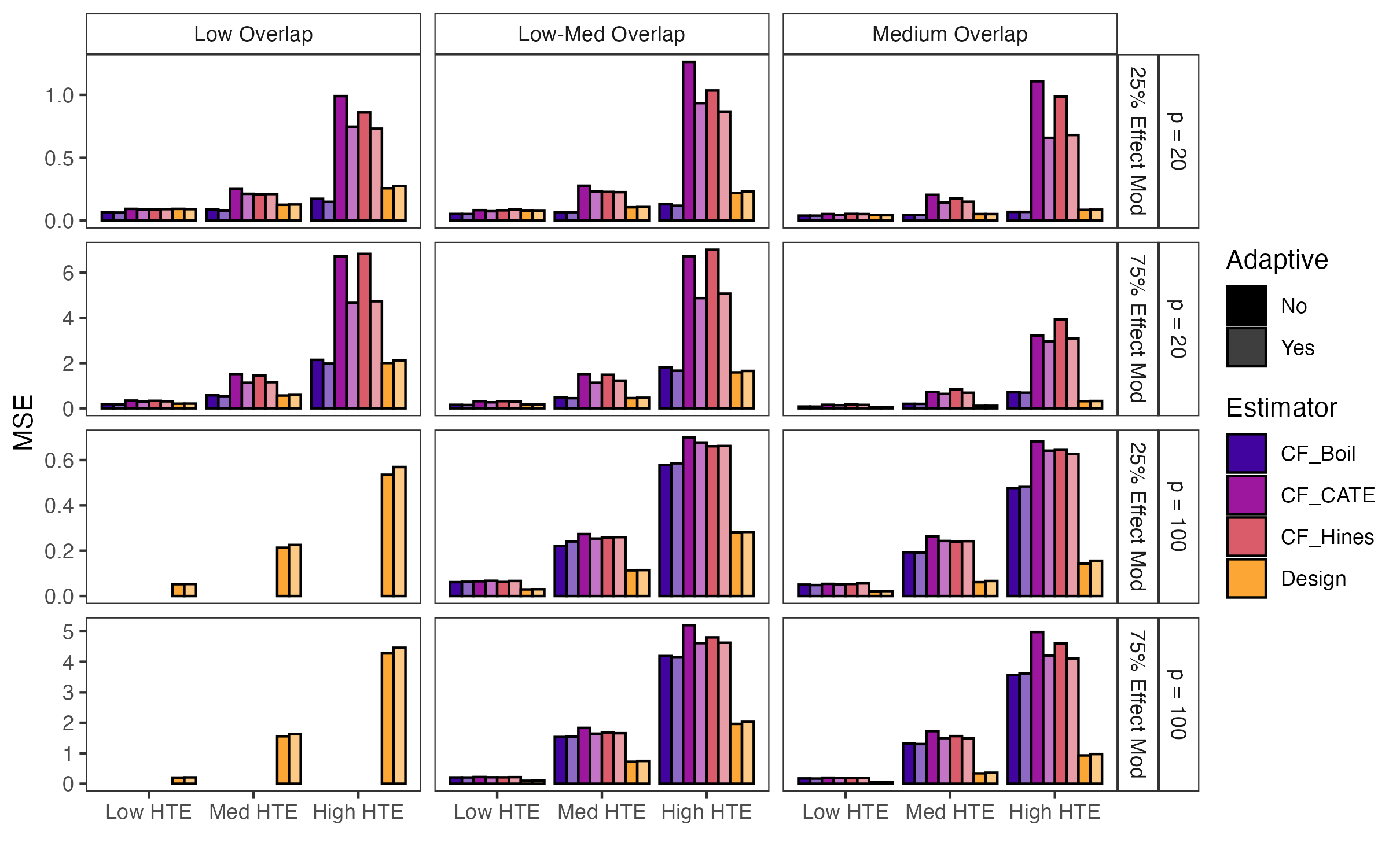}
    \caption{Comparison of estimator MSE with respect to the ATE for our proposed design-based estimator and model-based estimator (with the \cite{boileau_nonparametric_2025}, \cite{hines_variable_2023}, and linear CATE model coefficient metrics) for both supplementary Algorithm \ref{alg:simple} (not-adaptive) and Algorithm \ref{alg:adaptive} (adaptive). The scenarios are labeled on the top and right of the panels.}
    \label{fig:mse_alg_adapt}
\end{figure}
In this estimator, $\hat{e}(\bm{x}), \hat{\mu}_z(\bm{x})$ are propensity score and outcome model estimates.
Since propensity score estimates are likely to be extreme in scenarios with a lack of overlap we also explored $TEM^*_j = \text{abs}[\frac{1}{\sum_i X_{ij}^2} \sum_{i=1}^n X_{ij}\left\{\hat{\mu}_1(\bm{X}_i) - \hat{\mu}_0(\bm{X_i})\right\}]$ which only requires outcome model estimates. In simulations, estimators corresponding to $TEM_j$ tended to perform similarly or better to estimators corresponding to $TEM^*_j$ (supplementary Figure \ref{fig:mse_nops}) such that we recommend and implement $TEM_j$ in applications of our methods. 

\begin{figure}
    \centering
    \includegraphics[width=0.9\linewidth]{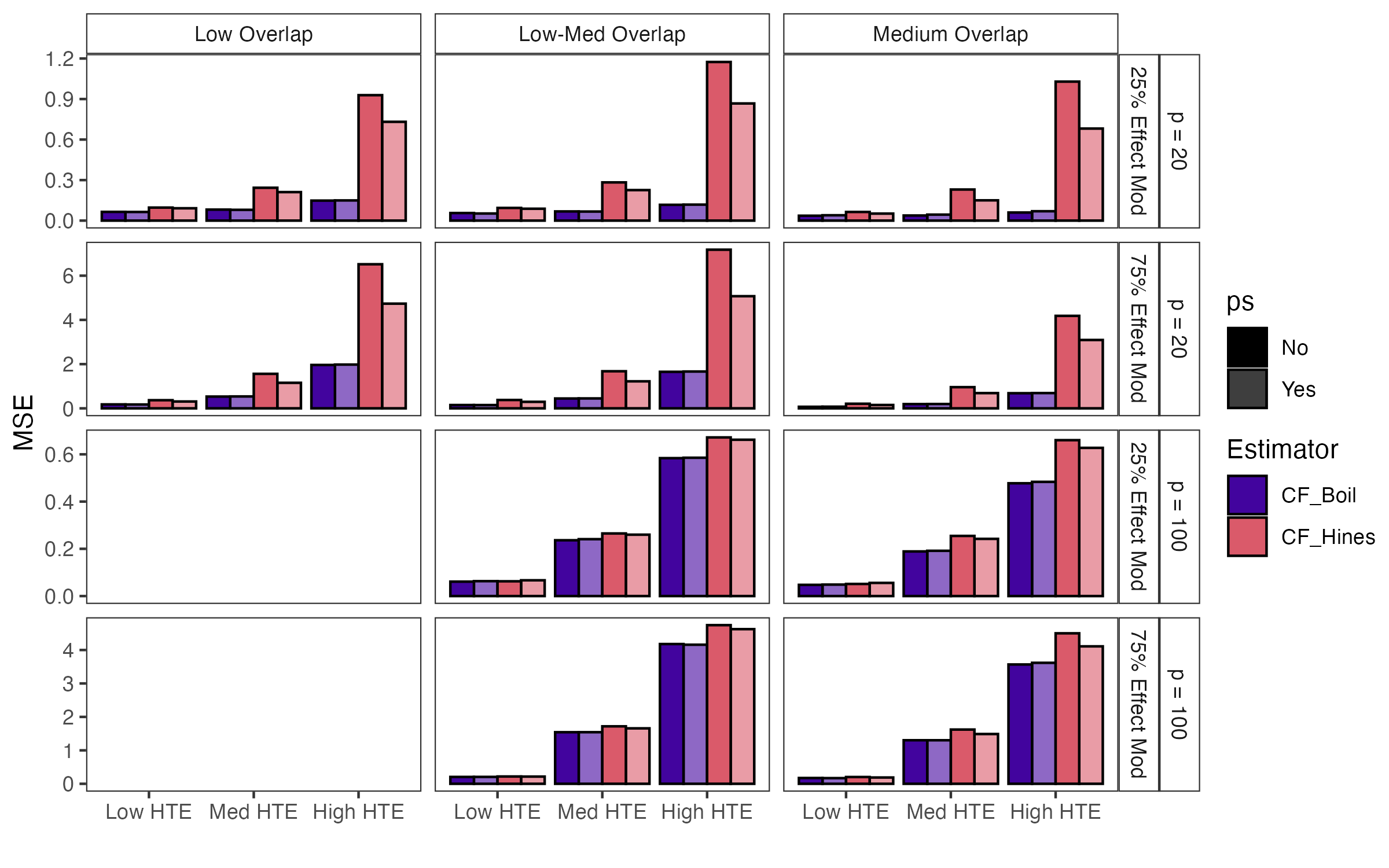}
    \caption{Comparison of estimator MSE for our model-based estimator with the \cite{boileau_nonparametric_2025} and \cite{hines_variable_2023} metrics calculated with and without estimated propensity scores. The scenarios are labeled on the top and right of the panels.}
    \label{fig:mse_nops}
\end{figure}

\section{Additional simulation methods}
\label{sec:sim_supp}

For the data generation propensity score model, $\alpha_j = 0 \pm a$ for $a = 0.5/\gamma, 1.33/\gamma, 2.16/\gamma, 3/\gamma$ for $j > 0$. Here, $a$ determines the strength of relationship between a given covariate and treatment assignment where there are an equal number of coefficients generated with each of $a = 0.5/\gamma, 1.33/\gamma, 2.16/\gamma, 3/\gamma$. As $\gamma$ increases, propensity scores become more extreme and overlap decreases. We set $\mu_0(\bm{X}) = 0.25X_1 + \cdots + 0.25X_p$ and $\mu_1(\bm{X}) = \sum_{i=0}^p \beta_i X_i$. We consider $\theta = 0.25, 0.75$ as the proportion of covariates are effect modifiers. We set $\beta_0 = 1$ and $\beta_j  = 0.25  \pm c$ for  $c = 0.75\delta,  0.6\delta, 0.45\delta, 0.3\delta, 0$. Here, $c$ determines the strength of relationship between a given covariate and $\tau(\bm{X})$; for $(1-\theta)p$ coefficients $\beta_j = 0$ and there are an approximately equal number of remaining coefficients set as each of $c = 0.75\delta,  0.6\delta, 0.45\delta, 0.3\delta$. Then, as $\delta$ increases, treatment effect heterogeneity increases for a given $\theta$. Therefore, we generate 20 different potential types of covariates for each scenario (i.e., all combinations of $\alpha_j$ and $\beta_j$ values). See supplementary Table \ref{tab:sim_params} for the $\gamma, \delta$ values used for each scenario.

\begin{table}[]
\caption{Overlap and treatment heterogeneity hyperparmeters used for data generation for each simulation scenario as described in Section \ref{sec:sim} and supplementary Section \ref{sec:sim_supp}.}
\centering
\begin{tabular}{l|llllll}
                    \hline  &      & $\bm{\gamma}$   &     &      & $\bm{\delta}$  &      \\
                      & Low  & Low-Med & Med & Low  & Medium & High  \\ \hline
20\% Treated, p = 20  & 0.75 & 1       & 2   & 0.50 & 1      & 2    \\
40\% Treated, p = 20  & 0.50 & 0.75    & 1   & 0.50 & 1      & 2    \\
20\% Treated, p = 100 & 2    & 4       & 5   & 0.25 & 0.75   & 1.25 \\
20\% Treated, p = 100 & 1    & 2       & 3   & 0.25 & 0.75   & 1.25 \\ \hline
\end{tabular}
\label{tab:sim_params}
\end{table}

\begin{algorithm}
	\caption{Identifying an approximately optimal $g(\bm{X})$ set given metric $m$} \label{alg:simple}
     \hspace*{\algorithmicindent}
	\begin{algorithmic}[1]
            \State \textbf{Inputs:} The matrix $\mathbf{X}_{n\times p}$ and set of covariate functions ordered according to some metric, $m$, $b_{(1)}(\bm{X}), \ldots, b_{(J)}(\bm{X})$
            \State \textbf{Initialize:} $b(\bm{X})^{-} = \emptyset $, $j = 0$
		\While {$\{w^*_i\}_{i=1}^n$ given $g(\bm{X}) = b(\bm{X})^{-}$ does not exist}
                \State Let $j = j + 1$
                \State Let $b(\bm{X})^{-} = \{b(\bm{X})^{-},  b_{(j)}(\bm{X}) \}$
			\State Solve Problem \eqref{eq:our_db} for $g(\bm{X}) = b(\bm{X})^{-}$ and $c(\bm{X}) = \{b_{(1)}(\bm{X}), \ldots, b_{(J)}(\bm{X})\} \setminus b(\bm{X})^{-}$
		\EndWhile
            \State \textbf{Outputs:} 
            \State \hspace{1.5em} $b(\bm{X})^{-}$ \Comment{Smallest set that yields a solution to Problem \eqref{eq:our_db} when $g(\bm{X}) = b(\bm{X})^{-}$} 
            \State \hspace{1.5em} $\{w^*_i\}_{i=1}^n$ \Comment{Resulting weights from \eqref{eq:our_db} with $g(\bm{X}) = b(\bm{X})^{-}$}
	\end{algorithmic} 
\end{algorithm}

\begin{figure}
    \centering
    \includegraphics[width=0.9\linewidth]{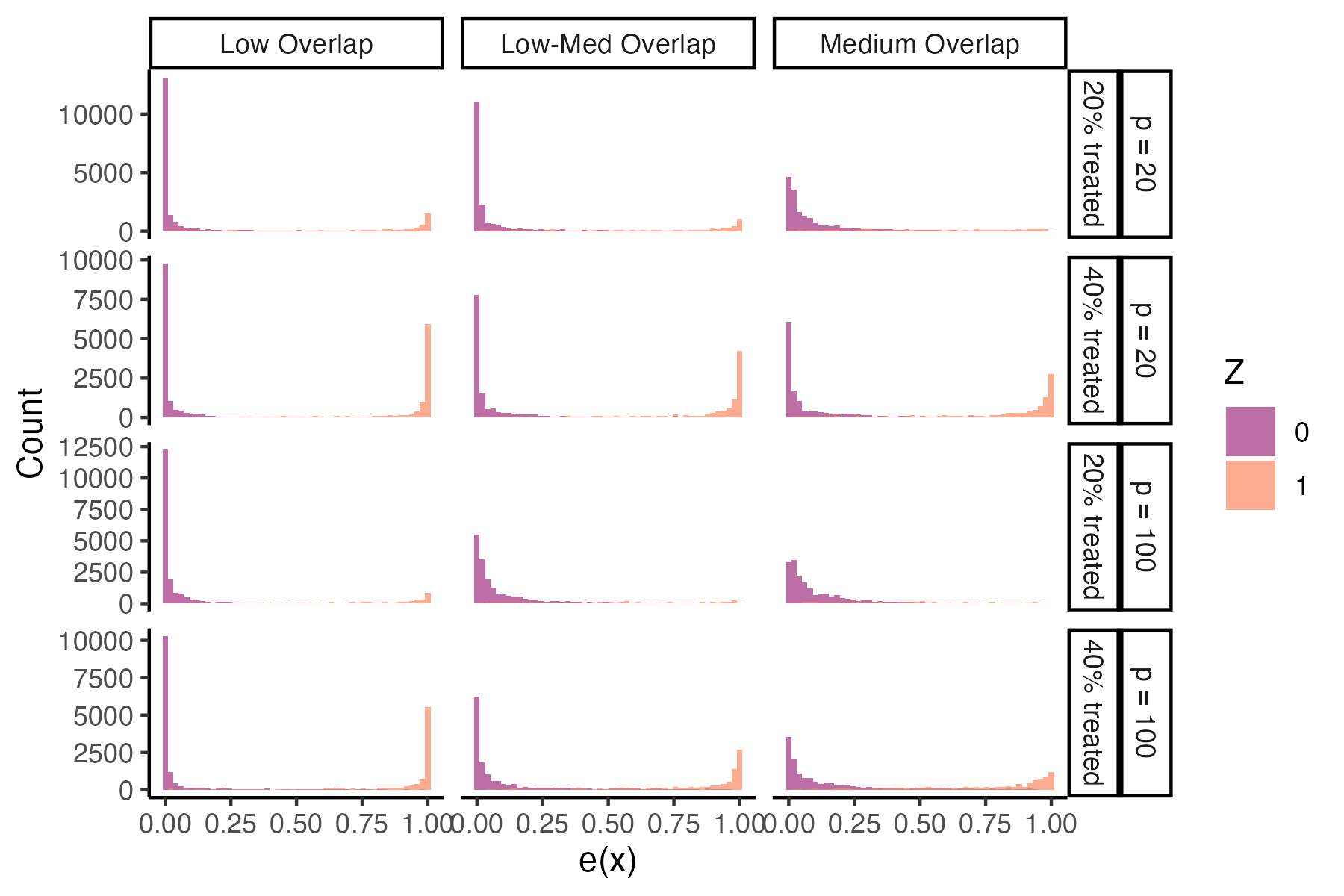}
    \caption{Propensity score distribution for the simulation scenarios described in Section \ref{sec:sim}. The scenarios are labeled on the top and right of the panels.}
    \label{fig:sim_ps_dist}
\end{figure}

\begin{figure}
    \centering
    \includegraphics[width=0.9\linewidth]{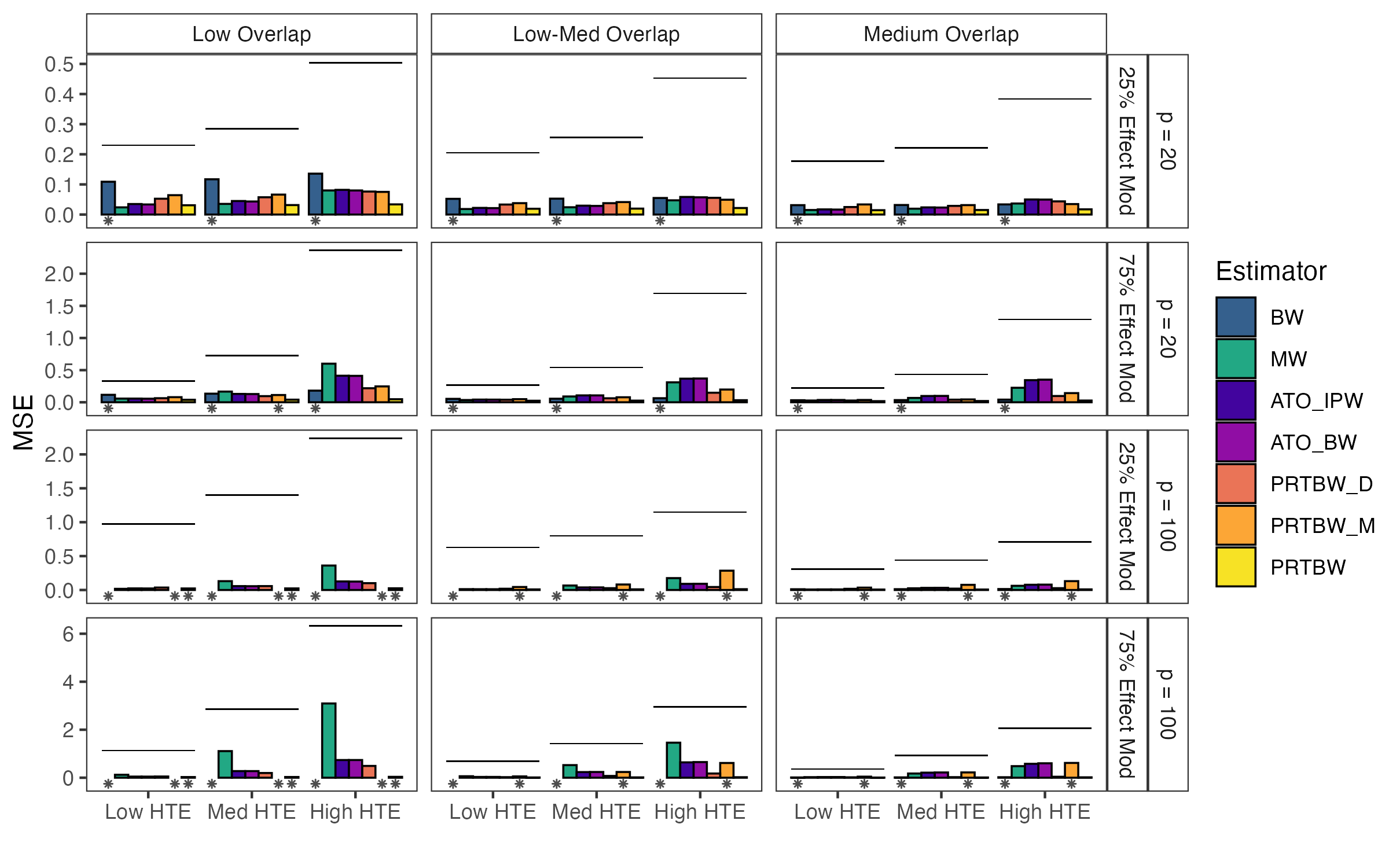}
    \caption{Estimator mean squared error (MSE) for all 40\% treated simulation scenarios. The line indicates the MSE for the IPW ATE estimator for each scenario. The asterisk, $*$, indicates estimators where a solution to the balancing weights optimization problem did not exist for all datasets. The horizontal axis indicates levels of treatment effect heterogeneity while the simulation scenarios are labeled on the top and right of the panels.}
    \label{fig:mse_sim_supp}
\end{figure}


\begin{figure}
    \centering
    \includegraphics[width=0.9\linewidth]{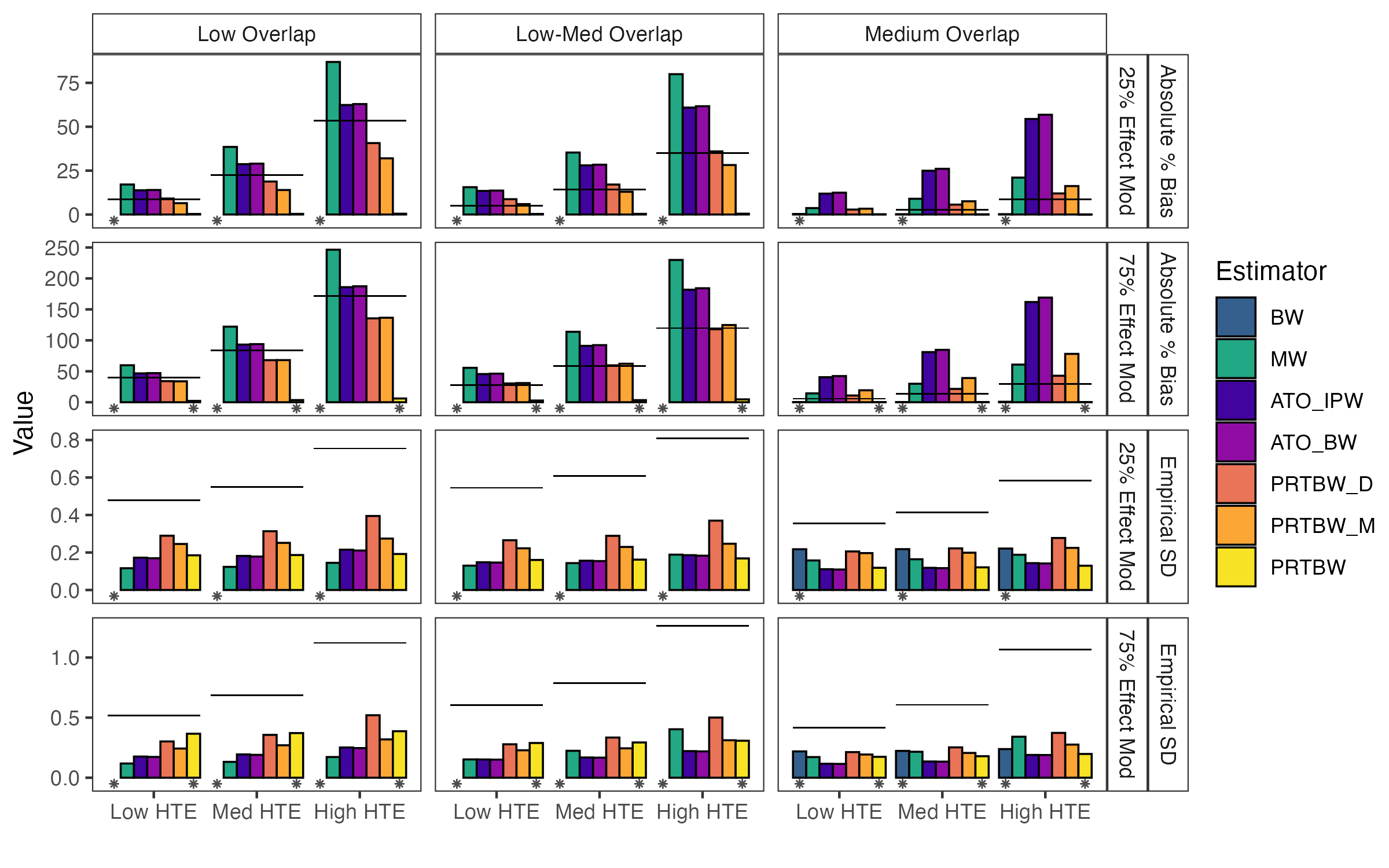}
    \caption{Estimator absolute percent bias with respect to the ATE  and empirical standard deviation (SD) for 20\% treated and $p=20$ scenarios. The line indicates the MSE for the IPW ATE estimator for each scenario. The asterisk, $*$, indicates estimators where a solution to the balancing weights optimization problem did not exist for all datasets. The horizontal axis indicates levels of treatment effect heterogeneity while the simulation scenarios are labeled on the top and right of the panels.}
    \label{fig:p20_sim}
\end{figure}

\begin{figure}
    \centering
    \includegraphics[width=0.9\linewidth]{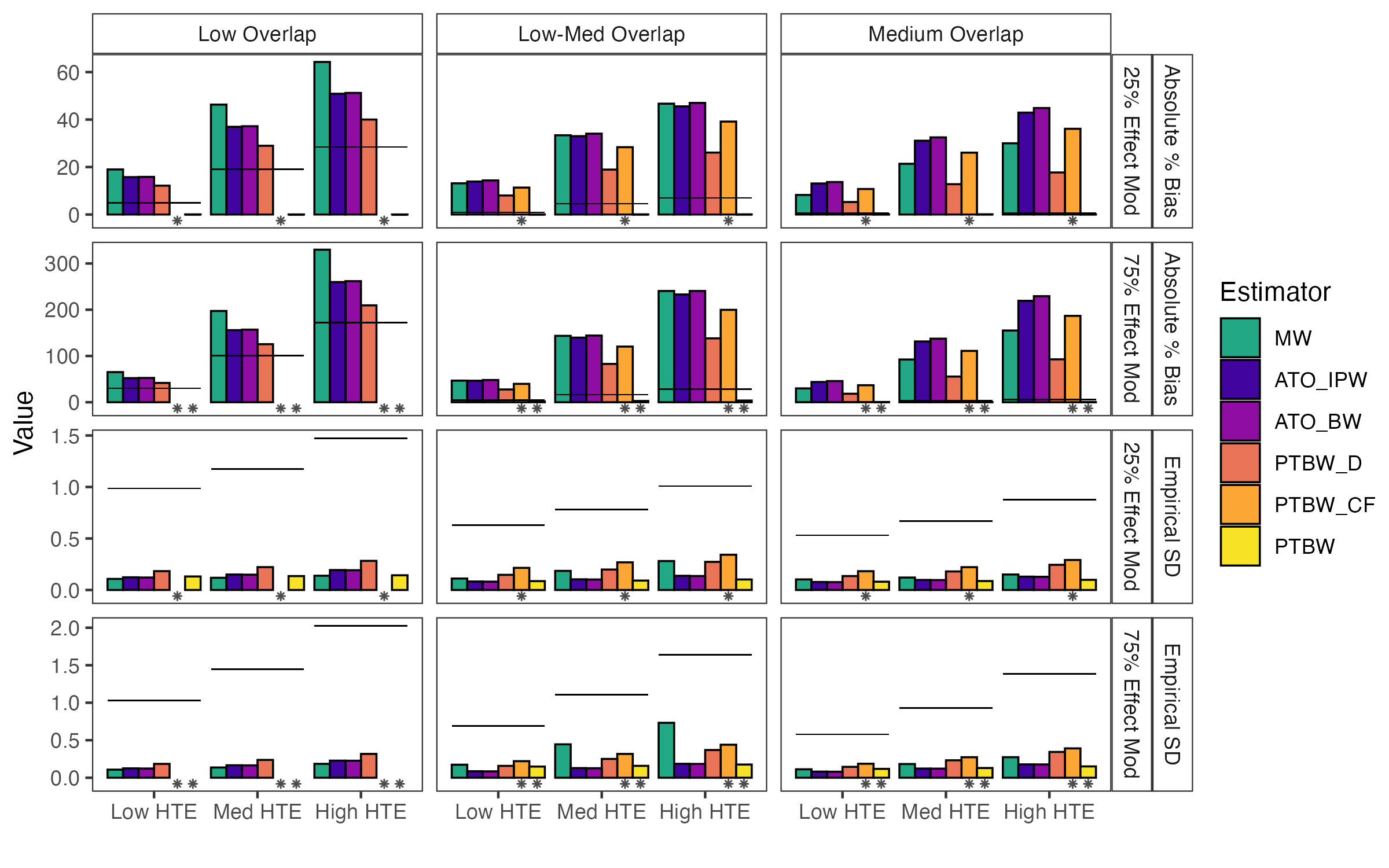}
    \caption{Estimator absolute percent bias with respect to the ATE  and empirical standard deviation (SD) for 20\% treated and $p=100$ scenarios. The line indicates the MSE for the IPW ATE estimator for each scenario. The asterisk, $*$, indicates estimators where a solution to the balancing weights optimization problem did not exist for all datasets. The horizontal axis indicates levels of treatment effect heterogeneity while the simulation scenarios are labeled on the top and right of the panels.}
    \label{fig:p100_sim}
\end{figure}

\begin{figure}
    \centering
    \includegraphics[width=0.9\linewidth]{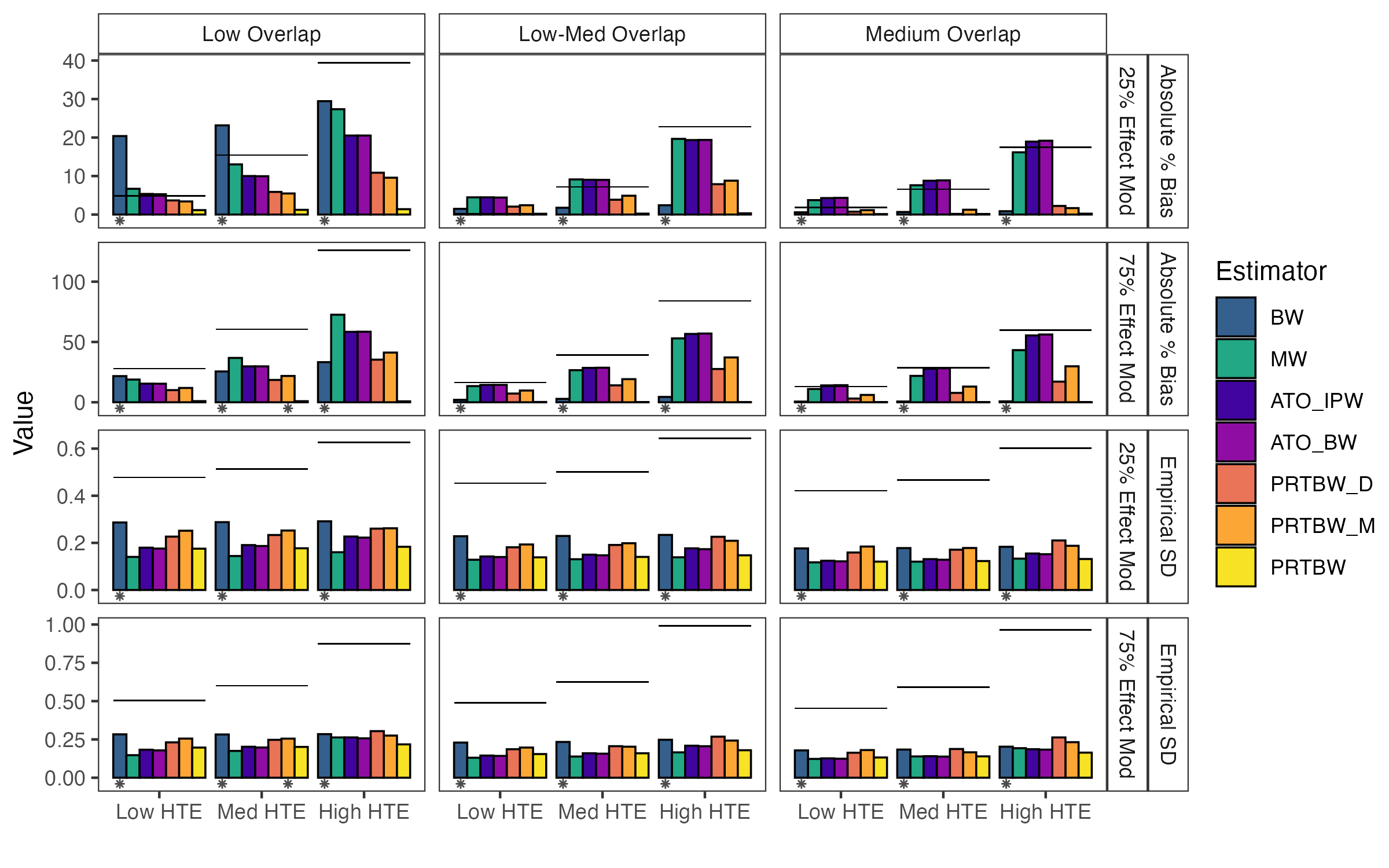}
    \caption{Estimator absolute percent bias with respect to the ATE  and empirical standard deviation (SD) for 40\% treated and $p=20$ scenarios. The line indicates the MSE for the IPW ATE estimator for each scenario. The asterisk, $*$, indicates estimators where a solution to the balancing weights optimization problem did not exist for all datasets. The horizontal axis indicates levels of treatment effect heterogeneity while the simulation scenarios are labeled on the top and right of the panels.}
    \label{fig:p20_sim_supp}
\end{figure}

\begin{figure}
    \centering
    \includegraphics[width=0.9\linewidth]{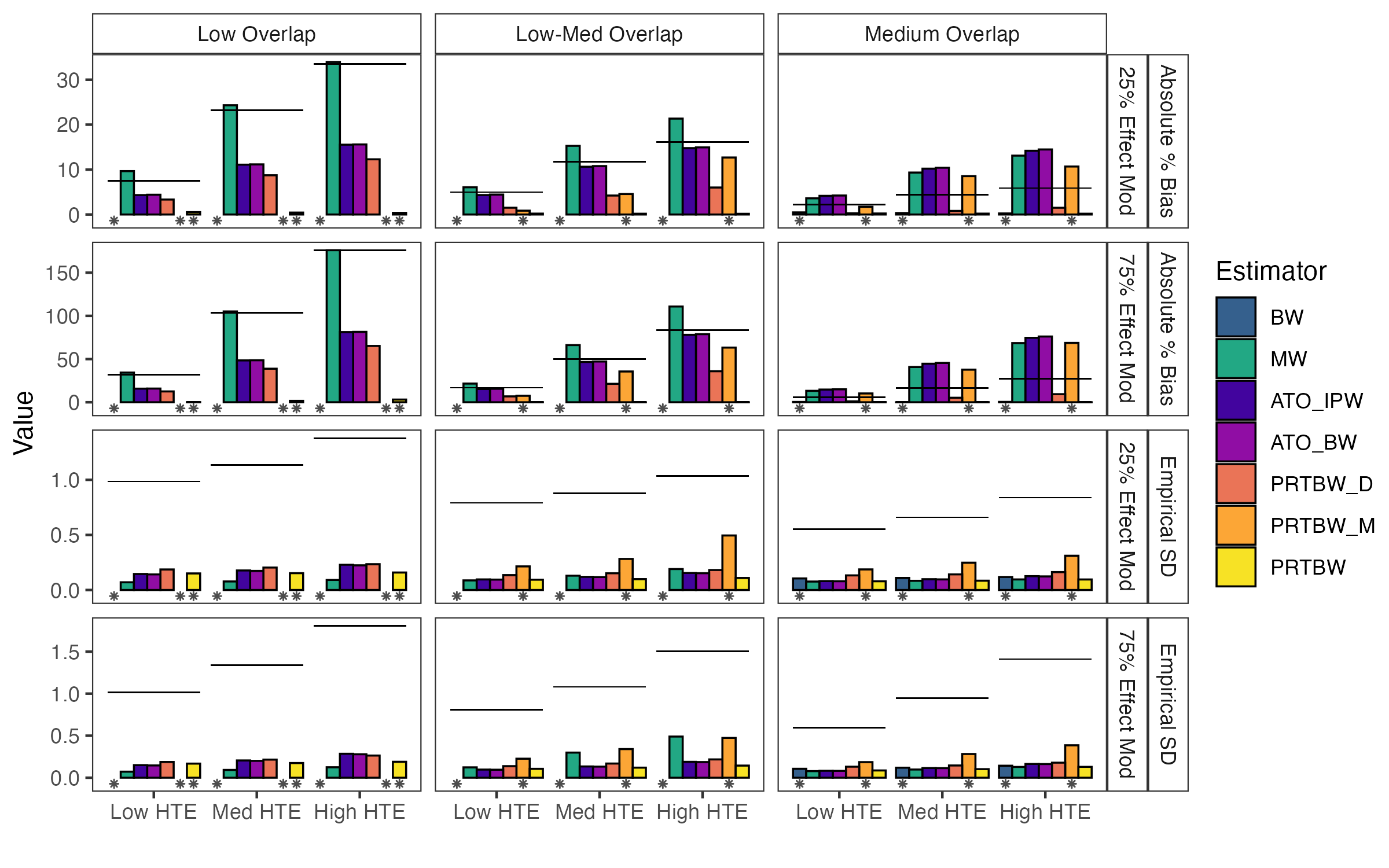}
    \caption{Estimator absolute percent bias with respect to the ATE  and empirical standard deviation (SD) for 40\% treated and $p=100$ scenarios. The line indicates the MSE for the IPW ATE estimator for each scenario. The asterisk, $*$, indicates estimators where a solution to the balancing weights optimization problem did not exist for all datasets. The horizontal axis indicates levels of treatment effect heterogeneity while the simulation scenarios are labeled on the top and right of the panels.}
    \label{fig:p100_sim_supp}
\end{figure}

\begin{figure}[h!]
    \centering
    \includegraphics[width=0.7\linewidth]{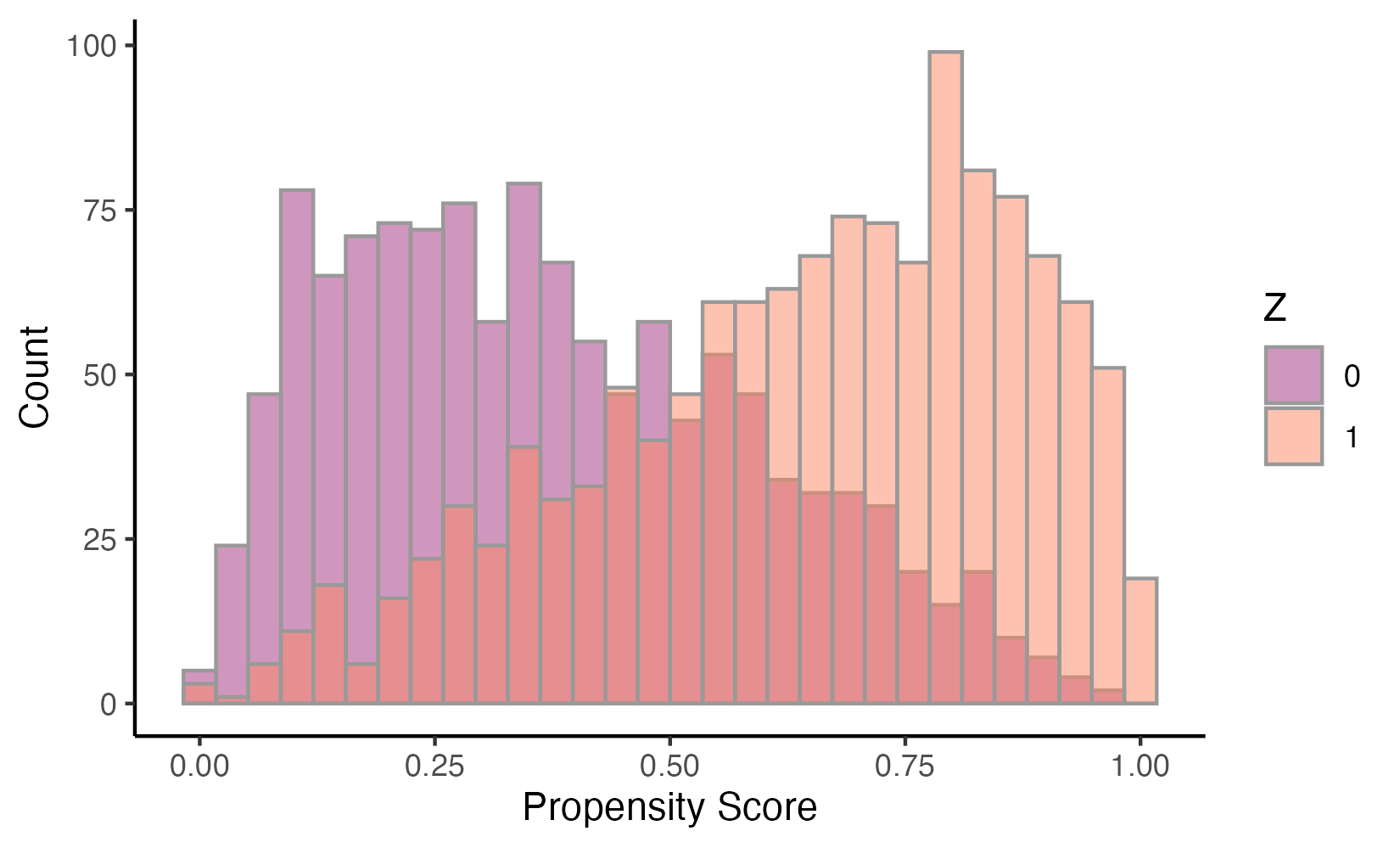}
    \caption{Distribution of estimated propensity scores for the study on indwelling arterial catheters (IAC) with data from the MIMIC-III database (Section \ref{sec:real_data_mimic})}
    \label{fig:MIMIC_ps}
\end{figure}

\begin{figure}[h!]
    \centering
    \includegraphics[width=1\linewidth]{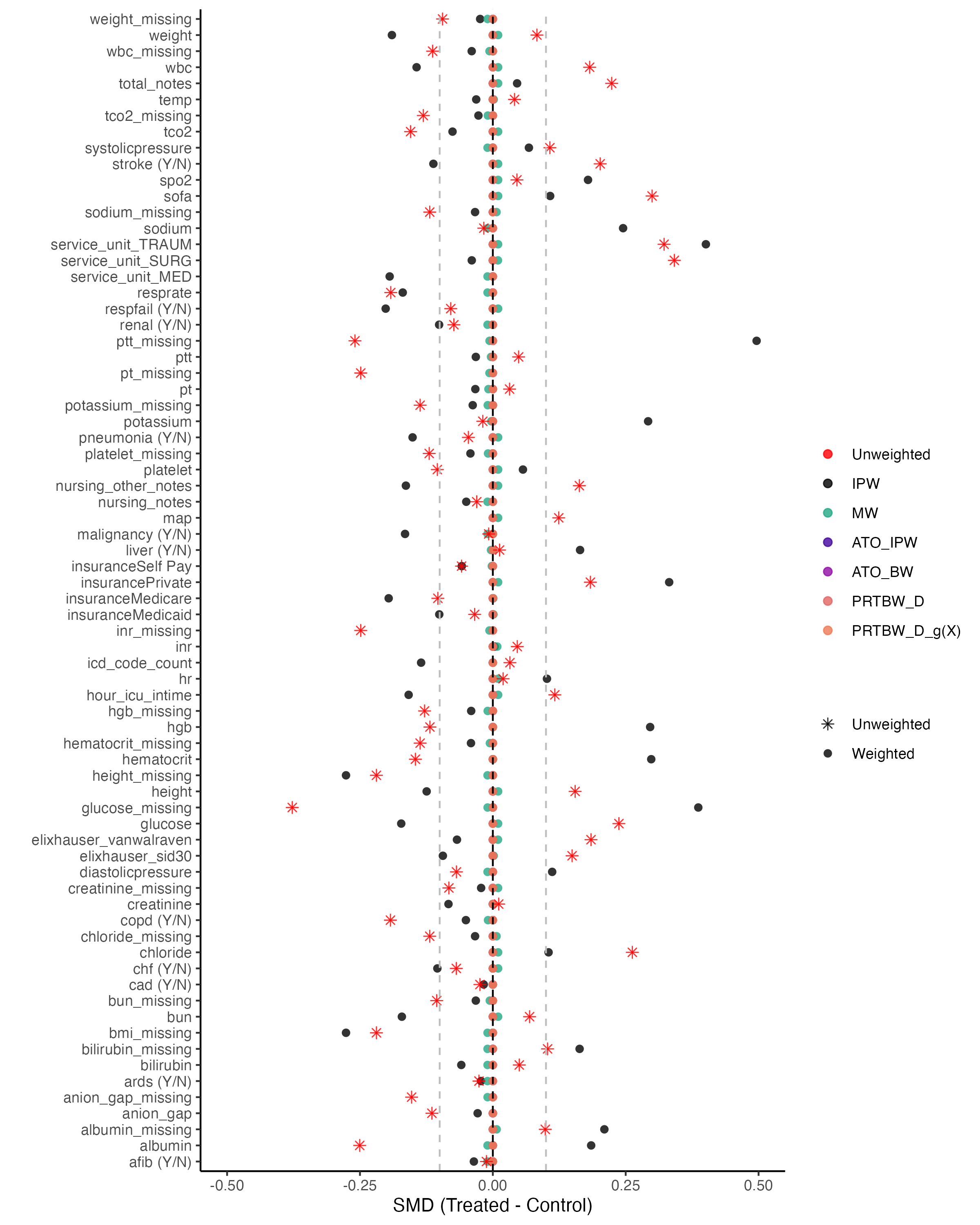}
    \caption{Treated and control group standardized mean differences (SMD) of all covariates for the study on indwelling arterial catheters (IAC) with data from the MIMIC-III database (Section \ref{sec:real_data_mimic}).}
    \label{fig:MIMIC_tc}
\end{figure}

\begin{figure}[h!]
    \centering
    \includegraphics[width=1\linewidth]{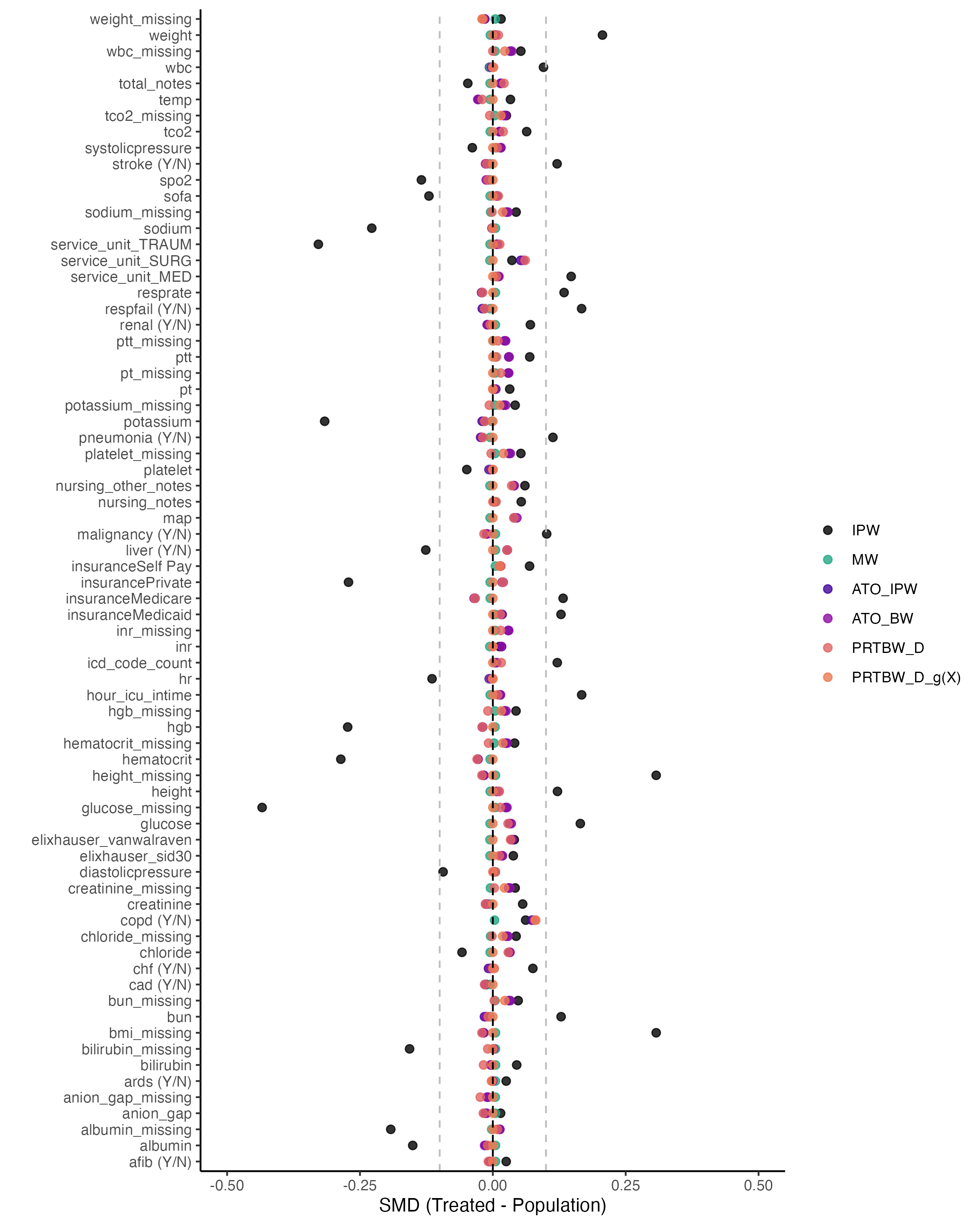}
    \caption{Treated and sample population standardized mean differences (SMD) of all covariates for the study on indwelling arterial catheters (IAC) with data from the MIMIC-III database (Section \ref{sec:real_data_mimic}).}
    \label{fig:MIMIC_tp}
\end{figure}

\begin{figure}[h!]
    \centering
    \includegraphics[width=1\linewidth]{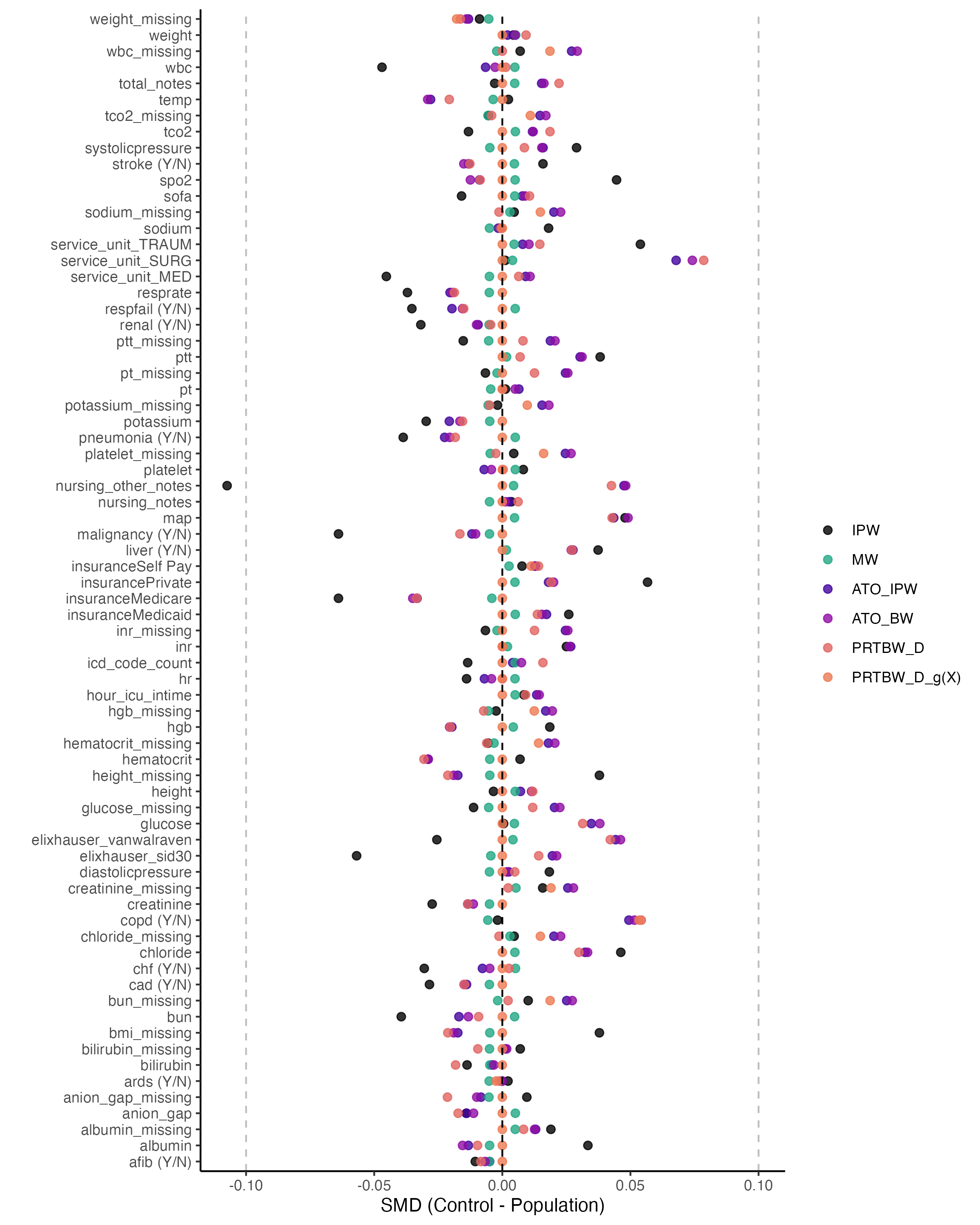}
    \caption{Control and sample population standardized mean differences (SMD) of all covariates for the study on indwelling arterial catheters (IAC) with data from the MIMIC-III database (Section \ref{sec:real_data_mimic}).}
    \label{fig:MIMIC_cp}
\end{figure}

\begin{figure}[h!]
    \centering
    \includegraphics[width=0.9\linewidth]{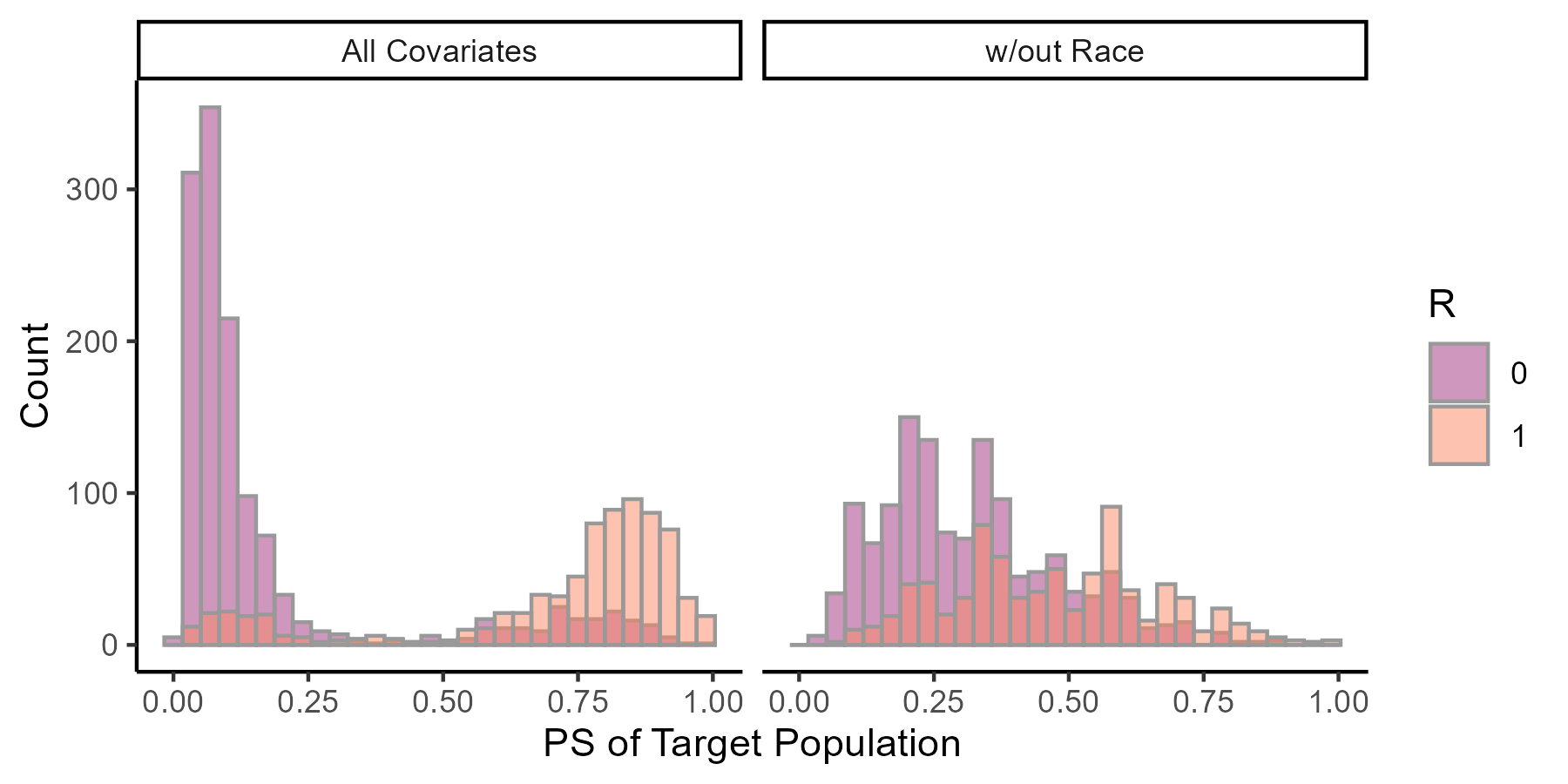}
    \caption{Distribution of the estimated probability of being in the target population when transporting the ``health care hotspotting'' RCT effect to a Midwestern academic health center population (Section \ref{sec:real_data_hotspot})}
    \label{fig:hotspot_ps}
\end{figure}

\begin{figure}[h!]
    \centering
    \includegraphics[width=1\linewidth]{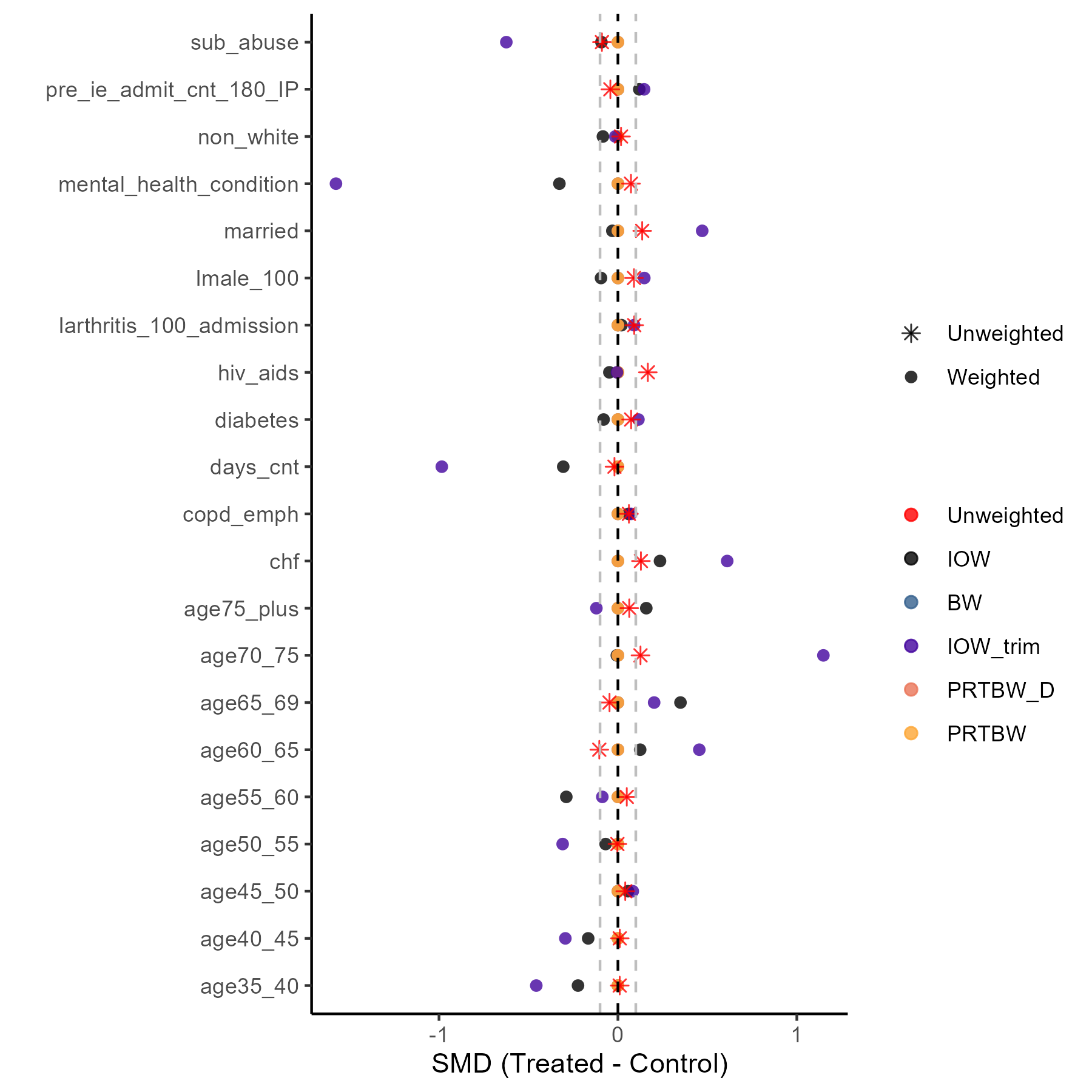}
    \caption{Treated and control group standardized mean differences (SMD) of all covariates when transporting the ``health care hotspotting'' RCT effect to a Midwestern academic health center population (Section \ref{sec:real_data_hotspot}).}
    \label{fig:hotspsot_tc}
\end{figure}

\begin{figure}[h!]
    \centering
    \includegraphics[width=1\linewidth]{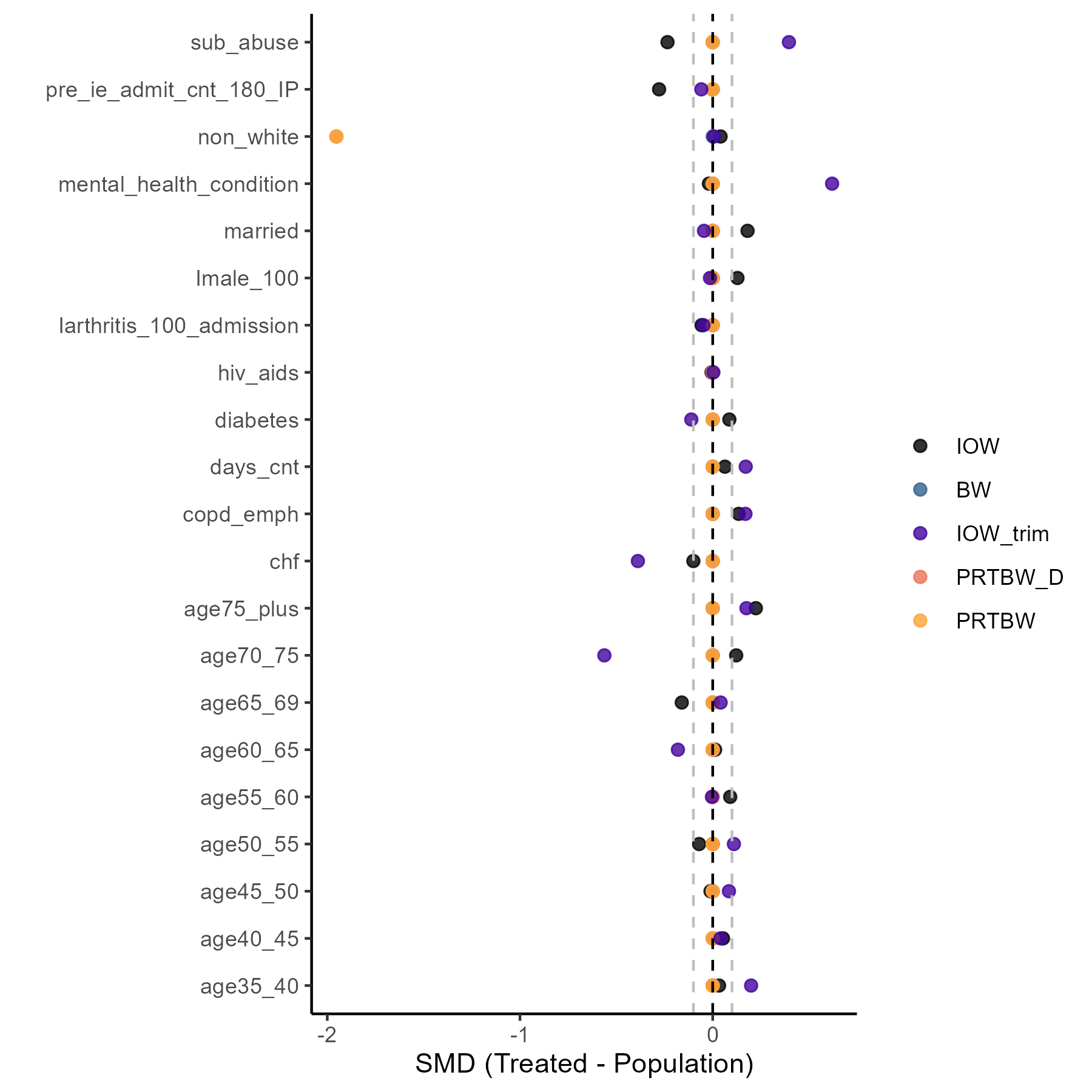}
    \caption{Treated and target population standardized mean differences (SMD) of all covariates when transporting the ``health care hotspotting'' RCT effect to a Midwestern academic health center population (Section \ref{sec:real_data_hotspot}).}
    \label{fig:hotspot_tp}
\end{figure}

\begin{figure}[h!]
    \centering
    \includegraphics[width=1\linewidth]{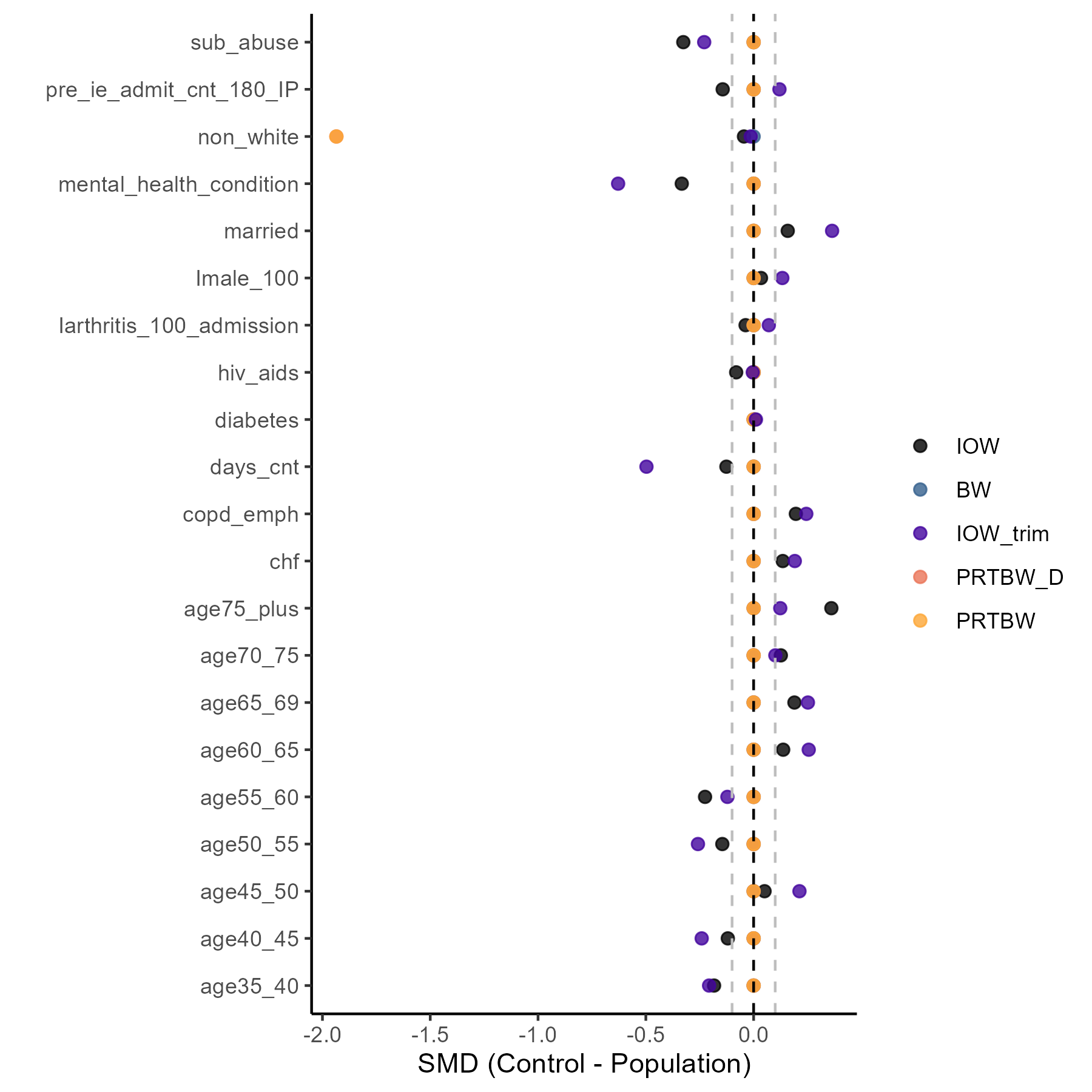}
    \caption{Control and target population standardized mean differences (SMD) of all covariates when transporting the ``health care hotspotting'' RCT effect to a Midwestern academic health center population (Section \ref{sec:real_data_hotspot}).}
    \label{fig:hotspot_cp}
\end{figure}

\clearpage

\putbib
\end{bibunit}

\end{document}